\documentclass[sigconf, nonacm]{acmart}

\newcommand\vldbdoi{XX.XX/XXX.XX}

\newcommand\vldbavailabilityurl{URL_TO_YOUR_ARTIFACTS}
\newcommand\vldbpagestyle{plain} 

\usepackage{amsmath}

\usepackage[noEnd=true, indLines=true]{algpseudocodex}
\usepackage{algorithm, bm, float, bigstrut, subcaption}
\algnewcommand\algorithmicinput{\textbf{Input:}}
\algnewcommand\Input{\item[\algorithmicinput]}
\algnewcommand\algorithmicout{\textbf{Output:}}
\algnewcommand\Out{\item[\algorithmicout]}

\usepackage{color, soul}
\usepackage{cleveref}
\newcommand{\todo}[1]{\textcolor{red}{\hl{[#1]}}}

\newtheorem{property}{Property}

\newcommand{\nop}[1]{}

\newcommand{\Dom}{\ensuremath{\operatorname{\mathsf{Dom}}}}
\newcommand{\cov}{\ensuremath{\operatorname{\mathsf{cov}}}}

\newcommand{\dotprec}{\ensuremath{\mathbin{\dot\prec}}}
\newcommand{\dotsucc}{\ensuremath{\mathbin{\dot\succ}}}
\newcommand{\dotsim}{\ensuremath{\mathbin{\dot\sim}}}
\newcommand{\sign}{\ensuremath{\operatorname{\mathsf{sign}}}}
\newcommand{\Arrow}[1]{\ensuremath{\parbox{#1}{\tikz{\draw[->](0,0)--(#1,0);}}}}
\newcommand{\substitute}[2]{\ensuremath{\langle #1\,\Arrow{.3em}\,#2 \rangle}}
\newcommand{\ubound}{\ensuremath{\operatorname{\mathsf{up}}}}
\newcommand{\lbound}{\ensuremath{\operatorname{\mathsf{low}}}}

\definecolor{lightblue}{RGB}{227,242,253}
\definecolor{darkblue}{RGB}{33,150,243}
\definecolor{lightpurple}{RGB}{243,229,245}
\definecolor{darkpurple}{RGB}{156,39,176}
\definecolor{lightorange}{RGB}{255,243,224}
\definecolor{darkorange}{RGB}{255,152,0}
\definecolor{lightgray}{RGB}{248,249,250}
\definecolor{statblue}{RGB}{52,152,219}
\definecolor{statgreen}{RGB}{46,204,113}
\definecolor{statred}{RGB}{231,76,60}

\definecolor{lightgreen}{RGB}{200,255,200}
\definecolor{lightyellow}{RGB}{255,255,200}
\definecolor{lightcyan}{RGB}{200,255,255}
\definecolor{lightpink}{RGB}{255,200,255}
\begin{document}
\begin{sloppy}

\title{Finding Non-Redundant Simpson's Paradox from Multidimensional Data}

\author{Yi Yang}
\affiliation{
\institution{Duke University}
\city{Durham}
\state{NC}
\country{USA}
}
\email{owen.yang@duke.edu}

\author{Jian Pei}
\affiliation{
\institution{Duke University}
\city{Durham}
\state{NC}
\country{USA}
}
\email{j.pei@duke.edu}

\author{Jun Yang}
\affiliation{
\institution{Duke University}
\city{Durham}
\state{NC}
\country{USA}
}
\email{junyang@cs.duke.edu}

\author{Jichun Xie}
\affiliation{
\institution{Duke University}
\city{Durham}
\state{NC}
\country{USA}
}
\email{jichun.xie@duke.edu}

\begin{abstract}
Simpson's paradox, a long-standing statistical phenomenon, describes the reversal of an observed association when data are disaggregated into sub-populations. It has critical implications across statistics, epidemiology, economics, and causal inference. Existing methods for detecting Simpson's paradox overlook a key issue: many paradoxes are \emph{redundant}, 
arising from equivalent selections of data subsets, identical partitioning of sub-populations, and correlated outcome variables, which obscure essential patterns and inflate computational cost.
In this paper, we present the first framework for discovering \emph{non-redundant} Simpson's paradoxes. We formalize three types of redundancy -- sibling child, separator, and statistic equivalence -- and show that redundancy forms an equivalence relation.
Leveraging this insight, we propose a concise representation framework for systematically organizing redundant paradoxes and design efficient algorithms that integrate depth-first materialization of the base table with redundancy-aware paradox discovery.
Experiments on real-world datasets and synthetic benchmarks show that redundant paradoxes are widespread, on some real datasets constituting over 40\% of all paradoxes, while our algorithms scale to millions of records, reduce run time by up to 60\%, and discover paradoxes that are structurally robust under data perturbation. These results demonstrate that Simpson's paradoxes can be efficiently identified, concisely summarized, and meaningfully interpreted in large multidimensional datasets.
\end{abstract}

\maketitle

\pagestyle{\vldbpagestyle}

\ifdefempty{\vldbavailabilityurl}{}{%
\let\thefootnote\relax\footnote{\textbf{Artifact Availability:} The source code, data, and/or other artifacts have been made available at \url{https://github.com/Owen-Yang-18/non-redundant-simpson-paradox}.}%
}

\section{Introduction}
\label{sec:introduction}

Simpson's paradox~\cite{https://doi.org/10.1111/j.2517-6161.1951.tb00088.x, pearson1899genetic} is a classic and widely studied phenomenon in statistics, probability, and data science~\cite{pearl2014comment, samuels1993simpson, wagner1982simpson, spirtes2000causation, freitas2007integrating, alipourfard2018can, xu2018detecting, sharma2022detecting, wang2023learning, zuyderhoff2025simpson, jiang2025fedcfa, salimi2018hypdb, krishnan2016activeclean, portela2019search, lenz1997summarizability, bonchi2018probabilistic}. 
This paradox refers to the reversal of an observed association between two variables when data are disaggregated into sub-populations. 
It has been recognized for more than a century and continues to play a central role in fields such as epidemiology, economics, machine learning, and causal inference~\cite{tu2008simpson, alipourfard2018can, ma2015simpson, pearl2014comment, salazar2021automated, lin2021detecting, guo2017you, deng2024outlier, youngmann2024summarized, tang2013mining, krishnan2016activeclean}, where decisions depend critically on understanding relationships in multidimensional data. 

\begin{table}[t]
  \centering
  \caption{Data table $T(A,B,C,Y_1)$ containing 7 records.}
    \begin{tabular}{|c|c|c|c|c|}
    \hline
     & $A$         & $B$         &  $C$        & $Y_1$ \bigstrut\\
    \hline
    $t_1$  & $a_1$      & $b_1$      & $c_1$      & 0 \bigstrut[t]\\
    $t_2$ & $a_1$      & $b_1$      & $c_1$      & 0 \\
    $t_3$ &$a_1$      & $b_2$      & $c_1$      & 0 \\
    $t_4$ &$a_2$      & $b_1$      & $c_2$      & 1 \\
    $t_5$ &$a_2$      & $b_1$      & $c_2$      & 1 \\
    $t_6$ &$a_2$      & $b_2$      & $c_2$      & 1 \\
    $t_7$ &$a_2$      & $b_2$      & $c_2$      & 1 \bigstrut[b]\\
    \hline
    \end{tabular}%
  \label{tab:ex1}%
\end{table}%

\begin{example}[Simpson's paradox]
\label{ex:simpson}
Consider the dataset in \Cref{tab:ex1}. Overall, the probability of $Y_1=1$ is lower for records with $B=b_2$ than for those with $B=b_1$: 
\[
P(Y_1=1 \mid B=b_2) = \tfrac{2}{4} = 0.50 
\, < \, 
P(Y_1=1 \mid B=b_1) = \tfrac{2}{3} \approx 0.67.
\]
However, when the data are partitioned by attribute $A$, the trend reverses. For $A=a_1$, both $B=b_1$ and $B=b_2$ yield $P(Y_1=1)=0$. For $A=a_2$, both $B=b_1$ and $B=b_2$ yield $P(Y_1=1)=1$. In each subgroup defined by $A$, the conditional probabilities satisfy 
\[
P(Y_1=1 \mid A=a_i, B=b_2) \geq P(Y_1=1 \mid A=a_i, B=b_1), \, i \in \{1,2\}.
\]
Thus, although the aggregated data suggest $B=b_1$ is more favorable, conditioning on $A$ eliminates the apparent disadvantage of $B=b_2$. This reversal of association between $B$ and $Y_1$ after conditioning on $A$ is an instance of Simpson's paradox.  
\qed
\end{example}

Simpson's paradox has been observed in diverse real-world domains, including medicine and social science~\cite{blyth1972simpson, kievit2013simpson, cates2002simpson}. 
In a well-known study of treatment effectiveness for kidney stones~\cite{charig1986comparison}, the overall recovery rate appears higher for one treatment, but when patients are stratified by stone size, the trend reverses in both subgroups. 
A similar paradox was documented in graduate admissions at the University of California, Berkeley~\cite{bickel1975sex}, where aggregate data suggested gender bias, yet department-level data showed the opposite pattern. 
These counterintuitive reversals---the essence of Simpson's paradox---demonstrate how aggregated data can obscure underlying relationships and highlight the importance of identifying such paradoxes for reliable analysis and decision-making.

Despite its importance, an overlooked issue in the literature is that instances of Simpson's paradox can be highly \emph{redundant}. 
In high-dimensional data, many partitions share identical sets of records in the base table, or different choices of separator or label attributes may yield equivalent partitions. 
As a result, multiple paradoxes can describe the same underlying phenomenon. 
For example, in \Cref{tab:ex1}, there is a one-to-one correspondence between attributes $A$ and $C$: every Simpson's paradox involving $A$ can also be expressed as one involving $C$. 
Although such paradoxes differ syntactically, they arise from the same overlapping population structure. 
Treating them as distinct not only inflates the number of reported paradoxes but also obscures the essential insights that analysts aim to extract.

One might question whether redundant paradoxes occur merely in theory or isolated cases. However, our empirical analysis using real-world datasets in different domains, reported in Section~\ref{sec:rq1}, shows that redundant Simpson’s paradoxes account for 20.3–47.8\% of all observed paradoxes.

Identifying \emph{non-redundant} Simpson's paradoxes poses several technical challenges. 
First, the search space of potential paradoxes grows exponentially with the number of attributes, making brute-force enumeration computation-heavy. 
Second, redundancies can arise in multiple ways, as analyzed in \Cref{sec:equivalence-types}, and distinguishing among them requires careful formalization. 
Third, even after redundancies are recognized, a principled method is needed to group redundant paradoxes and produce concise, non-overlapping representations without information loss. 


To address these challenges, our key idea is to exploit the mathematical structure underlying how data subsets (populations) relate to one another. We discover that redundant paradoxes exhibit patterns that allow us to group them into well-defined equivalence classes. We propose a concise representation for these equivalence classes that eliminates redundancy while ensuring completeness (i.e., discovery of all Simpson's paradoxes). This approach enables us to compactly capture large numbers of redundant paradoxes.


We make four main contributions in this paper. First, we formally define three sources of redundancy -- sibling child, separator, and statistic equivalence -- and show that redundancy is an equivalence relation. Second, we propose a concise representation framework that groups redundant paradoxes into compact, systematic summaries. Third, we develop efficient algorithms that combine depth-first materialization of the input base table and redundancy-aware discovery of Simpson's paradoxes. Last, through experiments on both real-world and synthetic datasets, we demonstrate that redundant paradoxes are common in practice, our methods scale efficiently, and the identified paradoxes (and redundancies) are structurally robust.

The rest of the paper is organized as follows. 
\Cref{sec:prelim} introduces preliminaries and definitions of Simpson's paradox. 
\Cref{sec:coverage} formalizes redundancy and presents our concise representation framework. 
\Cref{sec:finding} describes algorithms for discovering non-redundant paradoxes. 
\Cref{sec:experiments} reports experimental results on real and synthetic datasets. 
\Cref{sec:related} reviews related work, and \Cref{sec:conclusion} concludes the paper.



\section{Preliminaries}
\label{sec:prelim}

We introduce the foundational concepts used throughout the paper. We then present the formal definition of Simpson's Paradox in this context, along with illustrative examples and related variants. 

\subsection{Basic Notations}
\label{sec:prelim-notation}

Consider a base table $T$ containing $n$ \emph{categorical attributes} $\{X_1, \ldots, X_n\}$ and $m$ \emph{label attributes} $\{Y_1, \ldots, Y_m\}$,
where the domain $\Dom(X_i)$ of each categorical attribute $X_i$ $(1 \leq i \leq n)$ is finite,
and each label attribute $Y_i$ $(1 \leq i \leq m)$ is binary. Our results can be generalized to cases where label attributes are categorical with more than two classes. For simplicity, we focus on binary labels in this paper.
Each record $t \in T$ is an $(n+m)$-dimensional tuple $(t.X_1, \ldots, t.X_n, t.Y_1, \ldots, t.Y_m)$. 

A \textbf{population} $s$ is an $n$-dimensional tuple $(s[1], \ldots, s[n])$ such that $s[i] \in \Dom(X_i) \cup \{\ast\}$, where $\ast$ is a wildcard character equivalent to \texttt{ALL} in data cube terminology~\cite{492099, harinarayan1996implementing, vitter1998data, zhao1997array}. 
Populations serve as selection criteria to define subsets of records from the base table.
In paritcular, the \textbf{coverage} of a population $s$ with respect to $T$, denoted by 
$
\cov_T(s) = \{t \in T \mid t.X_i = s[i] \ \vee \ s[i] = *,\, 1 \leq i \leq n\}
$, 
is the set of records in $T$ matching $s$.
We omit the subscript $T$ in $\cov_T(\cdot)$ when the context is clear.

Given populations $s$ and $s'$, $s$ is a \textbf{parent} of $s'$ (and $s'$ a \textbf{child} of $s$), denoted by $s \dotsucc s'$, if (1) there exists an attribute $X_j$ $(1\leq j \leq n)$ such that $s[j] = \ast$ and $s'[j] \neq \ast$, and (2) for all attributes $X_i$ $(1 \leq i \neq j \leq n)$, $s[i] = s'[i]$. 
Clearly, if $s \dotsucc s'$, then $\cov(s) \supseteq \cov(s')$. 
We call $X_j$ the \textbf{differential attribute} between $s$ and $s'$ and $s'[j]$ the \textbf{differential value}. 
We write the child as $s' = s\substitute{X_j}{s'[j]}$, where $\substitute{X_j}{\cdot}$ denotes the substitution of the $j$-th component.
Generally, a population may have multiple parents and children. 

A population $s$ is an \textbf{ancestor} of population $s'$ (and $s'$ a \textbf{descendant} of $s$), denoted by $s \succ s'$, if for all attributes $X_i$ $(1 \leq i \leq n)$ either $s[i]=\ast$ or $s[i]=s'[i]$, and for at least one attribute $s[i] \neq s'[i]$. 
In this case, $\cov(s) \supseteq \cov(s')$. 
If $s$ is a parent of $s'$, then $s$ is also an ancestor of $s'$, but not vice versa. 
We write $s \succeq s'$ if $s \succ s'$ or $s=s'$. 

Two populations $s_1$ and $s_2$ are \textbf{siblings} if both are children of a common parent $s$, under the differential attribute $X_j$, with different differential values. 
In this case, $s_1[i]=s_2[i]$ for all attributes $X_i$ $(1 \leq i \leq n, i\neq j)$, and $s_1[j] \neq s_2[j]$ and neither equals $\ast$.
Moreover, $\cov(s_1) \cap \cov(s_2) = \emptyset$.

We are interested in how often each label attribute $Y_i$ ($1 \leq i \leq m$) takes the value $1$ within a given population $s$ where $\cov_T(s)\neq \emptyset$. We define the \textbf{frequency statistics} of a non-empty population $s$ (w.r.t. $T$) under label attribute $Y_i$ as the conditional probability $P(Y_i=1 | s)$, or simply denoted by $P(Y_i | s)$, given by
$
\frac{|\cov_T(s) \cap \{ t \in T \mid t.Y_i = 1 \}|}{|\cov_T(s)|}.
$
%
\begin{example}[Notations]
\label{ex:2.1}
Consider the table $T(A, B, C, Y_1)$ in \Cref{tab:ex1}, which contains three categorical attributes $A$, $B$, and $C$ and a label attribute $Y_1$. 
Population $(a_1, b_1, *)$ covers the set of records with $A=a_1$ and $B=b_1$, i.e., $\cov(a_1,b_1,*) = \{t_1, t_2\}$. 
Population $(a_1, b_2, *)$ is a sibling of $(a_1, b_1, *)$, sharing the common parent $(a_1, *, *)$ under the differential attribute $B$. 
We can write $(a_1,b_1,*) = (a_1,*,*)\substitute{B}{b_1}$ and $(a_1,b_2,*) = (a_1,*,*)\substitute{B}{b_2}$. 
It is easy to verify that $P(Y_1 | (*, b_1, *))=\tfrac{2}{4}=0.50$. \qed
\end{example}
The set of all populations in a given table forms a lattice under the parent-child relation $\dotsucc$. 
Let $\mathcal{P}$ be the set of all populations in a given table. 
Consider a subset $\mathcal{E}\subseteq \mathcal{P}$. 
Two populations $s, s' \in \mathcal{E}$ are \emph{directly connected}, denoted $s\dotsim s'$, if either $s \dotprec s'$ or $s \dotsucc s'$. 
They are \emph{connected}, denoted $s \sim s'$, if either $s \dotsim s'$ or there exists a sequence $s_1, \ldots, s_k \in \mathcal{E}$ $(k>2)$ such that $s=s_1$, $s'=s_k$, and $s_j \dotsim s_{j+1}$ $(1 \leq j < k)$. 
$\mathcal{E}$ is \textbf{convex} if all pairs of populations in $\mathcal{E}$ are connected and, whenever a pair $s, s' \in \mathcal{E}$ satisfies $s \succ s'$, then all intermediate populations $s''$ with $s \succ s'' \succ s'$ also belong to $\mathcal{E}$.
\Cref{fig:lattice} shows the population lattice of \Cref{tab:ex1}.

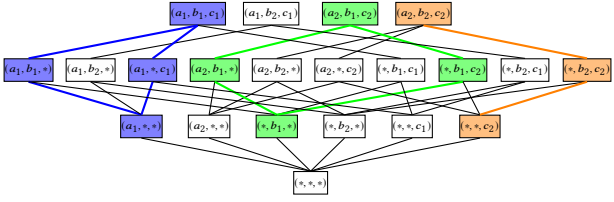
\begin{figure}[t]
    \centering
    \begin{tikzpicture}[
        every node/.style={font=\tiny, minimum size=0.4cm, inner sep=0pt},
        scale=0.75,
        transform shape
    ]
    \node[draw, fill=lightgray!20] (root) at (6,0) {$(\ast,\ast,\ast)$};
    
    \node[draw, fill=blue!50] (a1) at (3,1) {$(a_1,\ast,\ast)$};
    \node[draw, fill=lightgray!20] (a2) at (4.2,1) {$(a_2,\ast,\ast)$};
    \node[draw, fill=green!50] (b1) at (5.4,1) {$(\ast,b_1,\ast)$};
    \node[draw, fill=lightgray!20] (b2) at (6.6,1) {$(\ast,b_2,\ast)$};
    \node[draw, fill=lightgray!20] (c1) at (7.8,1) {$(\ast,\ast,c_1)$};
    \node[draw, fill=orange!50] (c2) at (9,1) {$(\ast,\ast,c_2)$};
    
    \node[draw, fill=blue!50] (a1b1) at (1,2) {$(a_1,b_1,\ast)$};
    \node[draw, fill=lightgray!20] (a1b2) at (2.1,2) {$(a_1,b_2,\ast)$};
    \node[draw, fill=blue!50] (a1c1) at (3.2,2) {$(a_1,\ast,c_1)$};
    \node[draw, fill=green!50] (a2b1) at (4.3,2) {$(a_2,b_1,\ast)$};
    \node[draw, fill=lightgray!20] (a2b2) at (5.4,2) {$(a_2,b_2,\ast)$};
    \node[draw, fill=lightgray!20] (a2c2) at (6.5,2) {$(a_2,\ast,c_2)$};
    \node[draw, fill=lightgray!20] (b1c1) at (7.6,2) {$(\ast,b_1,c_1)$};
    \node[draw, fill=green!50] (b1c2) at (8.7,2) {$(\ast,b_1,c_2)$};
    \node[draw, fill=lightgray!20] (b2c1) at (9.8,2) {$(\ast,b_2,c_1)$};
    \node[draw, fill=orange!50] (b2c2) at (10.9,2) {$(\ast,b_2,c_2)$};
    
    \node[draw, fill=blue!50] (a1b1c1) at (4,3) {$(a_1,b_1,c_1)$};
    \node[draw, fill=lightgray!20] (a1b2c1) at (5.3,3) {$(a_1,b_2,c_1)$};
    \node[draw, fill=green!50] (a2b1c2) at (6.7,3) {$(a_2,b_1,c_2)$};
    \node[draw, fill=orange!50] (a2b2c2) at (8,3) {$(a_2,b_2,c_2)$};
    
    \draw (root.north) -- (a1.south);
    \draw (root.north) -- (a2.south);
    \draw (root.north) -- (b1.south);
    \draw (root.north) -- (b2.south);
    \draw (root.north) -- (c1.south);
    \draw (root.north) -- (c2.south);
    
    \draw[blue, thick] (a1.north) -- (a1b1.south);
    \draw (a1.north) -- (a1b2.south);
    \draw[blue, thick] (a1.north) -- (a1c1.south);
    
    \draw (a2.north) -- (a2b1.south);
    \draw (a2.north) -- (a2b2.south);
    \draw (a2.north) -- (a2c2.south);
    
    \draw (b1.north) -- (a1b1.south);
    \draw[green, thick] (b1.north) -- (a2b1.south);
    \draw (b1.north) -- (b1c1.south);
    \draw[green, thick] (b1.north) -- (b1c2.south);
    
    \draw (b2.north) -- (a1b2.south);
    \draw (b2.north) -- (a2b2.south);
    \draw (b2.north) -- (b2c1.south);
    \draw (b2.north) -- (b2c2.south);
    
    \draw (c1.north) -- (a1c1.south);
    \draw (c1.north) -- (b1c1.south);
    \draw (c1.north) -- (b2c1.south);
    
    \draw (c2.north) -- (a2c2.south);
    \draw (c2.north) -- (b1c2.south);
    \draw[orange, thick] (c2.north) -- (b2c2.south);
    
    \draw[blue, thick] (a1b1.north) -- (a1b1c1.south);
    \draw (a1b2.north) -- (a1b2c1.south);
    \draw[blue, thick] (a1c1.north) -- (a1b1c1.south);
    \draw[green, thick] (a2b1.north) -- (a2b1c2.south);
    \draw (a2b2.north) -- (a2b2c2.south);
    \draw (a2c2.north) -- (a2b2c2.south);
    \draw (b1c1.north) -- (a1b1c1.south);
    \draw[green, thick] (b1c2.north) -- (a2b1c2.south);
    \draw (b2c1.north) -- (a1b2c1.south);
    \draw[orange, thick] (b2c2.north) -- (a2b2c2.south);
    \end{tikzpicture}
    \caption{Hasse diagram of the lattice formed by all populations in \Cref{tab:ex1} with respect to the parent-child relation $\dotsucc$. A parent is placed lower than its child. The \textcolor{blue}{blue} and \textcolor{green}{green} subsets are convex, while the \textcolor{orange}{orange} subset is non-convex.}
    \label{fig:lattice}
\end{figure}


\subsection{Simpson's Paradox}
\label{sec:sp}

Simpson's Paradox~\cite{https://doi.org/10.1111/j.2517-6161.1951.tb00088.x, pearson1899genetic} describes the counterintuitive phenomenon where the type of relationship (e.g., positive, negative, or independent) between two variables reverses when the population is partitioned into sub-populations. 
In this subsection, we formalize Simpson's Paradox in the multidimensional data, provide an illustrative example, and briefly review its well-known variants.

\begin{definition}
\label{def:simpson}
Consider a population $s$ and two sibling child populations $s_1 = s\substitute{X_j}{u_1}$ and $s_2 = s\substitute{X_j}{u_2}$ with differential attribute $X_j$ ($1 \le j \le n)$, where $u_1, u_2 \in \Dom(X_j)$ ($u_1 \neq u_2$) are the respective \emph{differential values}.
Let $X \in \{X_1, \ldots, X_n\} \setminus \{ X_j \}$ be a \textbf{separator attribute}
and $Y \in \{Y_1, \ldots, Y_m\}$ be a label attribute.
The tuple $(s_1, s_2, X, Y)$ is called an \textbf{association configuration (AC)}.
An AC is a \textbf{Simpson's Paradox} if the following holds:
\begin{enumerate}
    \item $P(Y | s_1) \geq P(Y | s_2)$;
    \item For every separator attribute value $v \in \Dom(X)$ with $\cov(s_1\substitute{X}{v}) \neq \emptyset$ and $\cov(s_2\substitute{X}{v}) \neq \emptyset$:
    \[
    P(Y | s_1\substitute{X}{v}) \leq P(Y | s_2\substitute{X}{v});
    \]
    \item Either the inequality in (1) is strict or all inequalities in (2) are strict.\qed
\end{enumerate}
\end{definition}
\nop{
\begin{definition}
Let $s$ be a population and $(u_1, u_2)$ $(u_1, u_2 \in \Dom(X_{i_0}),\ 1 \leq i_0 \leq n,\ u_1 \neq u_2)$ be a pair of \emph{differential values} for two sibling children populations $s_1 = s[X_{i_0} = u_1]$ and $s_2 = s[X_{i_0} = u_2]$. 
Let $X_{i_1}$ be a \textbf{separator attribute} $(1 \leq i_1 \leq n,\ i_1 \neq i_0)$ such that $s_1[i_1] = s_2[i_1] = *$, and let $Y_{i_2}$ be a label attribute. 
The tuple $(s_1, s_2, X_{i_1}, Y_{i_2})$ is called an \textbf{association configuration (AC)}. 

An AC is a \textbf{Simpson's Paradox} if the following hold:
\begin{enumerate}
    \item $P(Y_{i_2} \mid s_1) \geq P(Y_{i_2} \mid s_2)$;
    \item For every $v \in \Dom(X_{i_1})$ with $\cov(s_1[X_{i_1} = v]) \neq \varnothing$ and $\cov(s_2[X_{i_1} = v]) \neq \varnothing$, 
    \[
    P(Y_{i_2} \mid s_1[X_{i_1}=v]) \leq P(Y_{i_2} \mid s_2[X_{i_1}=v]);
    \]
    \item Either the inequality in (1) or all inequalities in (2) are strict. \qed
\end{enumerate}
\end{definition}
}

The directions of the inequalities in (1) and (2) may be reversed simultaneously. 
In addition, partitioning can be generalized to a set of multiple separator attributes $\mathbf{X}$:
for each value combination $\mathbf{v} \in \prod_{X_j \in \mathbf{X}} \Dom(X_j)$,
we consider sub-populations $s_1\substitute{\mathbf{X}}{\mathbf{v}}$ and $s_2\substitute{\mathbf{X}}{\mathbf{v}}$. 

\Cref{ex:simpson} shows an example of Simpson's paradox, where $((*, b_1, *), (*, b_2, *), A, Y_1)$ is an associate configuration.
For clarity, the remainder of this paper assumes a single separator attribute, though our results extend directly to the multi-attribute case.

Over the past century, several variants of Simpson's Paradox have been studied. 
The most widely used is the \textbf{Association Reversal (AR)}~\cite{samuels1993simpson}, as formalized in \Cref{def:simpson}. 
A special case, \textbf{Yule's Association Paradox (YAP)}~\cite{yule1903notes}, occurs when there is no association in the sub-populations, yet an association emerges in the aggregate. 
\Cref{ex:simpson} is an example of YAP. 
Another form, the \textbf{Amalgamation Paradox (AMP)}~\cite{good1987amalgamation}, arises when the strength of association in the aggregate is greater (or smaller) than in each sub-population. 
A variant of AMP, the \textbf{Averaged Association Reversal (AAR)}~\cite{alipourfard2018can, wang2023learning}, occurs when the aggregate association differs from the average association across sub-populations. 
Both AMP and AAR are special cases of AR. 
For an in-depth review of Simpson's Paradox, we refer the reader to the survey by~\citet{sprenger2021simpson}.
Our framework naturally extends to these variants.

\section{Redundancy Among Instances of Simpson's Paradox}
\label{sec:coverage}

In practice, multiple populations in a table may have identical coverage, leading to different association configurations that capture essentially the same paradoxical behavior. 
For example, in \Cref{tab:ex1}, $\cov(a_1, *, *) = \cov(a_1, *, c_1) = \{t_1, t_2, t_3\}$. 
When a table has many attributes but relatively sparse records, such overlaps are common~\cite{lakshmanan2002quotient, beyer1999bottom, han2001efficient, kenneth1997fast, chen2003computation}. 
This incidental identicality can generate multiple Simpson's paradoxes that are redundant. 
In this section, we formalize this insight by defining \textbf{redundancy} through three types of equivalences that give rise to it, and then unifying them into a single definition.

\subsection{Three Types of Redundancies}
\label{sec:equivalence-types}

Redundancy may arise from three distinct sources. 
We first describe each case with formal statements and examples.

\subsubsection{Sibling Child Equivalence}
When sibling populations have identical coverage, their corresponding paradoxes are redundant.

\begin{lemma}[Sibling child equivalence]
\label{prop:sibling-eq}
Consider two association configurations $p = (s_1, s_2, X, Y)$ and $p' = (s'_1, s'_2, X, Y)$
where $\cov(s_1) = \cov(s'_1)$ and $\cov(s_2) = \cov(s'_2)$.
If $p$ is a Simpson's paradox, then $p'$ is also a Simpson's paradox. \qed
\end{lemma}

\begin{example}[Sibling child equivalence]
\label{ex:sibling}
We extend \Cref{tab:ex1} to \Cref{tab:ex2} by adding attribute $D$ and label attribute $Y_2$. 
Similar to \Cref{ex:simpson}, $((*, b_1, *, *), (*, b_2, *, *), A, Y_1)$ is a Simpson's paradox. 
It can be verified that $((*, *, *, d_1), (*, *, *, d_2), A, Y_1)$ is also a Simpson's paradox due to sibling child equivalence. \qed
\end{example}

\begin{table}[t]
  \centering
  \caption{Data table $T(A,B,C,D,Y_1,Y_2)$ containing 7 records.}
    \begin{tabular}{|c|c|c|c|c|c|c|}
    \hline
          & $A$ & $B$ & $C$ & $D$ & $Y_1$ & $Y_2$ \bigstrut\\
    \hline
    $t_1$ & $a_1$ & $b_1$ & $c_1$ & $d_1$ & 0 & 0 \bigstrut[t]\\
    $t_2$ & $a_1$ & $b_1$ & $c_1$ & $d_1$ & 0 & 0 \\
    $t_3$ & $a_1$ & $b_2$ & $c_1$ & $d_2$ & 0 & 0 \\
    $t_4$ & $a_2$ & $b_1$ & $c_2$ & $d_1$ & 1 & 1 \\
    $t_5$ & $a_2$ & $b_1$ & $c_2$ & $d_1$ & 1 & 1 \\
    $t_6$ & $a_2$ & $b_2$ & $c_2$ & $d_2$ & 1 & 1 \\
    $t_7$ & $a_2$ & $b_2$ & $c_2$ & $d_2$ & 1 & 1 \bigstrut[b]\\
    \hline
    \end{tabular}%
  \label{tab:ex2}%
\end{table}%

\subsubsection{Separator Equivalence}
When two separator attributes induce partitions that are aligned via a one-to-one mapping, the resulting paradoxes are redundant.

\begin{lemma}[Separator equivalence]
\label{prop:division-equivalence}
Consider two association configurations $p = (s_1, s_2, X, Y)$ and $p' = (s_1, s_2, X', Y)$,
where $X \neq X'$ and there exists a one-to-one mapping $f : \Dom(X) \mapsto \Dom(X')$ such that
for every $v \in \Dom(X)$ and $s \in \{s_1, s_2\}$, $\cov(s\substitute{X}{v}) = \cov(s\substitute{X'}{f(v)})$.
If $p$ is a Simpson's paradox, then $p'$ is also a Simpson's paradox. \qed
%
\end{lemma}

\begin{example}[Separator equivalence]
\label{ex:division}
In \Cref{tab:ex2}, $((*, b_1, *, *), \\ (*, b_2, *, *), A, Y_1)$ is a Simpson's paradox. 
It can be verified that $((*, b_1, *, *), (*, b_2, *, *), C, Y_1)$ is also a Simpson's paradox due to separator equivalence. \qed
\end{example}

\subsubsection{Statistic Equivalence}
When label attributes are dependent, 
their paradoxes may be redundant. 
We identify three sufficient conditions for such equivalence.

\begin{lemma}[Statistic equivalence]
\label{prop:statistics-equivalence}
Consider two association configurations $p = (s_1, s_2, X, Y)$ and $p' = (s_1, s_2, X, Y')$ such that $Y \neq Y'$.
If $p$ is a Simpson's paradox and if any of the following (sufficient, and progressively less restrictive) conditions hold, then $p'$ is also a Simpson's paradox:
\begin{enumerate}
    \item For every $t \in \cov(s_1) \cup \cov(s_2)$, $t.Y = t.Y'$;
    \item For every $s \in \{s_1, s_2\}$, $P(Y | s) = P(Y'|s)$ and for every $v \in \Dom(X)$, $P(Y | s\substitute{X}{v}) = P(Y' | s\substitute{X}{v})$;    
    \item $\sign(P(Y | s_1) - P(Y | s_2)) = \sign(P(Y' | s_1) - P(Y' | s_2))$, and for every $v \in \Dom(X)$, 
    $\sign(P(Y | s_1\substitute{X}{v}) - P(Y | s_2\substitute{X}{v})) = \sign(P(Y' | s_1\substitute{X}{v}) - P(Y' | s_2\substitute{X}{v}))$.\qed
\end{enumerate}
%
\end{lemma}

\begin{example}[Statistic equivalence]
\label{ex:statistics}
We extend \Cref{tab:ex2} to \Cref{tab:ex3}. 
Consider $p = ((*, b_1, *, *), (*, b_2, *, *), A, Y_1)$. 
We observe $P(Y_1 | (*, b_1, *, *)) = 0.33 < P(Y_1 | (*, b_2, *, *)) = 0.40$. 
Partitioning by $A$ yields $P(Y_1 | (a_1, b_1, *, *)) = P(Y_1 | (a_1, b_2, *, *)) = 0$ and $P(Y_1 | (a_2, b_1, *, *)) = P(Y_1 | (a_2, b_2, *, *)) = 0.50$, confirming that $p$ is a Simpson's paradox. It follows that:
\begin{itemize}
    \item $((*, b_1, *, *), (*, b_2, *, *), A, Y_2)$ is statistic equivalent to $p$ (Case 1), since $t.Y_1 = t.Y_2$ for every record $t$.
    \item $((*, b_1, *, *), (*, b_2, *, *), A, Y_3)$ is statistic equivalent to $p$ (Case 2), as the frequency statistics of $Y_3$ match those of $Y_1$ across all relevant populations, even though $Y_1 \neq Y_3$ for some records.
    \item $((*, b_1, *, *), (*, b_2, *, *), A, Y_4)$ is statistic equivalent to $p$ (Case 3), since the signs of the frequency statistics differences coincide for $Y_1$ and $Y_4$ at both the aggregate and sub-populations levels. 
    \item Note that $Y_4$ could be constructed such that the sign of frequency statistics differences at the aggregate level matches that of $Y_1$, while for one sub-population the difference is negative for $Y_1$ but zero for $Y_4$. Despite sign mismatch, Simpson's Paradox is still preserved with $Y_4$, showing that Case 3 is sufficient but not necessary.\qed
\end{itemize}
\end{example}

\begin{table}[t]
  \centering
  \caption{Data table $T(A,B,C,D,Y_1,Y_2,Y_3,Y_4)$ with 11 records.}
    \begin{tabular}{|c|c|c|c|c|c|c|c|c|}
    \hline
          & $A$ & $B$ & $C$ & $D$ & $Y_1$ & $Y_2$ & $Y_3$ & $Y_4$ \bigstrut\\
    \hline
    $t_1$ & $a_1$ & $b_1$ & $c_1$ & $d_1$ & 0 & 0 & 0 & 0 \bigstrut[t]\\
    $t_2$ & $a_1$ & $b_1$ & $c_1$ & $d_1$ & 0 & 0 & 0 & 0 \\
    $t_3$ & $a_1$ & $b_1$ & $c_1$ & $d_2$ & 0 & 0 & 0 & 0 \\
    $t_4$ & $a_1$ & $b_2$ & $c_1$ & $d_1$ & 0 & 0 & 0 & 0 \\
    $t_5$ & $a_1$ & $b_2$ & $c_1$ & $d_2$ & 0 & 0 & 0 & 0 \\
    $t_6$ & $a_2$ & $b_1$ & $c_2$ & $d_1$ & 0 & 0 & 1 & 1 \\
    $t_7$ & $a_2$ & $b_1$ & $c_2$ & $d_1$ & 1 & 1 & 0 & 1 \\
    $t_8$ & $a_2$ & $b_1$ & $c_2$ & $d_2$ & 1 & 1 & 1 & 1 \\
    $t_9$ & $a_2$ & $b_2$ & $c_2$ & $d_1$ & 0 & 0 & 1 & 1 \\
    $t_{10}$ & $a_2$ & $b_2$ & $c_2$ & $d_2$ & 1 & 1 & 0 & 1 \\
    $t_{11}$ & $a_2$ & $b_2$ & $c_2$ & $d_2$ & 1 & 1 & 1 & 1 \bigstrut[b]\\
    \hline
    \end{tabular}%
  \label{tab:ex3}%
\end{table}%

\subsection{Equivalence Classes of Simpson's Paradoxes}
\label{sec:redundancy}

Redundant Simpson's paradoxes may arise from more than one of the equivalence types, sometimes combining sibling child, separator, and statistic equivalence within the same instances.

\begin{example}[Redundancy]
In \Cref{tab:ex2}, $((*, b_1, *, *), \\ (*, b_2, *, *), A, Y_1)$ is a Simpson's paradox. 
It can be verified that $((*, *, *, d_1), (*, *, *, d_2), C, Y_2)$ is also a Simpson's paradox, and in fact redundant due to sibling child, separator, and statistic equivalence simultaneously. 
\qed
\end{example}

Motivated by the above observation, we integrate the three equivalences into a unified notion.

\begin{definition}[Redundancy]
\label{def:coverage}
Two distinct paradoxes are:
\begin{itemize}
    \item \textbf{sibling child equivalent} if they satisfy \Cref{prop:sibling-eq};
    \item \textbf{separator equivalent} if they satisfy \Cref{prop:division-equivalence};
    \item \textbf{statistic equivalent} if they satisfy at least one of the conditions in \Cref{prop:statistics-equivalence}.
\end{itemize}
They are \textbf{redundant} if any of the above holds.\qed
\end{definition}

\begin{theorem}[Equivalence]
\label{thm:equiv}
Redundancy of Simpson's paradoxes is an equivalence relation. \qed
\end{theorem}

Because redundancy is an equivalence relation, the set of all Simpson's paradoxes can be partitioned into equivalence classes. 
Each class corresponds to a group of mutually redundant paradoxes. 
In some cases, an equivalence class of Simpson's paradox may contain only a single instance when no redundancy is observed.

\begin{example}[Equivalence class]
\label{ex:group-redundant}
Consider \Cref{tab:ex2}. 
From \Cref{ex:simpson}, $p_1 = ((*, b_1, *, *), (*, b_2, *, *), A, Y_1)$ is a paradox. 
We also obtain the following paradoxes redundant to $p_1$:
\begin{itemize}
    \item $p_2 = ((*, *, *, d_1), (*, *, *, d_2), A, Y_1)$ (sibling child equiv.);
    \item $p_3 = ((*, b_1, *, *), (*, b_2, *, *), C, Y_1)$ (separator equiv.);
    \item $p_4 = ((*, b_1, *, *), (*, b_2, *, *), A, Y_2)$ (statistic equiv.);
    \item $p_5 = ((*, *, *, d_1), (*, *, *, d_2), C, Y_1)$ (sibling and separator);
    \item $p_6 = ((*, *, *, d_1), (*, *, *, d_2), A, Y_2)$ (sibling and statistic);
    \item $p_7 = ((*, b_1, *, *), (*, b_2, *, *), C, Y_2)$ (separator and stat.);
    \item $p_8 = ((*, *, *, d_1), (*, *, *, d_2), C, Y_2)$ (all equivalences).
\end{itemize}
The set $\{p_1, p_2,\ldots, p_8\}$ forms an equivalence class. \qed
\end{example}

\subsection{Representing Equivalence Classes Concisely}
\label{sec:concise-representation}
In real datasets with many attributes, the number of redundant paradoxes can be large. 
We therefore propose a concise representation for an equivalence class of redundant Simpson's paradoxes
(or \textbf{redundant paradox group} for short).
To begin, we make the following observation.

\begin{lemma}[Product Space]
\label{prop:product-space}
Each redundant paradox group can be characterized by the product of:
$\mathcal{E}_1 \times \mathcal{E}_2 \times \mathbf{X} \times \mathbf{Y}$, where $\mathbf{X}$ is a set of separator attributes, $\mathbf{Y}$ is a set of label attributes, and $\mathcal{E}_1, \mathcal{E}_2$ are sets of sibling populations, each containing populations with identical coverage.
Any choice of $(s_1, s_2, X, Y) \in \mathcal{E}_1 \times \mathcal{E}_2 \times \mathbf{X} \times \mathbf{Y}$ where $s_1,s_2$ are siblings is a Simpson's paradox in the redundant paradox group.
\qed
\end{lemma}


\begin{example}
\label{ex:group-redundant-non-concise}
The redundant paradox group from \Cref{ex:group-redundant} is characterized by the product space of the following:
\begin{itemize}
    \item $\mathcal{E}_1 = \{(*, b_1, *, *), (*, *, *, d_1), (*, b_1, *, d_1)\}$,
    \item $\mathcal{E}_2 = \{(*, b_2, *, *), (*, *, *, d_2), (*, b_2, *, d_2)\}$,
    \item $\mathbf{X} = \{A, C\}$,
    \item $\mathbf{Y} = \{Y_1, Y_2\}$.
\end{itemize}
This product space encompasses multiple paradoxes. 
For instance, $((*, b_1, *, *), (*, b_2, *, *), A, Y_1)$ and $((*, *, *, d_1), (*, *, *, d_2), C, Y_2)$ are included. 
However, $((*, b_1, *, d_1), (*, b_2, *, d_2), A, Y_1)$ is not, since these two populations are not valid siblings. \qed
\end{example}

Next, we show that $\mathcal{E}_1$ and $\mathcal{E}_2$ can be represented more concisely.
We partition the set of all populations $\mathcal{P}$ based on coverage, denoted $\mathcal{P}/\equiv_{\cov}$, where each \textbf{coverage group} $\mathcal{E} \in \mathcal{P}/\equiv_{\cov}$ contains populations with identical coverage.
We show that each such coverage group exhibits a structural property in the population lattice: they form convex subsets. 
This means that if two populations belong to the same coverage group,
then all populations that lie between them in the lattice hierarchy must also belong to that group.
Given convexity of any $\mathcal{E} \in \mathcal{P}/\equiv_{\cov}$,
we call a population $s \in \mathcal{E}$ an \textbf{upper bound} of $\mathcal{E}$ if there is no $s' \in \mathcal{E}$ with $s \prec s'$ (i.e., least descendant);
similarly, we call $s$ a \textbf{lower bound} of $\mathcal{E}$ if there is no $s' \in \mathcal{E}$ with $s' \prec s$ (i.e., greatest ancestor).
We denote the set of upper bounds of $\mathcal{E}$ by $\ubound(\mathcal{E})$ and the set of lower bounds of $\mathcal{E}$ by $\lbound(\mathcal{E})$.
Convexity ensures that we can represent coverage groups $\mathcal{E}_1$ and $\mathcal{E}_2$ using only these bounds,
thereby avoiding enumerating all members, which can be prohibitively large in high-dimensional datasets.
Furthermore, we show that each coverage group has a unique upper bound
(though it can have multiple lower bounds).

\begin{property}[Convexity of coverage groups]
\label{prop:convex-property} 
Let $\mathcal{P}$ be the set of all populations.
For each coverage group $\mathcal{E} \in \mathcal{P}/\equiv_{\cov}$, $\mathcal{E}$ is a convex subset of coverage-identical populations.
Furthermore, $| \ubound(\mathcal{E}) | = 1$ and the least descendant is the unique upper bound.\qed
\end{property}


\begin{property}[Reconstruction from bounds]
\label{prop:convex-reconstruction}
Let $\mathcal{E} \subseteq \mathcal{P}$ be a convex subset of populations. 
Then $s \in \mathcal{E}$ if and only if there exist $s_l \in \lbound(\mathcal{E})$ and $\{s_u\} = \ubound(\mathcal{E})$ such that $s_l \preceq s \preceq s_u$. \qed
\end{property}

Using these properties and \Cref{prop:product-space}, we can concisely represent each redundant paradox group.
We remark that while Property~\ref{prop:convex-property} establishes that upper bounds uniquely identify (convex) coverage groups, reconstruction requires lower bounds. Given only the upper bound $s_u$, enumerating all $s \in \mathcal{E}$ requires verifying equality of coverage for each ancestor $s \prec s_u$. With lower bounds $\lbound(\mathcal{E})$, reconstruction is straightforward and efficient: following Property~\ref{prop:convex-reconstruction}, enumerate all $s$ satisfying $s_l \preceq s \preceq s_u$ for every $s_l \in \lbound(\mathcal{E})$. This reconstruction is essential because our concise representation must generate all Simpson's paradoxes in the group by enumerating every valid sibling pair across populations in $\mathcal{E}_1$ and $\mathcal{E}_2$ (as shown in Examples~\ref{ex:group-redundant},~\ref{ex:group-redundant-non-concise}).

\begin{definition}[Concise representation]
A redundant paradox group characterized by $\mathcal{E}_1 \times \mathcal{E}_2 \times \mathbf{X} \times \mathbf{Y}$ can be represented as:
\[
(\ubound(\mathcal{E}_1), \lbound(\mathcal{E}_1), \ubound(\mathcal{E}_2), \lbound(\mathcal{E}_2), \mathbf{X}, \mathbf{Y}),
\]
where $\ubound(\mathcal{E}_1)$ and $\ubound(\mathcal{E}_2)$ are the (unique) upper bounds of $\mathcal{E}_1$ and $\mathcal{E}_2$,
and $\lbound(\mathcal{E}_1)$ and $\lbound(\mathcal{E}_2)$ are their lower bounds.\qed
\end{definition}
By \Cref{prop:convex-reconstruction}, this representation is precise: all populations in $\mathcal{E}_1$ and $\mathcal{E}_2$ (and thus all redundant Simpson's paradoxes in the group) can be reconstructed from the bounds.

\begin{example}[Concise representation]
\label{ex:concise-representation}
From \Cref{ex:group-redundant}, the redundant paradox group can be represented as:
\begin{itemize}
    \item $\ubound(\mathcal{E}_1) = \{(*, a_1, *, b_1)\}$,
    \item $\lbound(\mathcal{E}_1) = \{(*, a_1, *, *), (*, *, *, b_1)\}$,
    \item $\ubound(\mathcal{E}_2) = \{(*, a_2, *, b_2)\}$, 
    \item $\lbound(\mathcal{E}_2) = \{(*, a_2, *, *), (*, *, *, b_2)\}$,
    \item $\mathbf{X} = \{A, C\}$,
    \item $\mathbf{Y} = \{Y_1, Y_2\}$.
\end{itemize}
All eight redundant Simpson's paradoxes from \Cref{ex:group-redundant} are captured.
For instance, $((*, b_1, *, *), (*, b_2, *, *), A, Y_1)$ is valid because $(*, b_1, *, *) \in \mathcal{E}_1$ and $(*, b_2, *, *) \in \mathcal{E}_2$ are siblings. 
By contrast, $((*, b_1, *, d_1), (*, b_2, *, d_2), A, Y_1)$ is not valid, since these two populations are not siblings. \qed
\end{example}

\section{Finding Non-Redundant Simpson's Paradoxes}
\label{sec:finding}

In this section, we first establish the \#P-hardness of finding non-redundant Simpson's paradoxes. Then we present our algorithmic framework for fast identification all non-redundant Simpson's paradoxes (i.e., redundant paradox groups) in a given table. 

\subsection{Complexity}

\begin{theorem}[\#P-Hardness]
\label{thm:sharp-p-hardness}
Finding all redundant paradox groups in a multidimensional table is \#P-hard.
\begin{proof}[Proof sketch]
We reduce from \#SAT~\cite{valiant1979complexity}. Given a Boolean formula $\phi$, we construct a table where each satisfying assignment of $\phi$ corresponds to a distinct group of redundant Simpson's paradoxes. Specifically, each satisfying assignment maps to a subset of records that exhibit paradoxes under Definition~\ref{def:simpson} and redundancies under Lemmas~\ref{prop:sibling-eq},~\ref{prop:division-equivalence}, and~\ref{prop:statistics-equivalence}. The reduction preserves the number of satisfying assignments of $\phi$ as the number of redundant paradox groups in the constructed table, making the problem \#P-hard.
\end{proof}
\end{theorem}

This hardness result shows that identifying non-redundant Simpson’s paradoxes is computationally challenging. Nevertheless, given their importance and practical relevance, in the remainder of this section we present techniques to accelerate the computation.

\subsection{General Algorithmic Framework}

Our approach builds directly on the concise representation developed in \Cref{sec:coverage}, ensuring that redundant instances are grouped into equivalence classes and represented concisely. Our framework proceeds in two main steps: 

\begin{enumerate}
    \item \textbf{Materialization.} We enumerate all non-empty populations, compute their coverage and frequency statistics, and organize populations into (convex) coverage groups.
    \item \textbf{Paradox discovery.} Using the materialized coverage and frequency statistics, we detect all instances of Simpson's paradox and simultaneously construct concise representations (\Cref{sec:redundancy}) for redundant paradox groups.
\end{enumerate}

\subsection{Materialization}
\label{sec:materialization}
\begin{algorithm}[t]
\caption{Brute-force Materialization}
\label{alg:brute-force-materialization}
\begin{algorithmic}[1]
\Input Data table $T = (\{X_i\}_{i=1}^n,\{Y_j\}_{j=1}^m)$
\Out Materialized populations
\For{each record $t \in T$}
\For{each ancestor $s$ of $t$}
\State Update $\cov(s)$ and $P(Y|s)$ for each $Y \in \{Y_1,\ldots,Y_m\}$.
\EndFor
\EndFor
\end{algorithmic}
\end{algorithm}
The first step of our algorithm materializes all non-empty populations, computing their coverage and frequency statistics.
A brute-force approach, illustrated in Algorithm~\ref{alg:brute-force-materialization},
is to iterate through each record $t$ in the base table $T$ and updates the coverage and frequency statistics for all ancestor populations of $t$. For each record, this requires enumerating all $2^n$ ancestors $s$ where $s[i] \in \{t.X_i, \ast\}$ for each attribute $X_i$. This method suffers from two inefficiencies. First, it performs repetitive computation by separately processing all $2^n$ ancestors for every record, leading to $|T| \times 2^n$ population updates even when many have identical coverage. Second, it fails to organize materialized populations into convex coverage groups, thereby missing opportunities to avoid repetitive materialization of intermediate populations within each coverage group. 

To address these limitations, we propose a depth-first search (DFS) approach adapted from~\cite{lakshmanan2002quotient, beyer1999bottom, han2001efficient}.
Effectively, the output will be coverage groups (i.e., convex subsets) of populations with identical coverage,
along with frequency statistics for each.
\Cref{alg:materialization} summarizes the procedure.

\begin{algorithm}[t]
\caption{DFS-based Materialization}
\label{alg:materialization}
\begin{algorithmic}[1]
\Input Data table $T = (\{X_i\}_{i=1}^n, \{Y_j\}_{j=1}^m)$, coverage threshold $\theta$
\Out Coverage-based partitioning $\mathcal{P}/\equiv_{\cov}$ of populations, where each group $\mathcal{E}$ is represented by $(\ubound(\mathcal{E}), \lbound(\mathcal{E}))$; frequency \textsc{stats} for each $\mathcal{E} \in \mathcal{P}/\equiv_{\cov}$, indexed by $\ubound(\mathcal{E})$.

\State Initialize the set $G$ of candidate coverage groups and \textsc{stats};
\State \textsc{DFS}$((\ast, \ast, \ldots, \ast), T, 0)$; \Comment{updates $G$ and \textsc{stats}}
\For{each unique upper bound $u$ such that $(u, \_) \in G$}
    \State $L \gets \{ s \mid (u, s) \in G \land (\nexists s': (u,s') \in G \land s \prec s' )\}$;
    \State Add $(u, L)$ as a coverage group in $\mathcal{P}/\equiv_{\cov}$;
\EndFor
\State \Return $\mathcal{P}/\equiv_{\cov}$, \textsc{stats};

\Function{DFS}{$s, T', k$}: \Comment{updates $G$ and \textsc{stats}}
\State $d \gets s$;
\For{each attribute $X_i\,(1\leq i \leq n)$ with $s[i]=\ast$}
    \If{$\exists\, v \in \Dom(X_i): \cov(s)=\cov(s\substitute{X_i}{v})$}
        \State $d[i] \gets v$; 
    \EndIf
\EndFor
\State Add $(d,s)$ to $G$; $\textsc{stats}(d) \gets \{P(Y_j | d)\}_{j=1}^m$;
\For{each attribute $X_h\,(k < h \leq n)$ with $d[h]=\ast$}
    \For{each $v \in \Dom(X_h)$}
        \State $T'_v \gets \{ t \mid t \in T' \land t.X_h = v \}$;
        \If{$|T'_v| \geq \theta \cdot |T|$}
            \State \textsc{DFS}$(d\substitute{X_h}{v}, T'_v, h)$;
        \EndIf
    \EndFor
\EndFor
\EndFunction
\end{algorithmic}
\end{algorithm}

\subsubsection{DFS-based Population Materialization}
\label{sec:dfs}

The algorithm builds the population lattice (see \Cref{fig:lattice} as an example) in a bottom-up manner. 
It starts from the root population $s_{\text{root}} = (\ast, \ast,\ldots,\ast)$ covering the entire dataset $T$, and progressively materializes populations that cover fewer records, thereby moving upward in the lattice.

At each recursive step, the DFS aims to identify a (candidate) convex coverage group populations. We begin with some lower-bound population $s$ and attempt to find its corresponding upper bound $s'$. 
The upper bound $s'$, intitially the same as $s$, is constructed by scanning the records in $\cov(s)$: for each attribute $X_i$ where $s[i] = \ast$, if all records share the value $v \in \Dom(X_i)$, we set $s'[i] = v$; otherwise, $s'[i] = \ast$. This constructions ensures $\cov(s') = \cov(s)$.

\begin{example}
\label{ex:upper-bound-jumping}
Consider \Cref{tab:ex2} and population $s = (a_1,\ast,\ast,\ast)$ with $\cov(s) = \{t_1,t_2,t_3\}$.
Scanning these records, we find that all share the value $c_1$ for attribute $C$. 
For $B$ and $D$, the records do not share a common value. 
Thus, the upper bound $s'$ is $(a_1,\ast,c_1,\ast)$. \qed
\end{example}

By \Cref{prop:convex-reconstruction}, a convex coverage group can be reconstructed from its upper and lower bounds. 
Thus, any population $s''$ between $s$ and $s'$ (i.e., $s \succeq s'' \succeq s'$) must have the same coverage. 
These intermediate populations do not need explicit materialization; their coverage and statistics can be inferred, greatly improving efficiency.

After identifying $s'$, the pair $(s, s')$ serves as a candidate coverage group,
and we recursively explore each child $\hat{s}$ of $s'$, continuing DFS with $\hat{s}$ as the lower bound of the next candidate coverage group.

\begin{example}
\label{ex:dfs-continuation}
Continuing from \Cref{ex:upper-bound-jumping}, the pair $(s=(a_1,\ast,\ast,\ast), s'=(a_1,\ast,c_1,\ast))$ defines a convex group.
We then explore children of $(a_1,\ast,c_1,\ast)$. 
For instance, $(a_1,b_1,c_1,\ast)$ is a child since $\cov((a_1,b_1,c_1,\ast)) \subset \cov((a_1,\ast,c_1,\ast))$ and is non-empty. DFS proceeds with $(a_1,b_1,c_1,\ast)$ as the next lower bound. \qed
\end{example}

The recursion stops when (1) a population $s$ covers fewer than a threshold $\theta \cdot |T|$ records (see \Cref{sec:population-pruning}, population pruning), or (2) the DFS reaches the top of the lattice.

\subsubsection{Coverage Group Merging}
\label{sec:group-merging}

DFS may discover the same coverage group via different lower bounds. 
Therefore, we merge candidate coverage groups that share the same upper bound. 
Each merged group has one upper bound and potentially multiple lower bounds.
We then refine the lower bounds by removing invalid ones, i.e., those that are descendants of others in the same group.

\begin{example}
\label{ex:group-merging}
In \Cref{tab:ex2}, populations $(a_1,\ast,\ast,\ast)$, $(\ast,\ast,c_1,\ast)$, and $(a_1,\ast,c_1,\ast)$ all have coverage $\{t_1,t_2\}$. 
DFS may discover this convex coverage group via two paths: 
(1) $(a_1,\ast,\ast,\ast) \to (a_1,\ast,c_1,\ast)$, or 
(2) $(\ast,\ast,c_1,\ast) \to (a_1,\ast,c_1,\ast)$. 
This yields two candidate coverage groups. 
After merging, the consolidated coverage group has upper bound $\{(a_1,\ast,c_1,\ast)\}$ and lower bounds $\{(a_1,\ast,\ast,\ast),(\ast,\ast,c_1,\ast)\}$, both valid since neither is an ancestor of the other. \qed
\end{example}

\subsubsection{Population Pruning}
\label{sec:population-pruning}

An important practical insight is that populations with very small coverage often have low analytical significance and are unlikely to be of interest to users. To address this, we introduce a simple pruning threshold $0 \leq \theta \leq 1$. Any population with coverage less than $\theta \cdot |T|$ is neither materialized nor considered in Simpson's paradoxes. This practical constraint also reduces computational cost: in the DFS-based materialization algorithm, if a population $s$ covers fewer than $\theta \cdot |T|$ records, we skip materializing both $s$ and all its descendants.

\begin{example}
\label{ex:population-pruning}
In \Cref{tab:ex2}, suppose $\theta=0.4$. 
Population $s = (\ast,\ast,c_1,\ast)$ covers $\{t_1,t_2,t_3\}$.
Since $3/7 \approx 43\% > 40\%$, $s$ is not pruned.
If $\theta=0.6$, then $s$ and all its descendants, such as $(a_1,\ast,c_1,\ast)$ and $(a_1,b_1,c_1,d_1)$, would be pruned. \qed
\end{example}

\label{sec:materialization-correctness}


\begin{theorem}[Completeness]
\label{thm:completeness}
\Cref{alg:materialization} materializes all non-empty populations that satisfy the coverage threshold. Furthermore, after group merging, \Cref{alg:materialization} yields maximal convex coverage groups of coverage-identical populations; that is, no population outside a group shares the same coverage as any population within it. \qed
\end{theorem}

\subsection{Finding Redundant Paradox Groups}
\label{sec:discover}

The discovery of Simpson's paradoxes can be viewed as a two-step process:  
(1) systematically enumerating all possible ACs, and  
(2) evaluating each AC against the definition of Simpson's paradox (\Cref{def:simpson}).  
\Cref{alg:brute-force} illustrates a brute-force method: for each non-empty population $s$, we enumerate all sibling child pairs $(s_1,s_2)$, and then combine each pair with every valid separator attribute and label attribute to form candidate ACs. Each AC is then checked against \Cref{def:simpson}.

\begin{algorithm}[t]
\caption{Brute-force finding of Simpson's paradoxes}
\label{alg:brute-force}
\begin{algorithmic}[1]
\Input Materialized populations $S$ from $T = (\{X_i\}_{i=1}^n, \{Y_j\}_{j=1}^m)$
\Out All instances of Simpson's paradox
\For{each population $s \in S$}
  \For{each $X_i$ with $s[i]=\ast$}
    \For{each pair $v_1,v_2 \in \Dom(X_i)$ where $u_1 \neq u_2$}
      \State $s_1 \gets s\substitute{X_i}{v_1};\quad s_2 \gets s\substitute{X_i}{v_2}$;
      \For{each $i' \neq i$ with $s[i']=\ast$}
        \For{each label attribute $Y$}
          \State Evaluate $(s_1,s_2,X_{i'}, Y)$ using \Cref{def:simpson};
        \EndFor
      \EndFor
    \EndFor
  \EndFor
\EndFor
\end{algorithmic}
\end{algorithm}

This exhaustive search has two major drawbacks.  
First, it does not organize discovered paradoxes into redundant paradox groups. 
Second, it wastes computation by (i) repeatedly iterating over populations with identical coverage, and (ii) evaluating ACs that are redundant to already discovered paradoxes.  

We therefore design optimizations to avoid repeated computation and concisely represent redundant paradox groups.

\subsubsection{Iteration over Coverage Groups}
\label{sec:upper-bound}

Instead of iterating over every non-empty population in \Cref{alg:brute-force}, we exploit the convex coverage groups discovered from \Cref{alg:materialization}.
From each coverage group, it suffices to consider only one representative population -- specifically, its unique upper bound (\Cref{prop:convex-property}) -- since all populations in the coverage group have identical coverage.
This helps avoid repeated computation over such populations.

\subsubsection{Constructing Redundant Paradox Groups}
\label{sec:extending}

Even after restricting to iteration over coverage groups, many ACs remain redundant.
This occurs due to three types of equivalence. First, sibling child equivalence: coverage groups contain multiple populations beyond their upper bounds, and these populations can form valid sibling pairs, generating sibling child equivalent Simpson's paradoxes. For example, if coverage groups $\mathcal{E}_1$ and $\mathcal{E}_2$ have upper bounds $(\ast,b_1,\ast,d_1)$ and $(\ast,b_2,\ast,d_1)$, the AC using the upper bounds is just one instance -- other populations like $(\ast,b_1,\ast,\ast) \in \mathcal{E}_1$ and $(\ast,b_2,\ast,\ast) \in \mathcal{E}_2$ are also siblings, creating sibling child equivalent ACs.
Second, separator and statistic equivalence: different separator attributes may induce identical sub-population partitions, and different label attributes may be perfectly correlated. For example, consider sibling populations $(\ast,b_1,\ast,\ast)$ and $(\ast,b_2,\ast, \ast)$ from \Cref{tab:ex2}. Observe that attributes $A$ and $C$ partition the data identically, and label values for $Y_1$ and $Y_2$ are perfectly correlated, then ACs $((\ast,b_1,\ast,\ast),(\ast,b_2,\ast,\ast),A,Y_1)$ and $((\ast,b_1,\ast,\ast),(\ast,b_2,\ast,\ast),C,Y_2)$ would be redundant.

We propose two strategies that exploit the redundancy conditions from \Cref{sec:redundancy}
to avoid repeated evaluations when constructing redundant paradox groups.
Importantly, we leverage the concise representation of such groups (\Cref{sec:concise-representation}),
and ensure that we maintain representational conciseness when constructing/extending these groups.

\paragraph{Construction by sibling child equivalence}  
Once a Simpson's paradox $p=(s_1,s_2,X,Y)$ is identified, all paradoxes sibling child equivalent to $p$ can be inferred without evaluation.
Thanks to \Cref{sec:concise-representation}, we can also encode the entire set of such paradoxes concisely.
This strategy is formalized below and implemented by \Cref{alg:sibling-equivalence-pruning}.

\begin{proposition}
\label{prop:coverage-equivalence-pruning}
Let $p=(s_1,s_2,X,Y)$ be a Simpson's paradox,
where $s_1$ and $s_2$ belong to coverage groups $\mathcal{E}_1$ and $\mathcal{E}_2$ in $\mathcal{P}/\equiv_{\cov}$, respectively.
Then for any $(s_1',s_2') \in \mathcal{E}_1 \times \mathcal{E}_2$ such that $s_1'$ and $s_2'$ are siblings, the AC $p'=(s_1',s_2',X,Y)$ is also a Simpson's paradox and redundant with respect to $p$.\qed
\end{proposition}

\begin{example}[Construction by sibling child equivalence]
\label{ex:coverage-equivalence-pruning}
In \Cref{tab:ex2}, consider the Simpson's paradox $p=((\ast,b_1,\ast,\ast),(\ast,b_2,\ast,\ast),A,Y_1)$
from \Cref{ex:sibling}.
With materialization by \Cref{alg:materialization}, populations $(\ast,b_1,\ast,\ast)$ and $(\ast,b_2,\ast,\ast)$ belong to coverage groups $\{(*, b_1, *, *), (*, *, *, d_1), (*, b_1, *, d_1)\}$ and $\{(*, b_2, *, *), (*, *, *, d_2), \\ (*, b_2, *, d_2)\}$, respectively.  
Valid sibling pairs across these groups yield additional paradoxes, such as $((\ast,*,*,d_1),(\ast,*,*,d_2),A,Y_1)$, which are redundant to $p$.
These can be directly included in the same redundant paradox group without further evaluation.
Furthermore, instead of enumerating these paradoxes, the group can be concisely represented by: 
\begin{align*}
    \ubound(\mathcal{E}_1)=\{(\ast,b_1,\ast,d_1)\}, &\quad \lbound(\mathcal{E}_1)=\{(\ast,b_1,\ast,\ast),(\ast,\ast,\ast,d_1)\},\\
    \ubound(\mathcal{E}_2)=\{(\ast,b_2,\ast,d_2)\}, &\quad \lbound(\mathcal{E}_2)=\{(\ast,b_2,\ast,\ast),(\ast,\ast,\ast,d_2)\},\\
    \mathbf{X}=\{A\}, &\quad \mathbf{Y}=\{Y_1\} \qquad \qquad \qquad \qquad \qquad \qed
\end{align*}
\end{example}

\begin{algorithm}[t]
\caption{Construction by sibling child equivalence}
\label{alg:sibling-equivalence-pruning}
\begin{algorithmic}[1]
\Function{ConstructBySib}{$p, \mathcal{P}/\equiv_{\cov}$}
  \State \textbf{Input:} Simpson's paradox $p=(s_1,s_2,X,Y)$; coverage-based partitioning $\mathcal{P}/\equiv_{\cov}$ of populations
  \State \textbf{Output:} Concise representation of a set of paradoxes sibling-child-equivalent to $p$
  \State Let $\mathcal{E}_1, \mathcal{E}_2 \in \mathcal{P}/\equiv_{\cov}$ be groups containing $s_1$ and $s_2$, resp.;
  \State \Return $(\ubound(\mathcal{E}_1), \lbound(\mathcal{E}_1), \ubound(\mathcal{E}_2), \lbound(\mathcal{E}_2), \{X\}, \{Y\})$;
\EndFunction
\end{algorithmic}
\end{algorithm}

\paragraph{Extension by separator and statistic equivalence}  
Many paradoxes differ only by separator or label attributes but are still redundant.  
Once we know a paradox $p'$ is separator- or statistic-equivalent to some $p$ in a group $\mathbf{P}$ of sibling-child-equivalent paradoxes,
we can apply the separator and label attributes of $p'$ to all members of $\mathbf{P}$ to obtain more redundant paradoxes,
without evaluation.
Again, thanks to \Cref{sec:concise-representation},
such an extension can be efficiently carried out by the concise represention of $\mathbf{P}$,
without enumerating members of $\mathbf{P}$.
Algorithm~\ref{alg:aggr-division-equivalence} implements this strategy.

\begin{proposition}
\label{prop:pruning-2}
Let $\mathbf{P}$ be a set of sibling-child-equivalent Simpson's paradoxes with separator $X$ and label $Y$.
Suppose $(s_1',s_2',X',Y')$, where $X'\neq X$ or $Y'\neq Y$, is a Simpson's paradox redundant with respect to some paradox in $\mathbf{P}$.
Then for every $p = (s_1,s_2,X,Y) \in \mathbf{P}$, the AC $(s_1,s_2,X',Y')$ is also a redundant Simpson's paradox with respect to $p$.\qed
\end{proposition}

\begin{example}[Extension by separator and statistic equivalence]
\label{ex:pruning-2}
Continuing from \Cref{ex:coverage-equivalence-pruning}, suppose we have now identified
$p'=((\ast,*,*,d_1),(\ast,*,*,d_2),C,Y_2)$
as redundant with respect to the sibling-child-equivalent redundant paradox group $\mathbf{P}$ characterized by $\mathcal{E}_1 \times \mathcal{E}_2 \times \{A\} \times \{Y_1\}$.
\Cref{prop:pruning-2} implies that all combinations from $\mathcal{E}_1$ and $\mathcal{E}_2$ with separator $C$ and label $Y_2$ are also redundant paradoxes.
For example, $((\ast,b_1,\ast,\ast),(\ast,b_2,\ast,\ast),C,Y_2)$ can be added directly.  
The characterization of the group simply becomes $\mathcal{E}_1 \times \mathcal{E}_2 \times \{A,C\} \times \{Y_1,Y_2\}$.\qed
\end{example}

\begin{algorithm}[t]
\caption{Extension by separator/statistic equivalence}
\label{alg:aggr-division-equivalence}
\begin{algorithmic}[1]
\Function{ExtendBySepStat}{$\Tilde{\mathbf{P}}, p'$}
  \State \textbf{Input:} Concise rep. $\Tilde{\mathbf{P}}$ for a sibling-child-equivalent paradox group;
  new Simpson's paradox $p'=(s_1',s_2',X',Y')$ redundant with respect to some paradoxes in $\Tilde{\mathbf{P}}$
  \State \textbf{Output:} Updated $\Tilde{\mathbf{P}}$
  \State $(\ubound(\mathcal{E}_1), \lbound(\mathcal{E}_1), \ubound(\mathcal{E}_2), \lbound(\mathcal{E}_2), \mathbf{X}, \mathbf{Y}) \gets \Tilde{\mathbf{P}}$;
  \State \Return $(\ubound(\mathcal{E}_1), \lbound(\mathcal{E}_1), \ubound(\mathcal{E}_2), \lbound(\mathcal{E}_2), \\ \phantom{ freturn } \mathbf{X} \cup \{X'\}, \mathbf{Y} \cup \{Y'\})$;
\EndFunction
\end{algorithmic}
\end{algorithm}

\subsubsection{Complete Algorithm}

Finally, we integrate these optimizations into a comprehensive algorithm (\Cref{alg:simpson}).  
We iterate only over coverage groups (using upper bounds),
constructing a sibling-child-equivalence redundant paradox group as soon as one paradox is found,
and extending groups with separator and statistic equivalence when applicable.
We maintain a hashmap $I$ where each key is $\textsc{Sig}(p)$ (see \Cref{def:info} below) and the value is the concise representation of a redundant paradox group
(though a group may contain only a single paradox if no redundancy is observed).

\begin{definition}[Signature]
\label{def:info}
Given populations $s_1$ and $s_2$, we define their \emph{joint signature} with respect to label attribute $Y$ as a triple:
\begin{align*}
    \textsc{JSig}_Y(s_1, s_2) = \langle \cov(s_1), \cov(s_2), \sign(P(Y|s_1) - P(Y|s_2)) \rangle.
\end{align*}
For an AC $p=(s_1,s_2,X,Y)$, its \textbf{signature} defined as:
\begin{align*}
    \textsc{Sig}(p) =
        & \langle \; \textsc{JSig}_Y(s_1, s_2),\\
        & \phantom{\langle} \left\{
            \textsc{JSig}_Y(s_1\substitute{X}{v}, s_2\substitute{X}{v}) \mid v \in \Dom(X)
          \right\} \; \rangle. \qquad \qed
\end{align*}
\end{definition}
%
In implementation, coverage sets $\cov(\cdot)$ are represented using integer hashes rather than storing full record sets. The signature $\textsc{Sig}(\cdot)$ becomes a vector containing integer hashes for coverage sets and sign values from $\{-1, 0, +1\}$ for frequency statistic differences. This enables efficient detection of redundant paradoxes: as established by \Cref{prop:signature}, all paradoxes within the same redundant paradox group share identical signatures. Thus, we can efficiently determine if a paradox $p$ belongs to an already discovered redundant paradox group by checking if $\textsc{Sig}(p)$ exists as a key in $I$.
%

\begin{lemma}
\label{prop:signature}
Two Simpson's paradoxes $p$ and $p'$ are redundant if and only if $\textsc{Sig}(p)=\textsc{Sig}(p')$. \qed
\end{lemma}

\begin{algorithm}[t]
\caption{Finding non-redundant Simpson's paradoxes}
\label{alg:simpson}
\begin{algorithmic}[1]
\Input $\mathcal{P}/\equiv_{\cov}$ and \textsc{stats} as produced by \Cref{alg:materialization}
\Out Hashmap $I$ storing concise representations of redundant paradox groups, keyed by $\textsc{Sig}$
\State Initialize $I \gets \emptyset$;
\For{each coverage group $\mathcal{E} \in \mathcal{P}/\equiv_{\cov}$}
  \State Let $s \gets \ubound(\mathcal{E})$ \Comment{use upper bound as representative}
  \For{each $X_i$ with $s[i]=\ast$}
    \For{each pair $v_1, v_2 \in \Dom(X_{i_0})$}
      \State $s_1 \gets s\substitute{X_i}{v_1},\quad s_2 \gets \substitute{X_i}{v_2}$;
      \For{each $X_j \neq X_i$ with $s[j]=\ast$}
        \For{each label attribute $Y$}
          \State $p \gets (s_1,s_2,X_j,Y)$; \Comment{check first if $p$ is included in existing redundant paradox groups}
          \If{$\exists(\mathcal{E}_1,\mathcal{E}_2,\mathbf{X},\mathbf{Y})\in I$ $:$ $s_1\in \mathcal{E}_1 \wedge s_2\in \mathcal{E}_2 \wedge \\ \phantom{ if }\, X_j \in \mathbf{X} \wedge Y \in \mathbf{Y}$}
          \textbf{continue} \EndIf 
          \State Evaluate $p$ acc. Def.~\ref{def:simpson}; compute $\textsc{Sig}(p)$;
          \If{$p$ is a Simpson's paradox}
            \If{$I(\textsc{Sig}(p)) = \emptyset$}
                \State $I(\textsc{Sig}(p)) \gets \\\;\;\; \textsc{ConstructBySib}(p,\mathcal{E}/\equiv_{\cov})$;
            \Else
                \State $I(\textsc{Sig}(p)) \gets \\\;\;\; \textsc{ExtendBySepStat}(I(\textsc{Sig}(p)), p)$;
            \EndIf
          \EndIf
        \EndFor
      \EndFor
    \EndFor
  \EndFor
\EndFor
\State \Return $I$
\end{algorithmic}
\end{algorithm}

Together, these optimizations transform the brute-force enumeration into an efficient method that avoids repeated work while producing concise representations of redundant paradox groups.

\section{Experimental Results}
\label{sec:experiments}

In this section, we evaluate instances of coverage redundant Simpson's paradoxes on both real-world and synthetic datasets, described in \Cref{sec:datasets}. 
Our study is guided by three research questions (RQs): 
\textbf{RQ1} investigates whether (coverage) redundant Simpson's paradoxes are rare in practice (\Cref{sec:rq1}); 
\textbf{RQ2} evaluates the scalability of our computational framework (\Cref{sec:rq2}); and 
\textbf{RQ3} examines the structural robustness of the discovered (coverage) redundant Simpson's paradoxes (\Cref{sec:rq3}). 

For each question, we conduct quantitative experiments and provide detailed analyses. 
Overall, our results show that redundant Simpson's paradoxes occur frequently in real-world datasets, our method scales effectively in practice, and the identified paradoxes are structurally robust under data perturbation.

\subsection{Experimental Setup}
\label{sec:datasets}
We conducted our experiments on the Duke Computer Science Department's computing cluster, using nodes equipped with Intel Xeon Gold 5317 processors (3.0 GHz, 12 cores) and 64 GB of RAM.

We evaluate our methods on both real-world categorical datasets and synthetic datasets generated with controlled parameters. The real-world datasets allow us to measure the prevalence of redundant paradoxes, while the synthetic datasets provide a way to systematically assess efficiency and scalability under controlled conditions.

\subsubsection{Real-World Datasets}
\label{sec:real-world-data}

We use datasets from diverse domains:
\begin{itemize}
    \item \textbf{Adult:} A census income dataset with 48{,}842 records, 8 attributes (e.g., education, occupation), and a binary label indicating whether annual income exceeds \$50K~\cite{misc_adult_2}.
    \item \textbf{Mushroom:} A dataset of 8{,}124 records describing mushrooms using 22 categorical attributes (e.g., cap shape, habitat), with edibility as the binary label~\cite{misc_mushroom_73}.
    \item \textbf{Loan:} A large financial dataset\footnote{\url{https://www.kaggle.com/datasets/ikpeleambrose/irish-loan-data}} containing about 3 million loan applications, with 12 categorical attributes (e.g., loan purpose, home ownership) and a label of loan approval.
    \item \textbf{CDC Diabetes Health Indicators:} A healthcare dataset\footnote{\url{https://www.kaggle.com/datasets/alexteboul/diabetes-health-indicators-dataset/data}} of 253{,}681 individuals, with 35 categorical attributes (covering demographics, laboratory results, and lifestyle factors) and a binary label indicating diabetes status.
\end{itemize}

\begin{table}[t]
  \centering\small
  \caption{Simpson's paradoxes in real-world datasets.}
    \begin{tabular}{|c|c|c|c|c|}
    \hline
    Dataset & Adult & Mushroom & Loan  & Diabetes \\
    \hline \hline
    \#paradoxes & 3,880 & 6,878 & 18,330 & 1,464,250 \\
    \#groups & 3,460 & 4,931 & 16,293 & 1,065,189 \\
    \#standalone & 3,094 & 3,590 & 14,354 & 809,388 \\
    \#sibling-child eq. & 366 & 1,220 & 1,939 & 255,690 \\
    \#separator eq. & 0 & 146 & 0 & 340 \\
    \#statistic eq. & 0 & 0 & 0 & 0 \\
    \hline
    \end{tabular}
  \label{tab:sp-real-world}
\end{table}

\subsubsection{Synthetic Datasets}
\label{sec:synthetic}

To evaluate performance in a controlled setting, we employ a synthetic data generator that produces datasets with user-specified structural properties. The generator accepts the following key parameters:
\begin{itemize}
    \item $n$: number of categorical attributes;
    \item $m$: number of label attributes;
    \item $d$: number of values per attribute.
\end{itemize}

The generation process proceeds in two steps. First, it constructs individual instances of Simpson's paradox. For a given AC $(s_1, s_2, X, Y)$, the generator enforces a consistent trend across all sub-populations (e.g., $P(Y | s_1\substitute{X}{v}) > P(Y | s_2\substitute{X}{v})$ for all $v \in \Dom(X)$). Then, it solves an optimization problem to distribute records across sub-populations such that the aggregate association reverses (e.g., $P(Y | s_1) \leq P(Y | s_2)$). Detailed procedures are provided in the artifact supplements.

Second, the framework introduces redundancy by modifying the generated records. Sibling child equivalence is induced by aligning sibling pairs $(s_1, s_2)$ and $(s_1', s_2')$ so that $\cov(s_1) = \cov(s_1')$ and $\cov(s_2) = \cov(s_2')$. Separator equivalence is enforced by mapping domains of two separator attributes $X$ and $X'$ and updating each record $r$ with $r.X' = f(r.X)$. Finally, statistic equivalence is created by defining a new label $Y'$ as a direct copy of an existing label $Y$.

This workflow is repeated until the dataset reaches a target size. Each iteration generates at least one unique non-redundant paradox, together with multiple redundant variants, while ensuring that all populations and sub-populations involved in a paradox contain at least a minimum number of records.

By varying $n$, $d$, and $m$, the generator naturally controls the richness of paradoxes and redundancies. Larger $n$ and $d$ values, for instance, expand the number of potential sibling values and separator attributes, increasing opportunities for Simpson's paradoxes and sibling child and separator equivalences. We study these effects in detail in Figure~\ref{fig:synthetic-rq1} in Section~\ref{sec:rq1}.

\subsection{Q1: Are Coverage-Redundant Simpson's Paradoxes Rare?}
\label{sec:rq1}
 
Our analysis shows that (coverage) redundant Simpson's paradoxes are common in real-world datasets. 
As summarized in \Cref{tab:sp-real-world}, a substantial fraction of discovered paradoxes are redundant: 20.3\% in Adult, 47.8\% in Mushroom, 21.7\% in Loan, and 44.7\% in Diabetes.  
Among the three types of equivalence, \emph{sibling child equivalence} is the most prevalent across all datasets.  
\emph{Separator equivalence}, which requires a one-to-one correspondence between separator attributes, is less frequent and appears only in the higher-dimensional Mushroom (10.7\%) and Diabetes (0.1\%) datasets.  
No \emph{statistic equivalence} is observed because each dataset contains only a single label attribute.

\begin{figure}[t]
    \centering
    \begin{subfigure}[b]{0.22\textwidth}
        \centering
        \includegraphics[width=\textwidth, height=70pt]{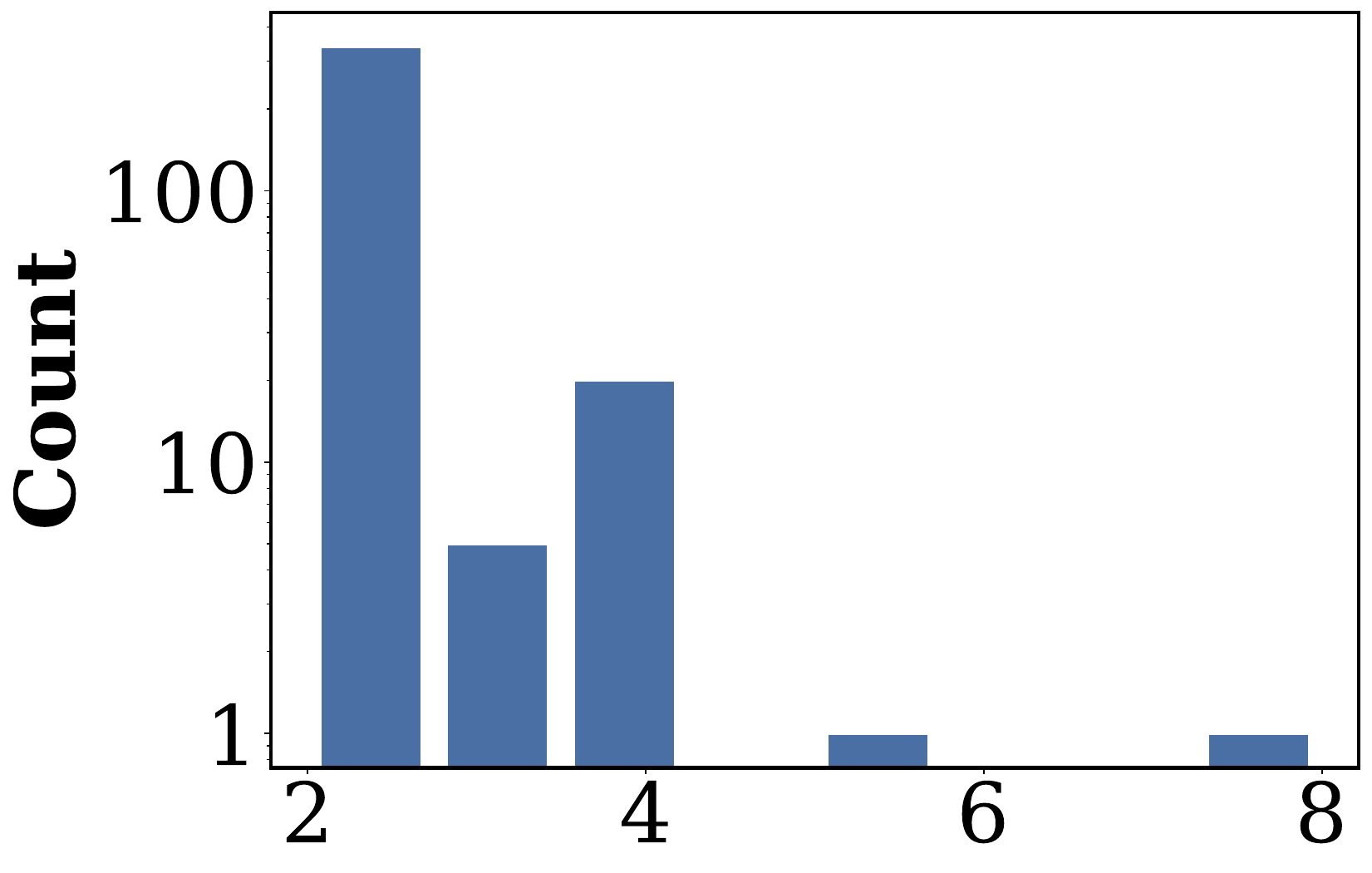}
        \caption{Adult}
    \end{subfigure}
    \hfill
    \begin{subfigure}[b]{0.22\textwidth}
        \centering
        \includegraphics[width=\textwidth, height=70pt]{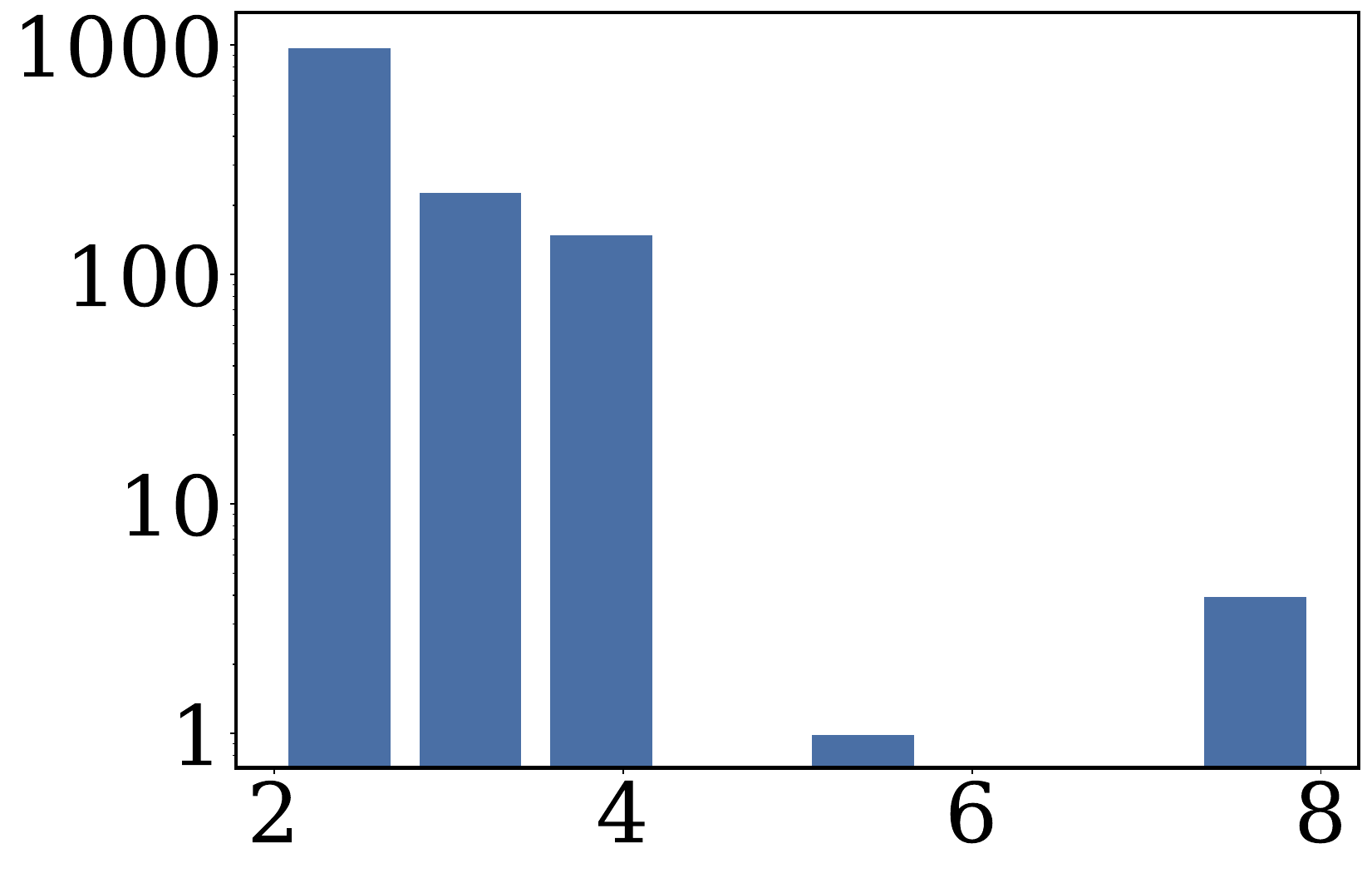}
        \caption{Mushroom}
    \end{subfigure}
    
    \begin{subfigure}[b]{0.22\textwidth}
        \centering
        \includegraphics[width=\textwidth]{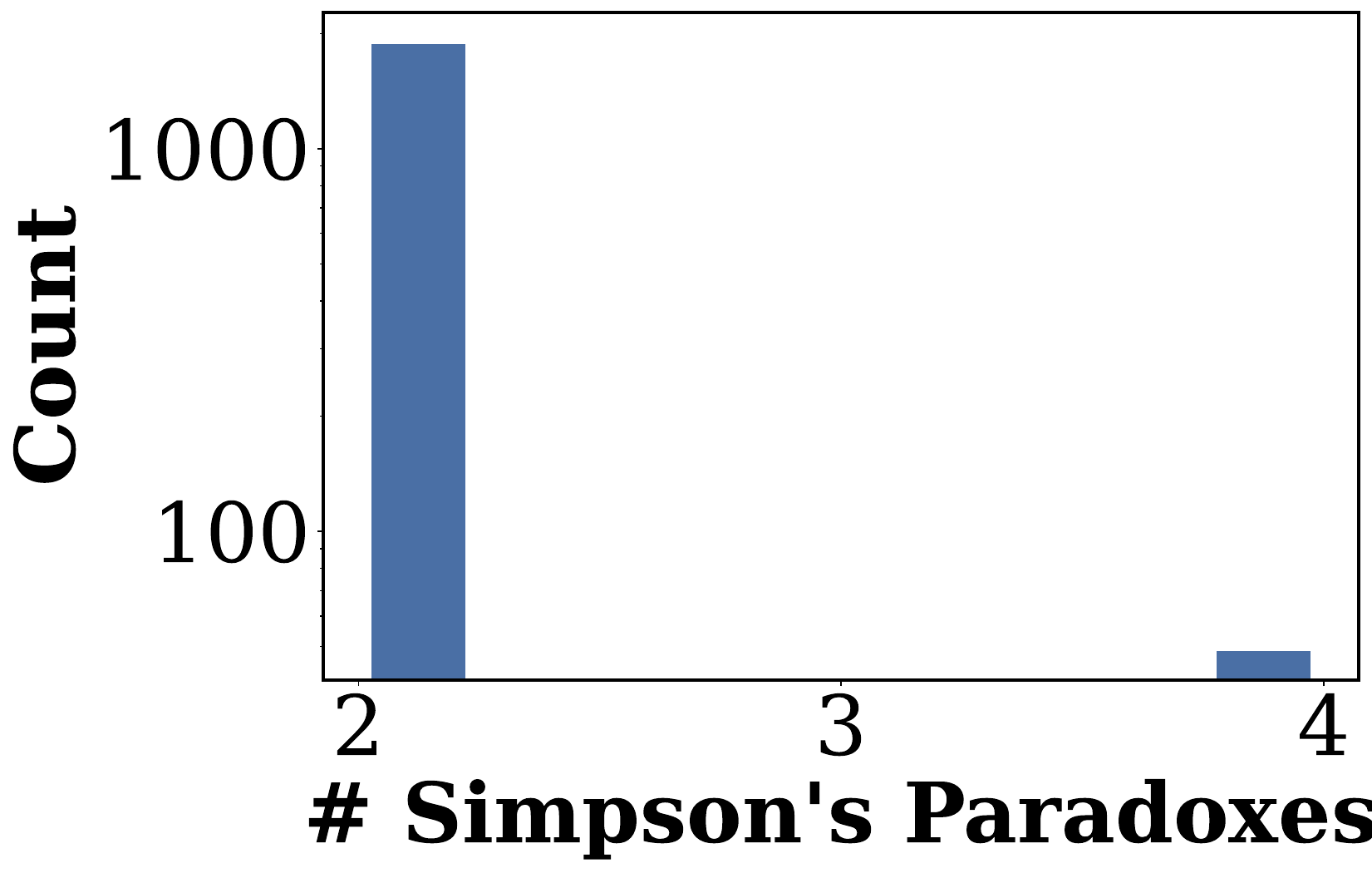}
        \caption{Loan}
    \end{subfigure}
    \hfill
    \begin{subfigure}[b]{0.22\textwidth}
        \centering
        \includegraphics[width=\textwidth]{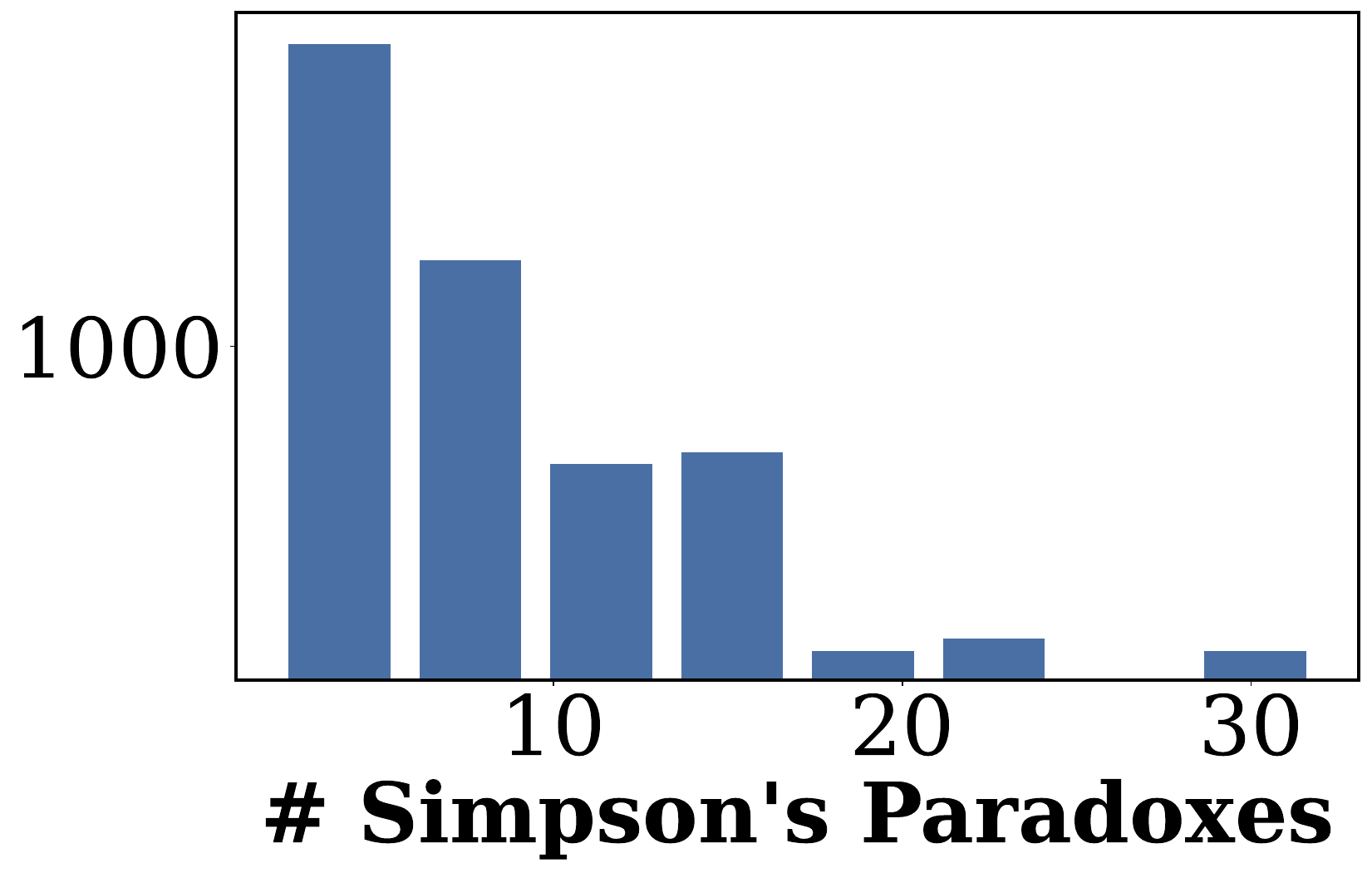}
        \caption{Diabetes}
    \end{subfigure}
    
    \caption{Distribution of the number of Simpson's paradoxes per redundant group in four real-world datasets.}
    \label{fig:dist-subset-size-real-world}
\end{figure}

To further analyze redundancy, \Cref{fig:dist-subset-size-real-world} shows the distribution of group sizes. 
In all datasets, most redundant paradox groups consist of only 2 or 3 paradoxes. 
In higher-dimensional datasets such as Mushroom and Diabetes, however, groups can grow much larger, containing up to 30 paradoxes due to the combined effects of sibling child and separator equivalences.

We also examine how redundancy patterns emerge in synthetic datasets.  
Since each synthetic dataset is generated with a fixed number of records (\Cref{sec:synthetic}), the number of unique, non-redundant paradoxes is limited, keeping the number of redundant paradox groups relatively stable across parameter settings.  
Our analysis therefore focuses on how generator parameters ($n=8$, $m=4$, $d=8$ by default) affect the \emph{total} number of paradoxes.  
With a stable number of paradox groups, growth in the total count reflects an increase in redundant instances.  
\Cref{fig:synthetic-rq1} presents the results, from which we draw the following observations:
\begin{itemize}
    \item \textbf{Number of categorical attributes ($n$):} The total number of paradoxes grows exponentially with $n$, as additional attributes create more opportunities for sibling child and separator equivalences.  
    \item \textbf{Domain cardinality ($d$):} The total number of paradoxes first increases and then decreases with $d$. Larger domains cause each paradox to cover more records, reducing the number of unique, non-redundant paradoxes that can be generated under the fixed record budget.  
    \item \textbf{Number of label attributes ($m$):} The total number of paradoxes increases linearly with $m$. Each additional label introduces a statistic-equivalent version of every existing paradox, effectively scaling the total count.  
\end{itemize}

\begin{figure}[t]
    \centering
    \begin{subfigure}[b]{0.15\textwidth}
        \centering
        \includegraphics[width=\textwidth]{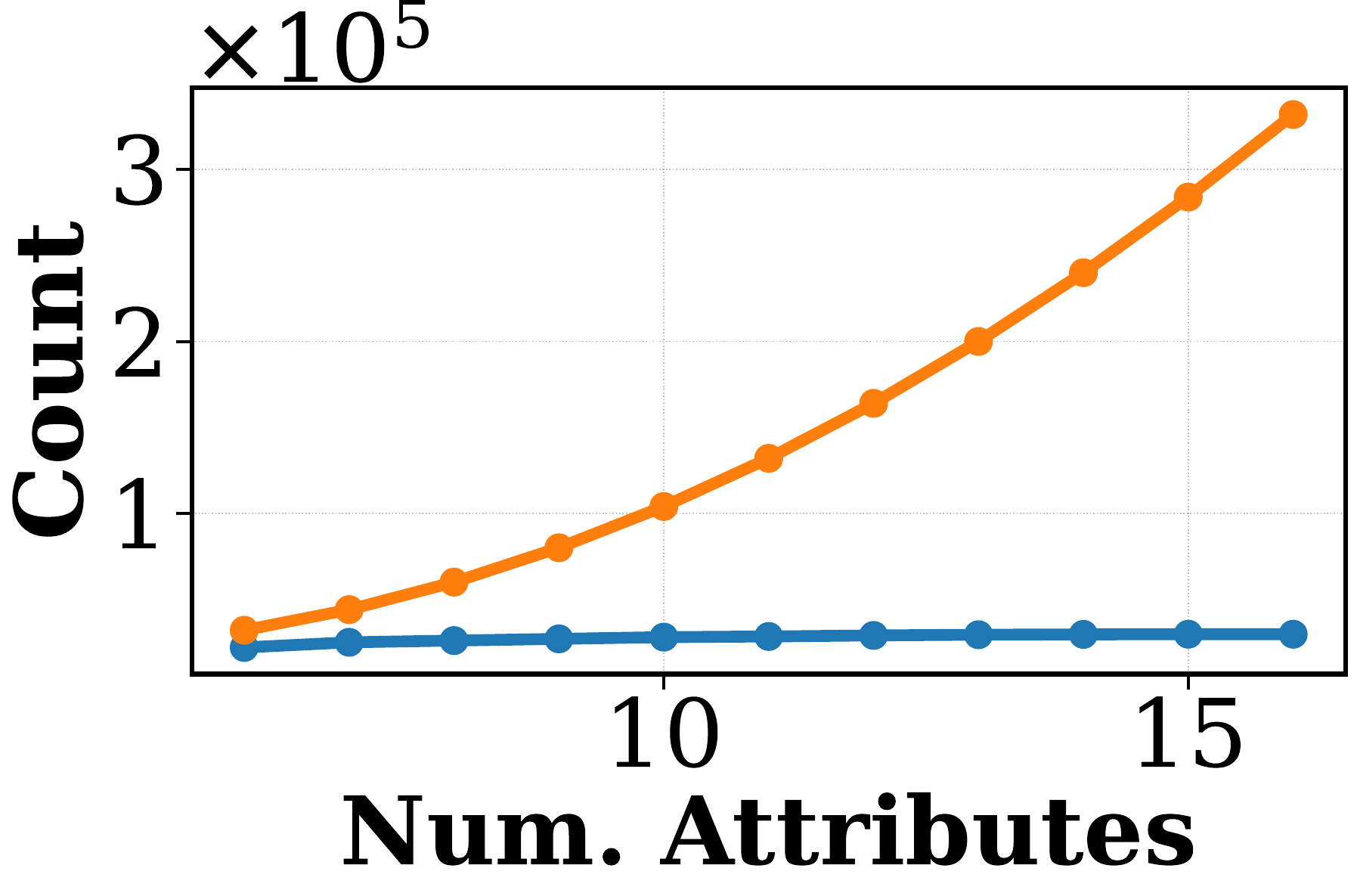}
        \label{fig:rq1a}
    \end{subfigure}
    \hfill
    \begin{subfigure}[b]{0.15\textwidth}
        \centering
        \includegraphics[width=\textwidth]{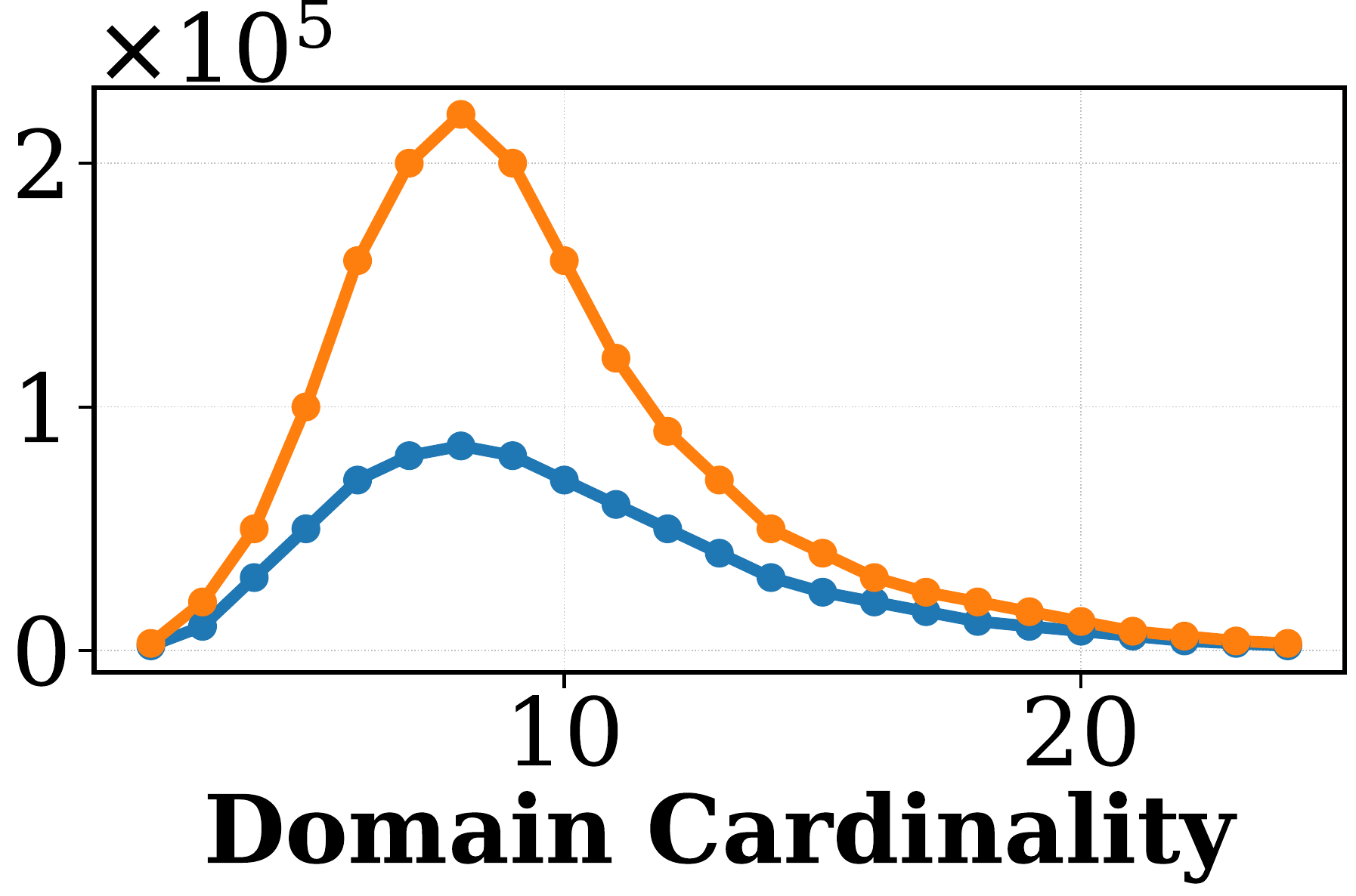}
        \label{fig:rq1b}
    \end{subfigure}
    \hfill
    \begin{subfigure}[b]{0.15\textwidth}
        \centering
        \includegraphics[width=\textwidth]{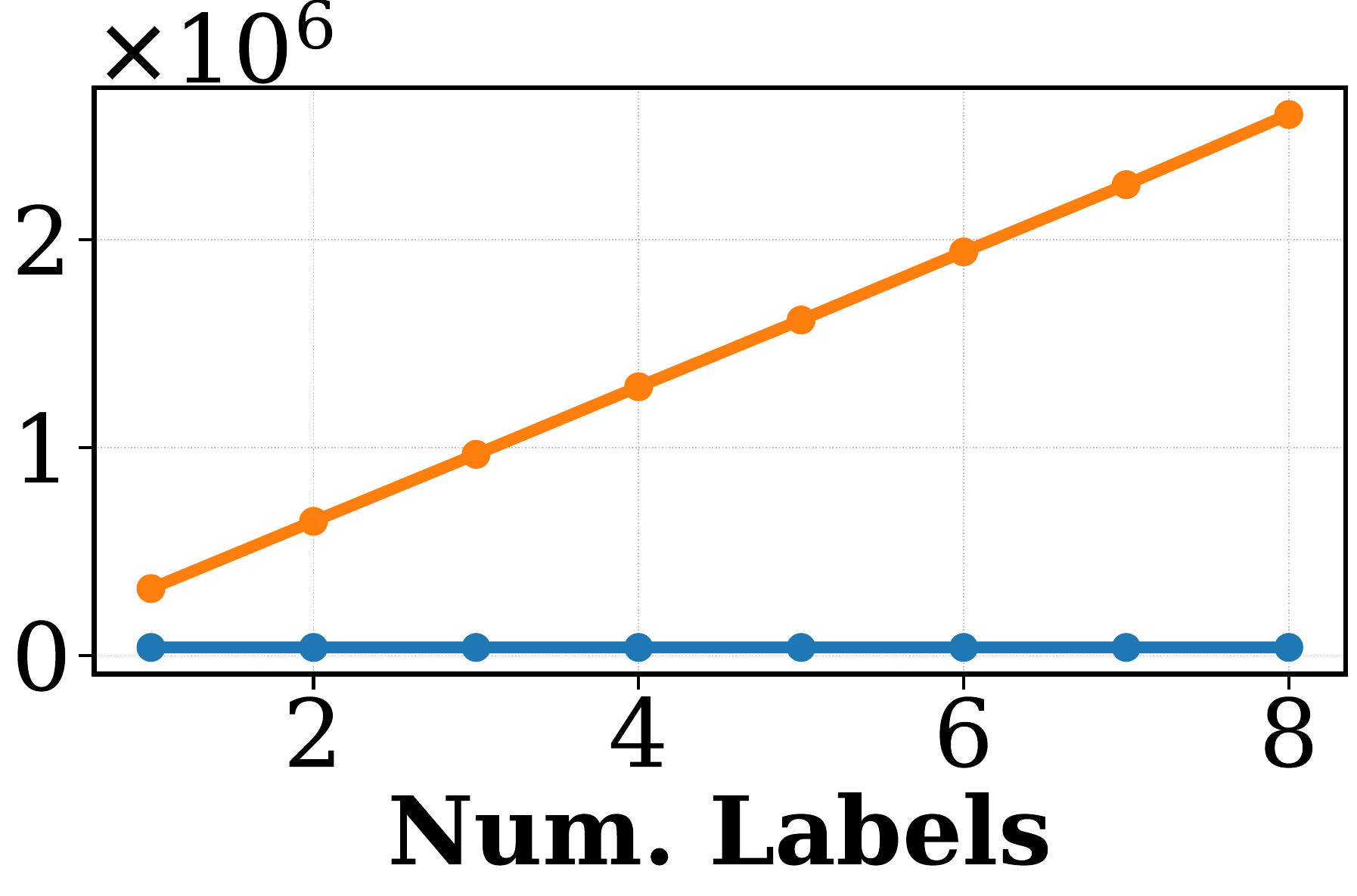}
        \label{fig:rq1c}
    \end{subfigure}\vspace{-5mm}
    \caption{Effect of dataset parameters on the total number of Simpson's paradoxes (orange) and redundant paradox groups (blue) in synthetic data.}
    \label{fig:synthetic-rq1}
\end{figure}

\subsection{Q2: Scalability}
\label{sec:rq2}

We next evaluate the computational efficiency of our method on both real-world and synthetic datasets. 
Our experiments yield three main findings:  
(i) our base algorithm (cf.\ Alg.~\ref{alg:materialization},~\ref{alg:simpson}), without population pruning (Sec.~\ref{sec:population-pruning}), achieves an average $6.72\times$ speedup over brute-force baselines (cf.\ Alg.~\ref{alg:brute-force-materialization},~\ref{alg:brute-force}) on real-world datasets, enabled by DFS-based materialization and redundancy-aware paradox discovery;  
(ii) population pruning, which excludes populations covering fewer than 0.1\% of total records, yields an additional 50\% run time reduction on top of our base algorithm; and  
(iii) run time on synthetic datasets scales predictably with the total number of Simpson's paradoxes, as characterized in Section~\ref{sec:rq1} (Fig.~\ref{fig:synthetic-rq1}).

\subsubsection{Scalability on Real-World Datasets}
\label{sec:real-world-scalability}

Figure~\ref{fig:real-world-scala} (a), on the left, compares three methods: (1) the brute-force baseline (including brute-force materialization cf.\ Alg.~\ref{alg:brute-force-materialization} and paradox discovery cf.\ Alg.~\ref{alg:brute-force}); (2) our base algorithm (including DFS-based materialization cf.\ Alg.~\ref{alg:materialization} and redundancy-aware paradox discovery cf.\ Alg.~\ref{alg:simpson}) with no population pruning (Sec.~\ref{sec:population-pruning}); and (3) our base algorithm with 0.1\% population pruning.  

Across all datasets, the base algorithm achieves an average $6.72\times$ speedup relative to brute-force.  
This gain is driven by two key optimizations:  
(1) DFS-based materialization (Alg.~\ref{alg:materialization}) yields an average $4.94\times$ improvement by identifying upper and lower bounds of coverage groups, avoiding explicit enumeration of intermediate populations;  
(2) redundancy-aware paradox discovery (Alg.~\ref{alg:sibling-equivalence-pruning}--\ref{alg:simpson}) achieves an average $20.24\times$ improvement by eliminating repeated enumeration and evaluation of redundant paradoxes.  
Finally, pruning populations below the 0.1\% threshold reduces run time by an additional 50\%, confirming that small-coverage populations are abundant in practice.  
Together, these optimizations make the problem computationally tractable even for high-dimensional datasets such as Diabetes, with over 250{,}000 records and 35 attributes.

Figure~\ref{fig:real-world-scala} (b), on the right, shows the effect of varying the pruning threshold $\theta$ from 0\% to 0.2\% of total records.  
We observe a reciprocal ($1/x$-like) trend: even minimal pruning (0.05\%) reduces run time by an average of 41.2\%, as many low-coverage populations are immediately eliminated.  
Improvements diminish beyond 0.1\%, where most such populations have already been pruned and run time stabilizes.

\begin{figure}[t]
    \centering
    \begin{subfigure}[b]{0.284\textwidth}
    \centering
    \includegraphics[width=0.95\linewidth]{figures/brute-scala-compare.pdf}
    \end{subfigure}
    \hfill
    \begin{subfigure}[b]{0.176\textwidth}
    \centering
    \includegraphics[width=0.995\linewidth]{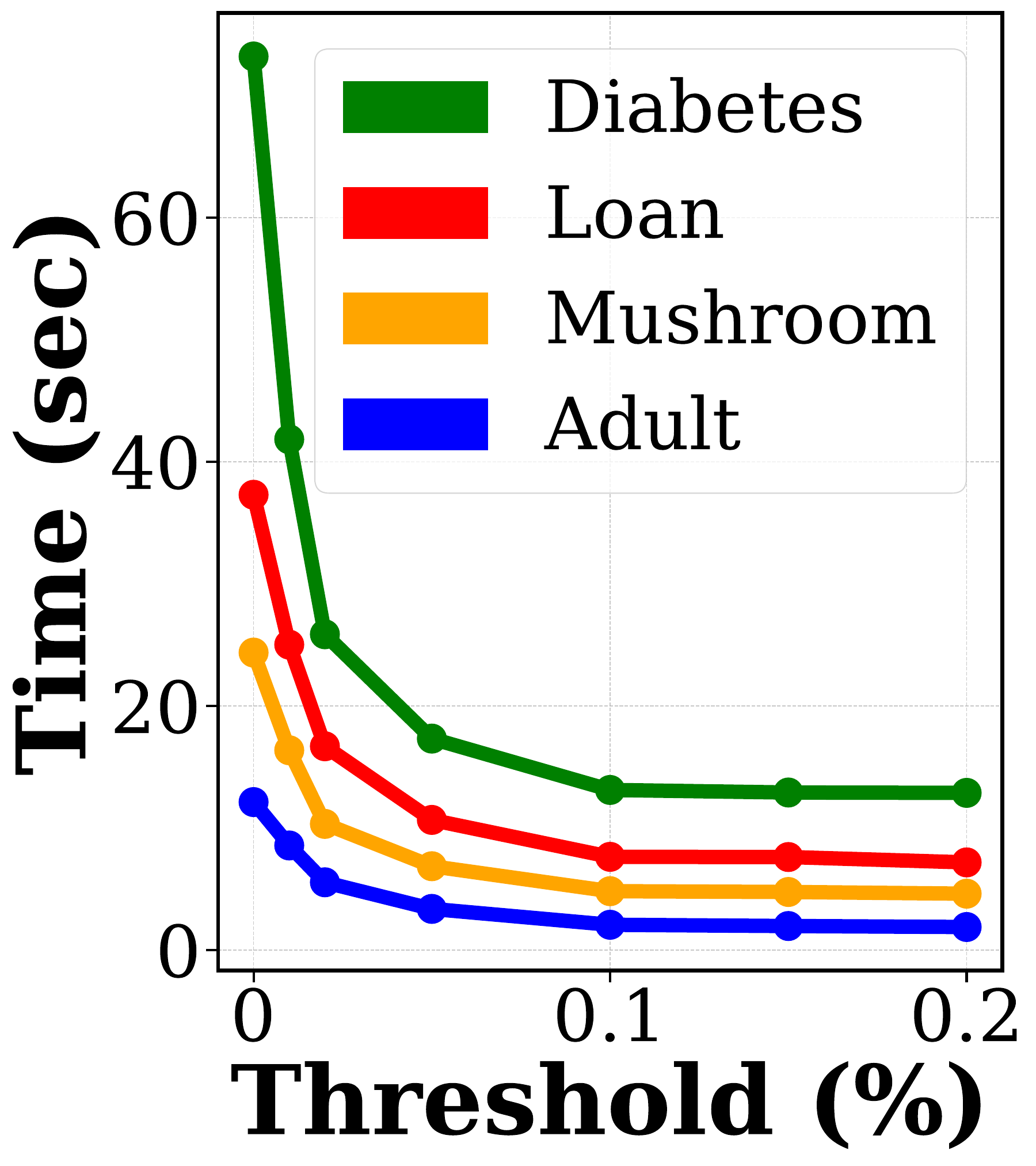}
    \end{subfigure}
    \caption{(a) Run time comparison on real-world datasets. Yellow shaded regions represent materialization time. (b) Run time vs.\ pruning threshold on real-world datasets.}
    \label{fig:real-world-scala}
\end{figure}

\subsubsection{Scalability on Synthetic Datasets}
\label{sec:synthetic-scalability}

\Cref{fig:synthetic-rq2} reports run time with respect to three parameters in synthetic datasets: number of attributes ($n$), domain cardinality ($d$), and number of labels ($m$), with default values $n=8$, $d=8$, and $m=4$.  
To stress-test scalability, we generate datasets with up to 30 million records. Each plot shows both materialization time and total run time.

Run time trends align closely with the total number of Simpson's paradoxes observed in Figure~\ref{fig:synthetic-rq1}.  
This correlation arises because the materialization phase dominates computation and directly depends on the number of populations:
\begin{itemize}
    \item \textbf{Number of categorical attributes ($n$):} As $n$ grows, the population lattice expands exponentially, since each population can branch into multiple children for the new attribute’s values. This explains the exponential increase in materialization time.
    \item \textbf{Domain cardinality ($d$):} Larger domains reduce the number of paradoxes (Figure~\ref{fig:synthetic-rq1}), which in turn reduces the number of populations requiring materialization. Run time therefore decreases after an initial peak.
    \item \textbf{Number of label attributes ($m$):} Each additional label requires computing frequency statistics for the same set of populations, yielding linear scaling in materialization time.
\end{itemize}

In contrast, paradox discovery time (the difference between total and materialization time) remains nearly constant across parameter settings.  
This stability follows from our synthetic generation procedure, where the number of redundant paradox groups is held constant regardless of the total number of paradoxes (\Cref{fig:synthetic-rq1}).

\begin{figure}[t]
    \centering
    \begin{subfigure}[b]{0.15\textwidth}
        \centering
        \includegraphics[width=\textwidth]{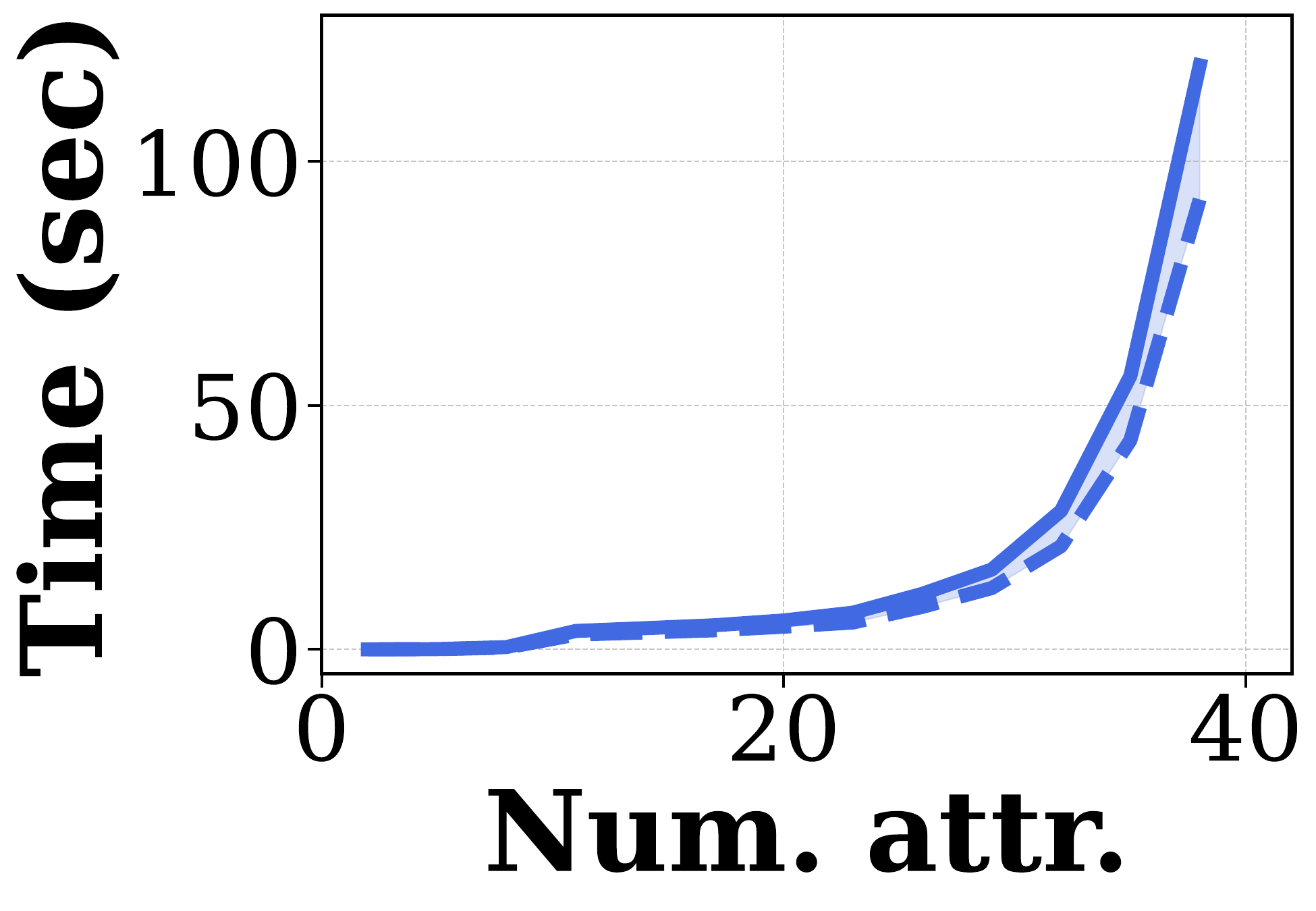}
    \end{subfigure}
    \hfill
    \begin{subfigure}[b]{0.15\textwidth}
        \centering
        \includegraphics[width=\textwidth]{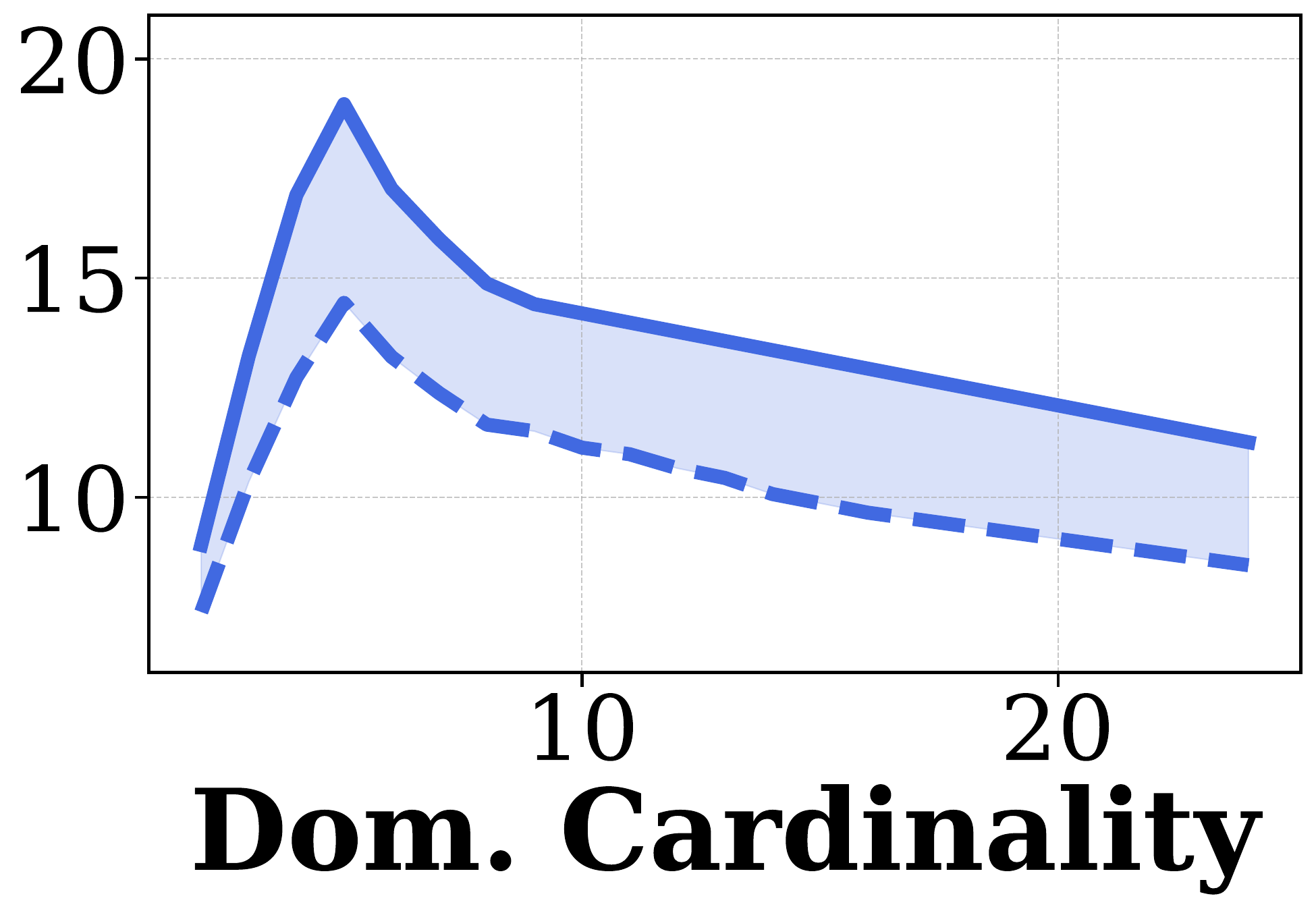}
    \end{subfigure}
    \hfill
    \begin{subfigure}[b]{0.15\textwidth}
        \centering
        \includegraphics[width=\textwidth]{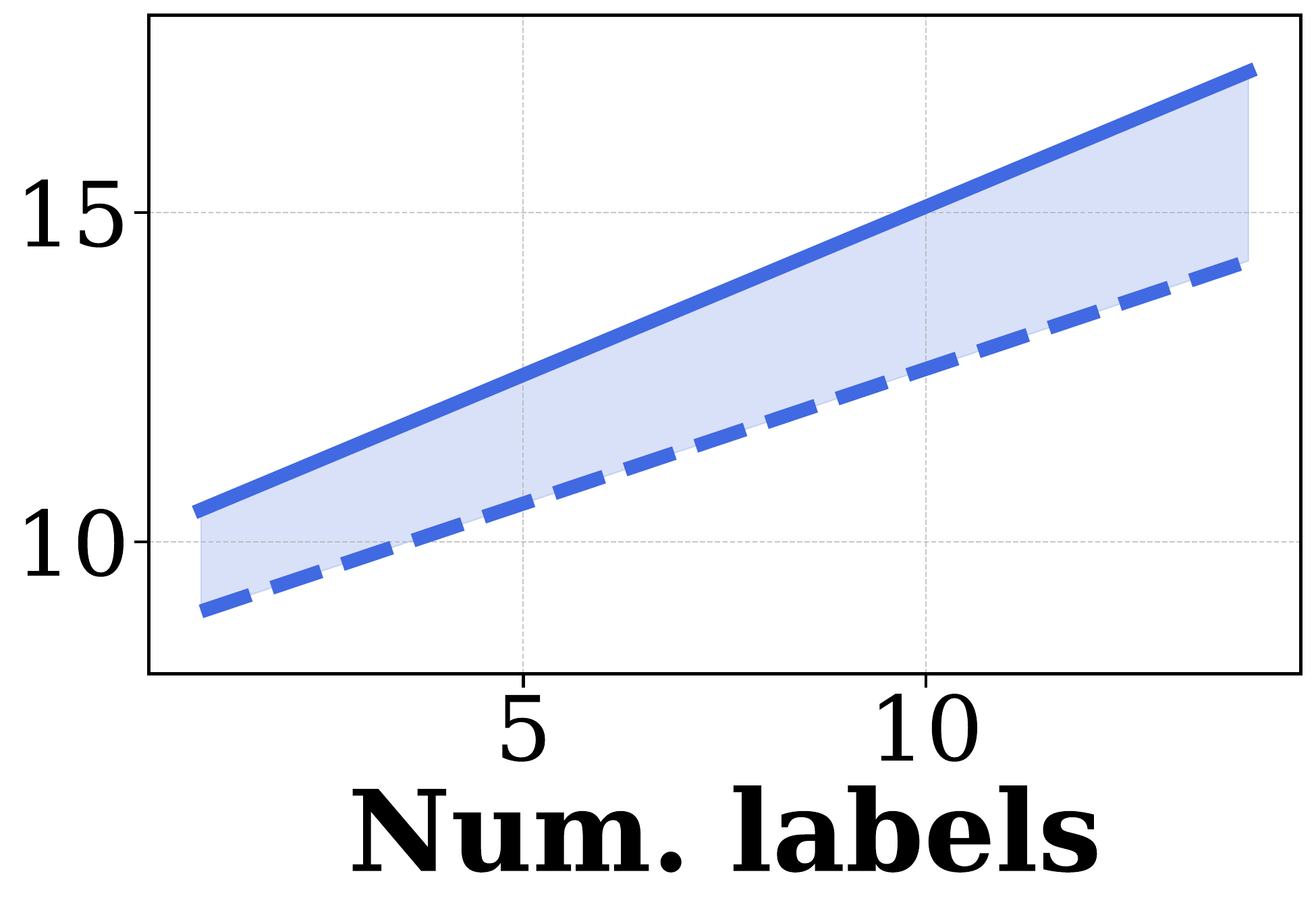}
    \end{subfigure}
    \caption{Run time scaling with synthetic dataset parameters. Solid lines denote total run time; dotted lines denote materialization time.}
    \label{fig:synthetic-rq2}
\end{figure}

\subsection{Q3: Are Coverage-Redundant Simpson's Paradoxes Robust?}
\label{sec:rq3}

Beyond demonstrating the prevalence of (coverage) redundant Simpson's paradoxes, we now investigate whether these paradoxical reversals and redundancies reflect \emph{genuine structural properties} of the data or are merely random artifacts introduced by noise or errors in data collection.  
To this end, we adopt a perturbation-based framework to test the \emph{robustness} of paradoxes and redundancies.  

At a high level, we measure \emph{tolerance:} given an observed Simpson's paradox (or redundancy), we quantify how much the data can be perturbed before the paradoxical (or redundant) relationship disappears. Robust patterns should persist under small perturbations, whereas random artifacts should vanish quickly.  

We evaluate two aspects:  
(1) the \textbf{robustness} of individual Simpson's paradoxes under label and record perturbations; and  
(2) the \textbf{persistence} of coverage redundancies under record perturbations.  

\subsubsection{Robustness of Individual Simpson's Paradoxes}

We first examine whether Simpson's paradoxes persist under perturbations. Any frequency statistics $P(Y | s)$ can be decomposed as
\[
P(Y|s) = \sum_{v \in \Dom(X)} \frac{|\cov(s\substitute{X}{v})|}{|\cov(s)|} \, P(Y | s\substitute{X}{v}).
\]
This weighted sum consists of two components: (i) weights representing record distribution across sub-populations, and (ii) frequency statistics within each sub-population. Simpson's paradoxes emerge from specific interactions between these components.  

Accordingly, for each paradox $p = (s_1, s_2, X, Y)$, we apply two perturbation strategies:  
\begin{itemize}
    \item \textbf{Label perturbation:} Randomly flip labels of 5\% of the records in $\cov(s_1) \cup \cov(s_2)$ to alter the frequency statistics $P(Y | s\substitute{X}{v})$. This tests whether paradoxical reversals depend critically on exact label assignments.  
    \item \textbf{Coverage perturbation:} Randomly modify the separator attribute $X$ for 5\% of records to change the weights $|\cov(s\substitute{X}{v})| / |\cov(s)|$ across sub-populations. We then reassign labels in each sub-population according to their original frequency statistics, thereby isolating the effect of record distribution.  
\end{itemize}

Each perturbation is repeated 10{,}000 times, and robustness is measured as the survival rate (percentage of trials where the paradox persists).  
Figure~\ref{fig:real-rq3} shows robustness under label perturbations. The fraction of robust paradoxes -- those surviving in at least 95\% of trials -- increases with higher pruning thresholds. This indicates that paradoxes supported by larger populations are more tolerant to perturbations.  

\begin{figure}[t]
    \centering
    \begin{subfigure}[b]{0.22\textwidth}
        \centering
        \includegraphics[width=\textwidth]{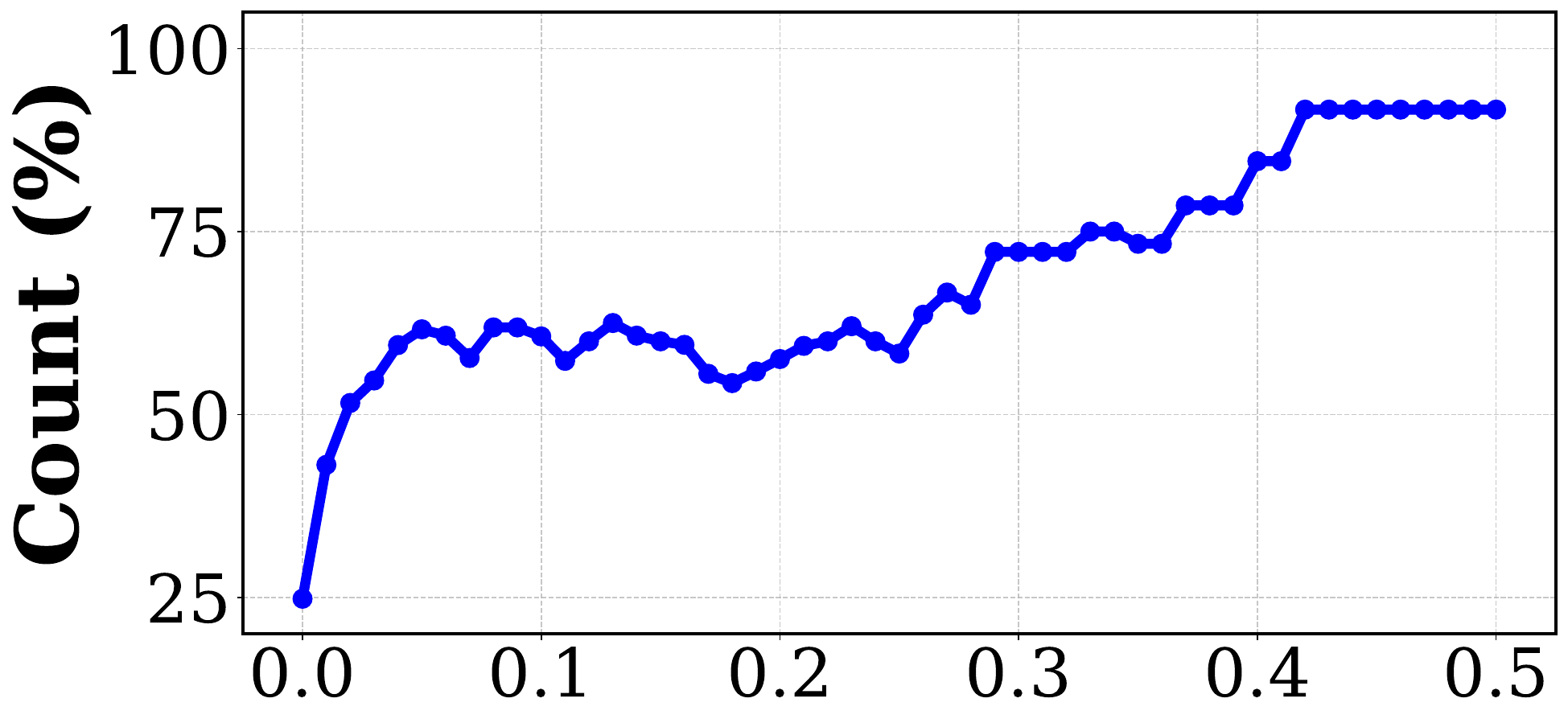}
        \caption{Adult}
        \label{fig:rq3a}
    \end{subfigure}
    \hspace{0.025\textwidth}
    \begin{subfigure}[b]{0.22\textwidth}
        \centering
        \includegraphics[width=\textwidth]{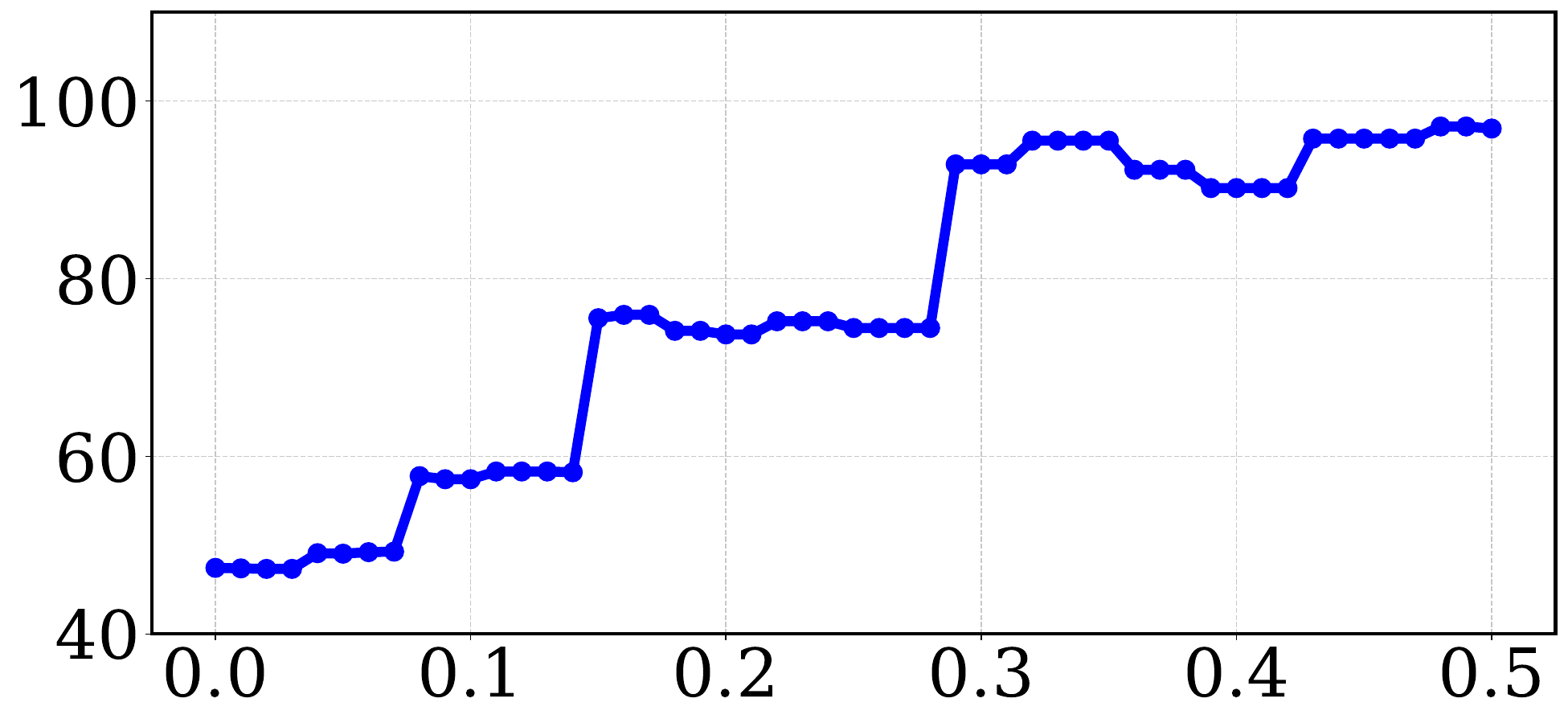}
        \caption{Mushroom}
        \label{fig:rq3b}
    \end{subfigure}
    \begin{subfigure}[b]{0.22\textwidth}
        \centering
        \includegraphics[width=\textwidth]{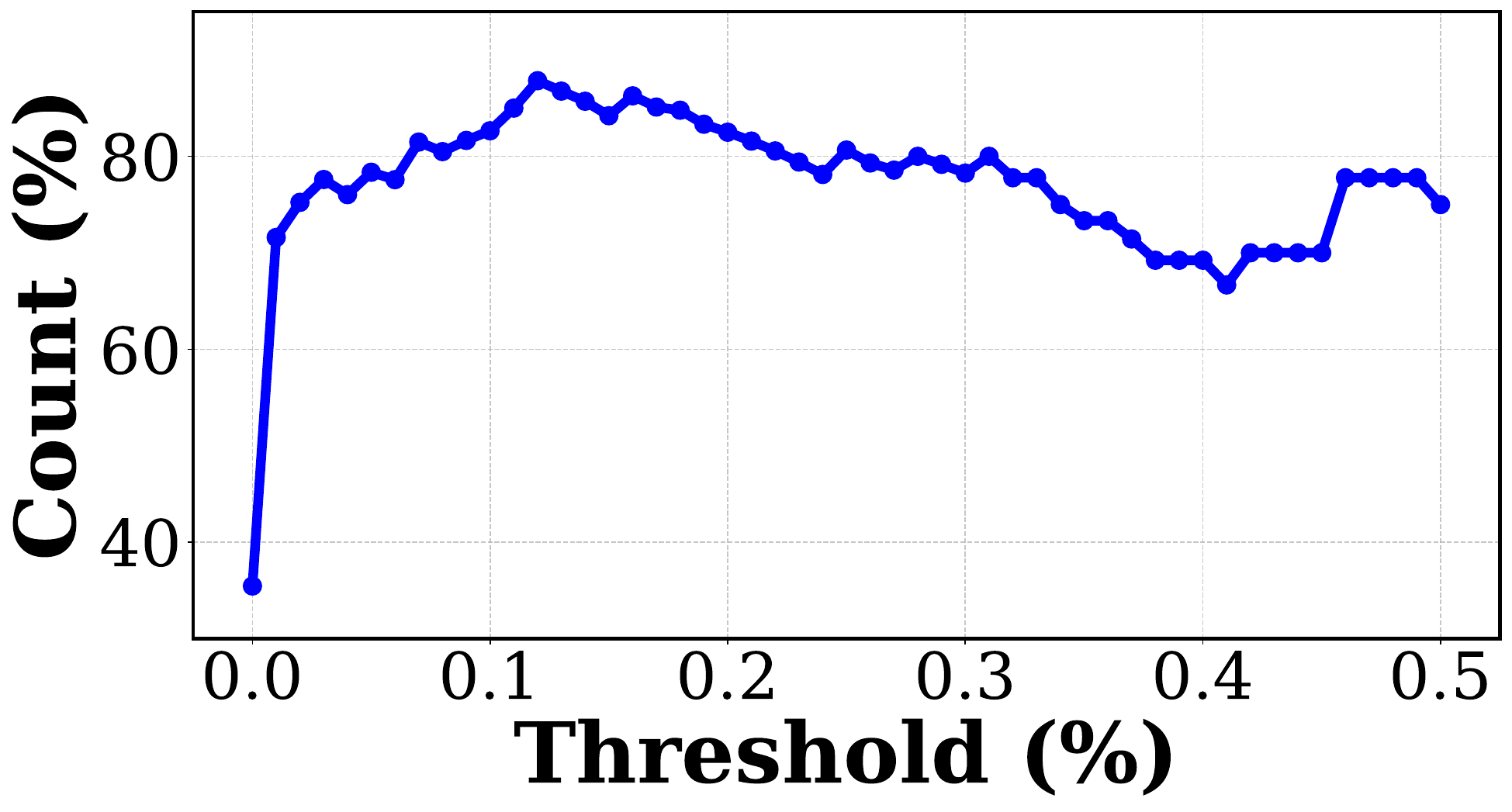}
        \caption{Loan}
        \label{fig:rq3c}
    \end{subfigure}
    \hspace{0.025\textwidth}
    \begin{subfigure}[b]{0.22\textwidth}
        \centering
        \includegraphics[width=\textwidth]{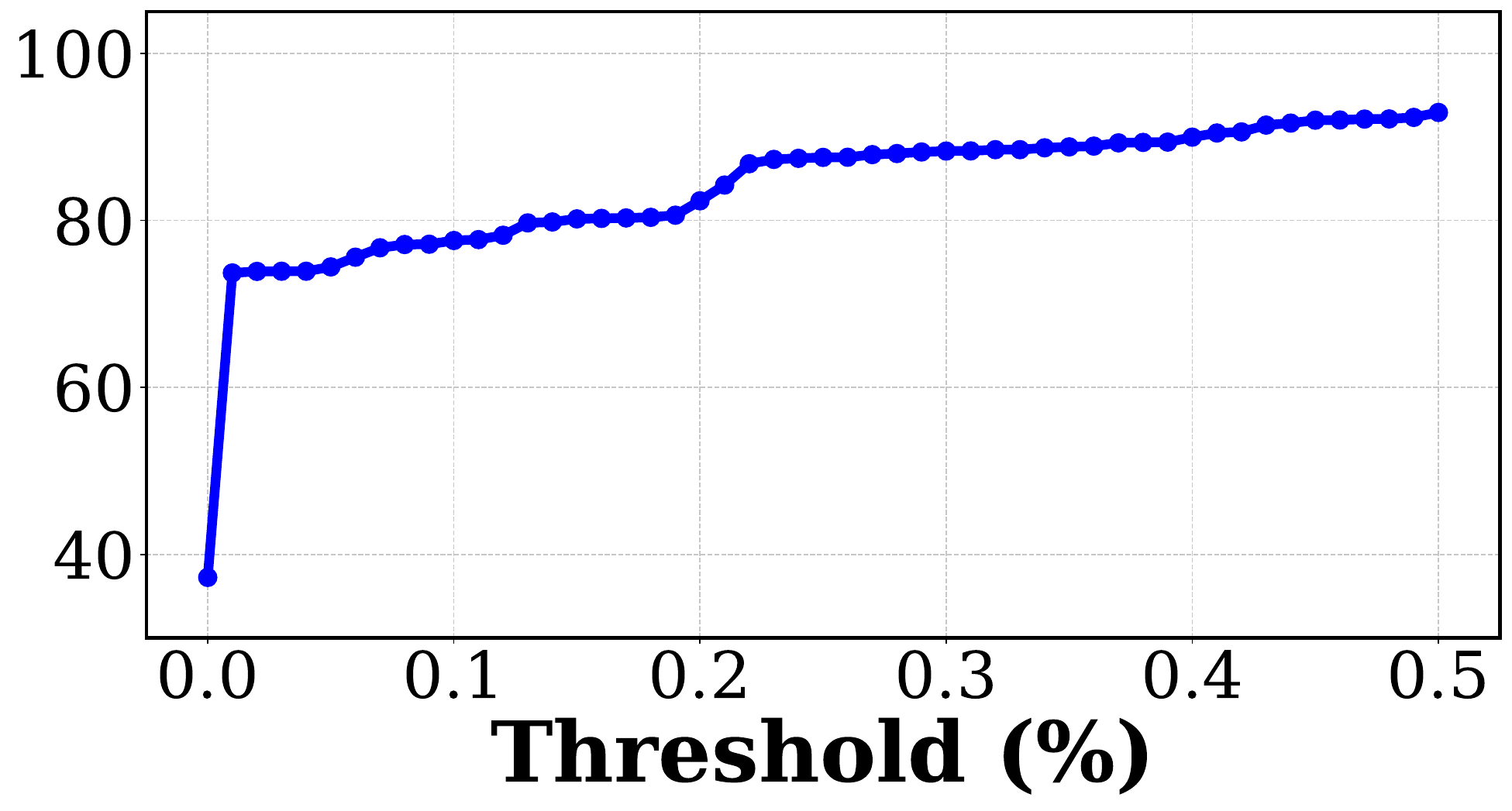}
        \caption{Diabetes}
        \label{fig:rq3d}
    \end{subfigure}
    \caption{Fraction of structurally robust Simpson's paradoxes vs.\ pruning threshold across four real-world datasets.}
    \label{fig:real-rq3}
\end{figure}

\subsubsection{Robustness of Redundant Groups}

Sibling child and separator equivalence rely on populations with identical coverage, which can be organized into convex coverage groups. Redundancy is therefore robust if coverage identicality persist under perturbations. We test robustness with two strategies:  
\begin{itemize}
    \item \textbf{Sibling child equivalence} (\Cref{prop:sibling-eq}): For sibling-child-equivalent paradoxes $p=(s_1,s_2,X,Y)$ and $p'=(s_1',s_2',X,Y)$, we randomly alter one attribute value in 5\% of the records in $\cov(s_1)\cup\cov(s_2)$. This tests whether $\cov(s_1)=\cov(s_1')$ and $\cov(s_2)=\cov(s_2')$ remain intact.  
    \item \textbf{Separator equivalence} (\Cref{prop:division-equivalence}): For separator-equivalent paradoxes $p=(s_1,s_2,X,Y)$ and $p'=(s_1,s_2,\\X',Y)$, we perturb 5\% of records in $\cov(s_1)\cup\cov(s_2)$ on attributes other than $X$ and $X'$. The one-to-one mapping $f:X \mapsto X'$ is preserved, and we test whether equivalence $\cov(s\substitute{X}{v}) = \cov(s\substitute{X'}{f(v)})$ persists despite changes in other attributes.  
\end{itemize}

As before, each perturbation is repeated 10{,}000 times, and robustness is measured as the survival rate of redundant paradox groups.  
Figure~\ref{fig:cov-sign-rq3} reports the results. The fraction of robust redundant groups rises with higher pruning thresholds, indicating that sibling- and division-equivalent paradoxes with larger coverage are more tolerant to perturbations.  

\begin{figure}[t]
    \centering
    \begin{subfigure}[b]{0.21\textwidth}
        \centering
        \includegraphics[width=\textwidth]{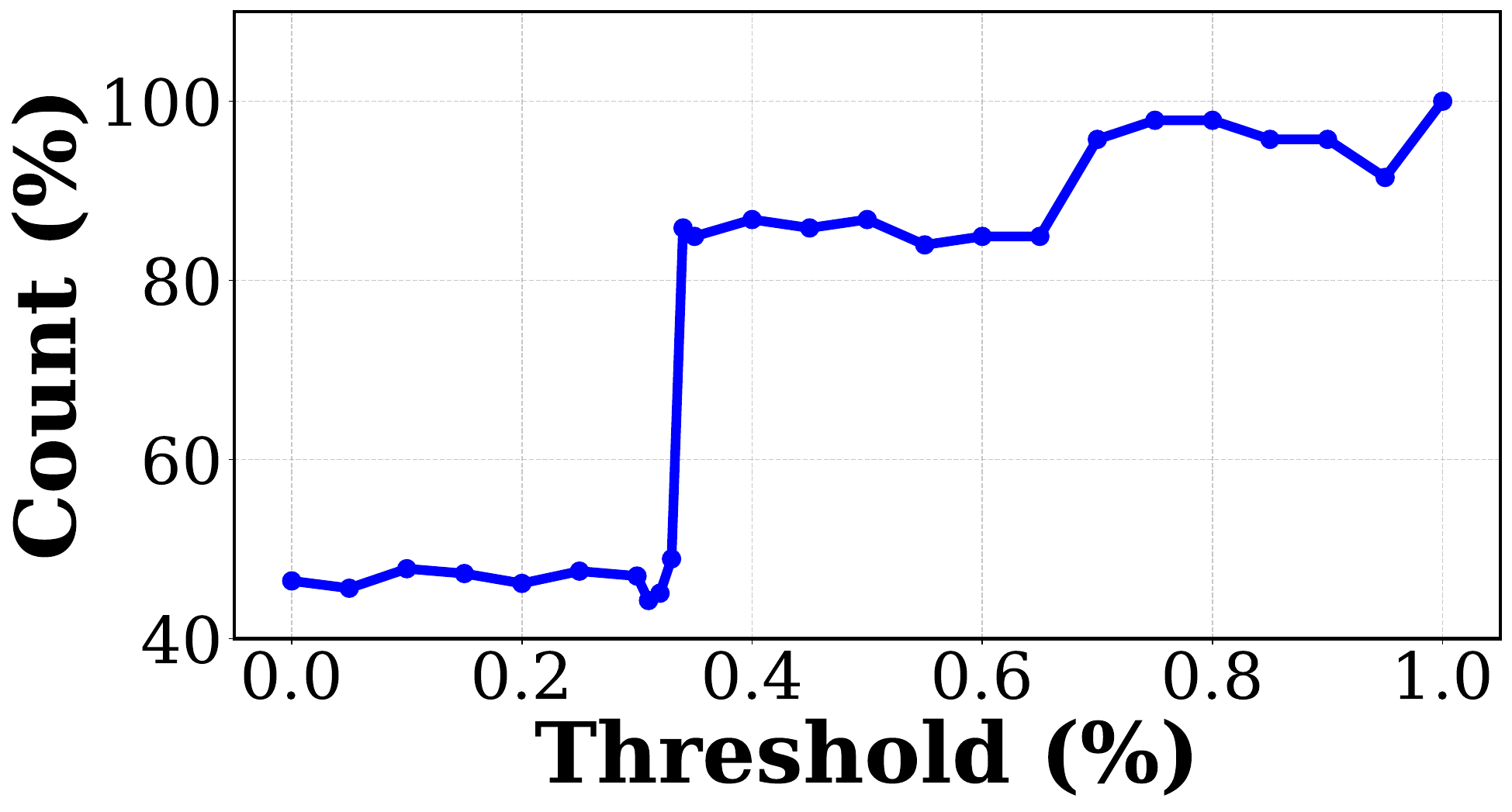}
        \caption{Adult}
        \label{fig:rq3a}
    \end{subfigure}
    \hspace{0.025\textwidth}
    \begin{subfigure}[b]{0.23\textwidth}
        \centering
        \includegraphics[width=\textwidth]{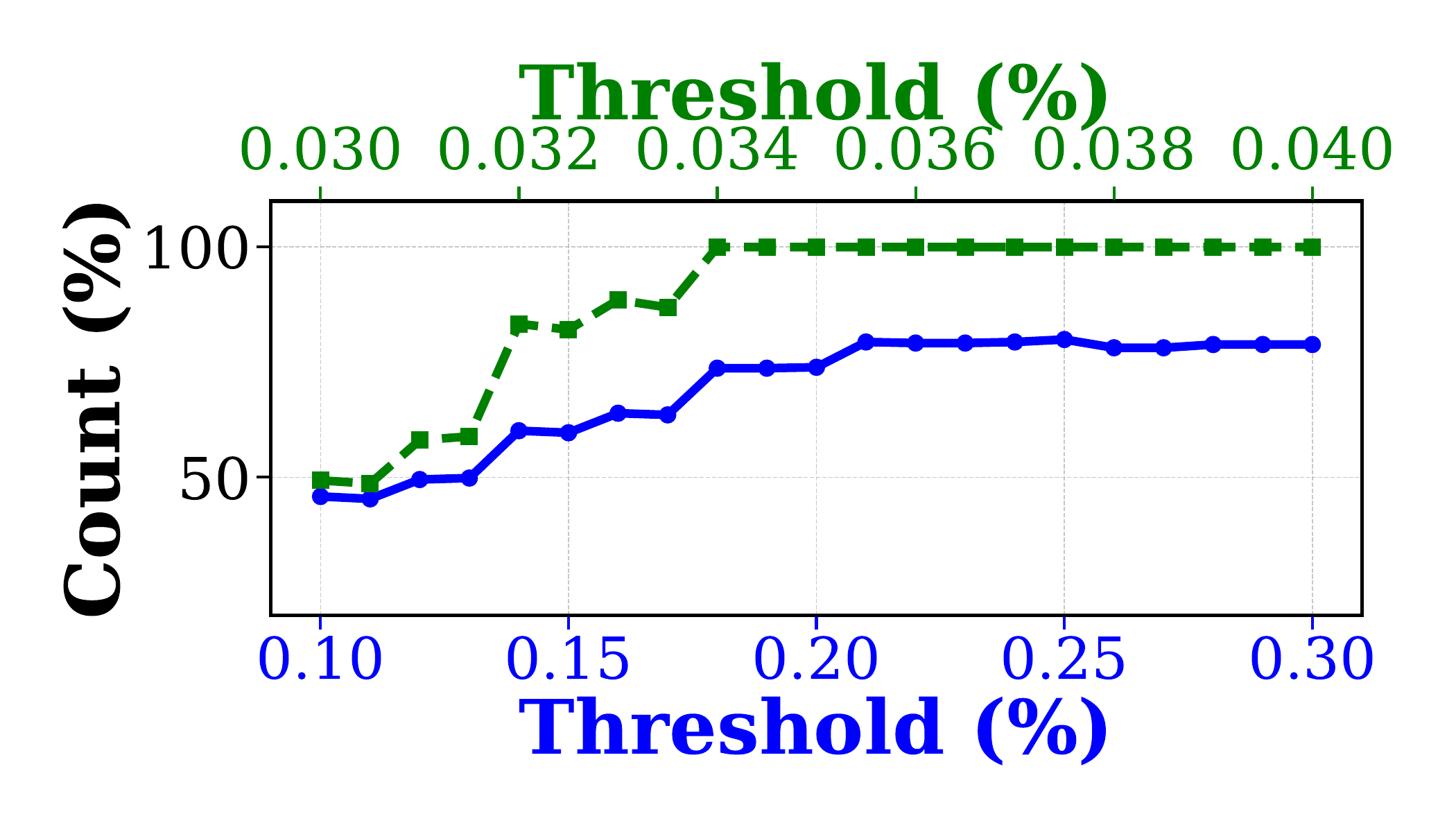}
        \caption{Mushroom}
        \label{fig:rq3b}
    \end{subfigure}

    \begin{subfigure}[b]{0.21\textwidth}
        \centering
        \includegraphics[width=\textwidth]{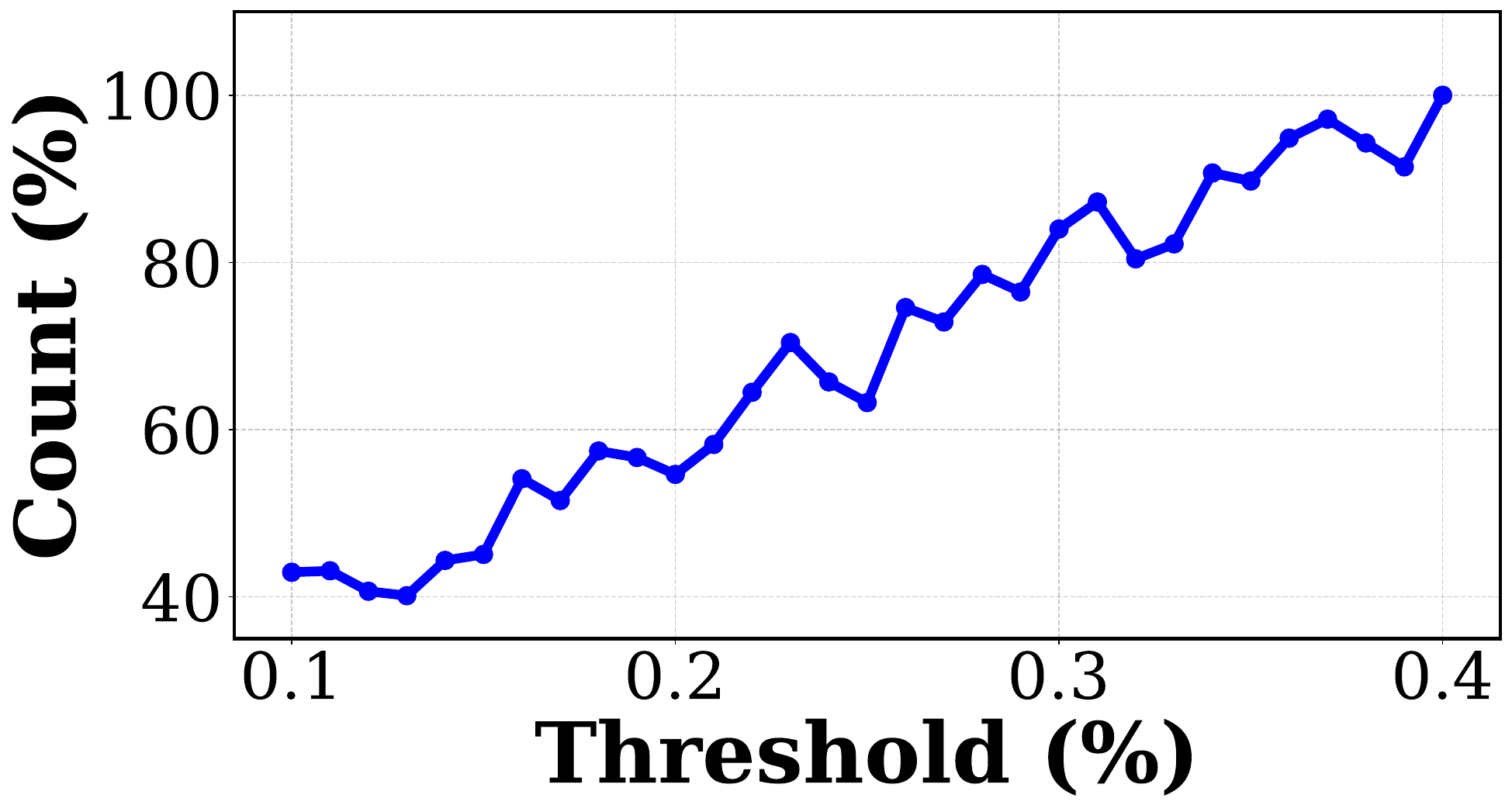}
        \caption{Loan}
        \label{fig:rq3c}
    \end{subfigure}
    \hspace{0.025\textwidth}
    \begin{subfigure}[b]{0.23\textwidth}
        \centering
        \includegraphics[width=\textwidth]{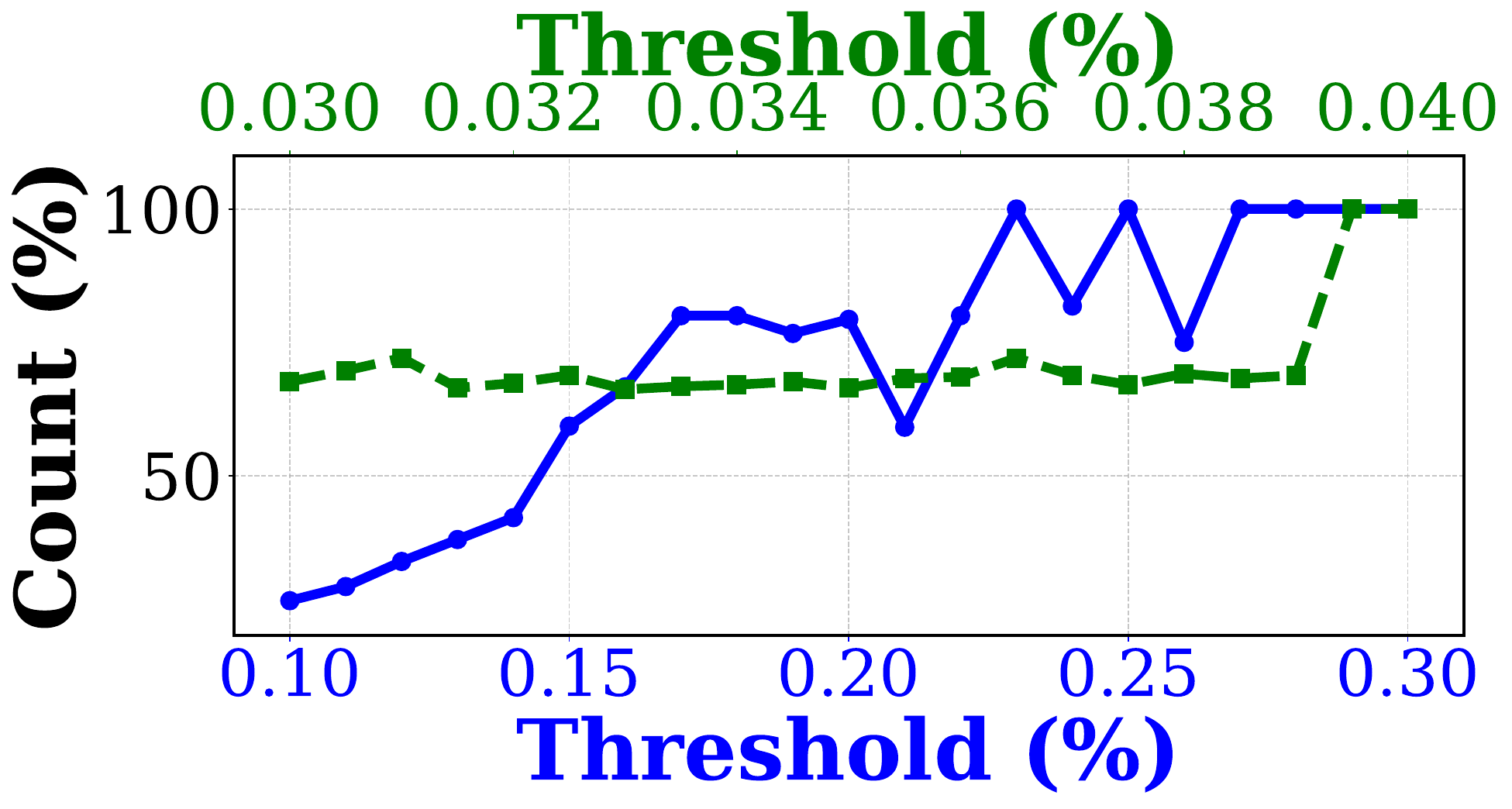}
        \caption{Diabetes}
        \label{fig:rq3d}
    \end{subfigure}
    \caption{Fraction of robust redundant paradox groups vs.\ pruning threshold across four real-world datasets.}
    \label{fig:cov-sign-rq3}
\end{figure}

\subsection*{Summary of Findings}

Across RQ1-RQ3, our experiments highlight three main insights.  
First, coverage-redundant Simpson's paradoxes are common: they represent a substantial fraction of all paradoxes in real-world datasets, with sibling child equivalence being the most frequent redundancy type.  
Second, our computational framework scales efficiently, achieving practical run times even on moderate- to high-dimensional data, and further benefits from population pruning that eliminate low-coverage populations early.  
Third, both individual paradoxes and their coverage redundancies exhibit structural robustness under data perturbations, confirming that these patterns reflect genuine properties of the data rather than random noise or collection errors.  

Overall, our findings show that coverage-redundant Simpson's paradoxes are prevalent, can be detected efficiently, and capture meaningful structural characteristics in real-world datasets.

\section{Related Work}
\label{sec:related}
To the best of our knowledge, this is the first work to study redundancy among Simpson's paradoxes. 
Our contribution connects to two lines of prior research: (1) methods for detecting Simpson's paradox in high-dimensional data; and (2) techniques for concisely summarizing data populations. 

\subsection{Detecting Simpson's Paradox}

Simpson’s paradox has been extensively studied since the introduction of Association Reversal (AR) and Amalgamation Paradox (AMP)~\cite{yule1903notes}. Early detection methods relied on statistical modeling. \citet{freitas2007integrating} constructed Bayesian networks to detect paradoxes by analyzing network structure, and \citet{fabris2006discovering} extended this approach by ranking paradoxes according to their estimated surprisingness for knowledge discovery. \citet{alipourfard2018can} introduced a statistical test comparing global trends with disaggregated sub-population patterns, while \citet{xu2018detecting} used Pearson correlation to detect reversals between continuous variables. \citet{sharma2022detecting} further extended this approach to categorical variables through binarization.

Beyond standalone methods, several works integrate paradox detection into broader analytical frameworks. Salimi et al.~\cite{salimi2018bias, salimi2018hypdb, salimi2020database} developed a system that identifies biased OLAP queries susceptible to Simpson’s paradox using independence tests and resolves them through automated query rewriting. Liu et al.~\cite{liu2011towards} proposed a data-driven framework that discovers sub-populations for hypothesis testing and reveals confounding factors underlying paradoxes. More recently, automated methods have been introduced: \citet{wang2023learning} employed neural models to disaggregate data and evaluate associations across subgroups, and \citet{jiang2025fedcfa} designed a federated learning framework that mitigates association reversals in distributed data via counterfactual learning. Domain-specific efforts include \citet{portela2019search}, who applied regression trees to identify conditional outliers affected by Simpson's paradox.

Despite these advances, existing methods emphasize global analyses and overlook paradoxes within \emph{local populations} defined by subspaces of the data. Such local paradoxes can still reveal structural patterns and are important for causal analysis and decision-making. The closest related work is \citet{xu2022finding}, who proposed a combinatorial search over all possible local subspaces. However, their algorithm (cf.\ Alg.\ref{alg:brute-force-materialization},\ref{alg:brute-force}) does not detect redundancy and is computationally expensive, repeatedly enumerating populations with identical coverage and evaluating redundant paradoxes.

\subsection{Summarizing Data Populations}

A cornerstone of our approach is identifying and organizing subsets of populations with identical coverage. This connects to the extensive literature on data cubes~\cite{10.5555/645481.655593, chen2020exploring, phan2016data, xie2016olap}, which address computational challenges in OLAP. \citet{kenneth1997fast} proposed efficient materialization algorithms for sparse multidimensional data using divide-and-conquer partitioning. \citet{beyer1999bottom} introduced bottom-up materialization with coverage-based pruning, which is closely related to our method (cf.\ Alg.~\ref{alg:materialization}). Other work has focused on cube condensation to reduce storage and improve query performance~\cite{wang2002condensed, sismanis2002dwarf}. More recently, \citet{john2022high} proposed partial materialization that reconstructs missing populations via linear programming, and \citet{you2025soc} developed an adaptive caching system that selectively materializes convex equivalence classes under memory constraints.  

Most directly relevant are quotient cubes~\cite{lakshmanan2002quotient, lakshmanan2003qc}, which partition the population lattice into equivalence classes defined by monotone aggregate functions (e.g., coverage, count, min, max). The quotient cube preserves roll-up and drill-down semantics, improving the efficiency of analytical queries such as \texttt{GROUP BY} and \texttt{CUBE BY}. While our method exploits similar structural properties -- specifically convex equivalence classes via coverage-based partitioning -- our focus is different. Unlike prior cube condensation techniques~\cite{beyer1999bottom, wang2002condensed, sismanis2002dwarf, you2025soc, kenneth1997fast}, which aim to optimize storage and query time, we leverage convex partitions of populations to identify and eliminate redundancy among Simpson's paradoxes.  

\section{Conclusions}
\label{sec:conclusion}

In this paper, we addressed the problem of redundancy in Simpson's paradox, a long-standing statistical phenomenon with broad applications in data analysis and causal inference. 
We showed that many paradoxes in multidimensional data are redundant, arising from populations with identical coverage or equivalent separator and label attributes. 
To resolve this issue, we formally defined three types of coverage redundancy, proved that redundancy forms an equivalence relation, and introduced a concise representation based on convexity properties of the population lattice. 
We further developed efficient algorithms that combine depth-first materialization, pruning, and redundancy-aware evaluation to discover all non-redundant Simpson's paradoxes. 
Experiments on both real-world and synthetic datasets demonstrated that redundant paradoxes are prevalent in practice, that our algorithms scale efficiently, and that the discovered paradoxes are structurally robust.  

Future work includes extending our framework to richer data types and continuous attributes, incorporating causal semantics to further refine redundancy definitions, and applying our methods in practical domains such as healthcare, finance, and social science where Simpson's paradox continues to pose challenges for interpretation and decision-making.



\bibliographystyle{ACM-Reference-Format}
\bibliography{main}
\appendix
\section{Synthetic Generation}

In this section, we introduce a few additional findings on coverage redundant Simpson's paradoxes from a data generative perspective. We first discuss the generation of separate instances of Simpson's paradox in Section~\ref{sec:sp-generation}. We then introduce the process of generating coverage redundant Simpson's paradoxes in Section~\ref{sec:cr-generation}. Finally, we summarize the overall data synthesization procedure in Section~\ref{sec:generative-procedure}.

\subsection{Generating Simpson's Paradoxes}
\label{sec:sp-generation}
Recall from Definition~\ref{def:simpson} that an association configuration (AC) $p = (s_1,s_2,X,Y)$ is a Simpson's paradox if:
\begin{align*}
    (1) \; &P(Y | s_1) \leq P(Y | s_2); \text{ and}\\
    (2) \; &P(Y | s_1\substitute{X}{v}) \geq P(Y | s_2\substitute{X}{v}), \; \forall v \in \Dom(X). 
\end{align*}
In particular, the frequency statistics $P(Y | s_j)$ for $j=1,2$ in condition (1), is obtained by weighted averaging their sub-population frequency statistics in condition (2), where the weights are determined by the relative coverage sizes of each sub-population. Specifically, we have that:
\begin{align*}
    P(Y | s_j) &= \sum_{v \in \Dom(X)} \frac{|\cov(s_j\substitute{X}{v})|}{|\cov(s_j)|} \cdot P(Y | s_j\substitute{X}{v}) \\
    &= \bm{Q}(s_j | X) \cdot \bm{P}(s_j | Y, X)^\top,
\end{align*}
where
\begin{multline*}
    \bm{Q}(s_j | X) = \left[ \frac{|\cov(s_j\substitute{X}{v_1})|}{|\cov(s_j)|}, \frac{|\cov(s_j\substitute{X}{v_2})|}{|\cov(s_j)|}, \right. \\ \left. \ldots, \frac{|\cov(s_j\substitute{X}{v_{|\Dom(X)|}})|}{|\cov(s_j)|} \right]
\end{multline*}
is the sample distribution of $s_j$ partitioned under $X$, and
\begin{multline*}
    \bm{P}(s_j | Y, X) = \left[ P(Y | s_j\substitute{X}{v_1}), P(Y | s_j\substitute{X}{v_2}), \right. \\ \left. \ldots, P(Y | s_j\substitute{X}{v_{|\Dom(X)|}}) \right]
\end{multline*}
is the frequency statistics of $s_j$'s sub-populations partitioned by $X$. With this, we can rephrase Definition~\ref{def:simpson} by substituting the terms, that $(s_1,s_2,X,Y)$ is a Simpson's paradox if:
%
%
%
\nop{
This gives us the following inequalities,
\begin{multline}
\label{eq:aggr-bound}
\min_{1 \leq j \leq |\Dom(X)|} \bm{P}(s_1 \mid Y_1, X)[j] \leq P(Y_1 \mid s_1) \\ \leq \max_{1 \leq j \leq |\Dom(X)|} \bm{P}(s_1 \mid Y_1, X)[j].
\end{multline}
Suppose $s_2$ being a sibling of $s_1$ with differential values $\{u_1,u_2\}$ under the differential attribute $X_{i_0}$ ($X_{i_0} \neq X$), and let $s$ denote their common parent, following Definition~\ref{def:simpson}, $(s, X_{i_0}, \{u_1,u_2\}, X, Y_1)$ is a Simpson's paradox if,
}
\begin{align}
    \bm{Q}(s_1 | X) \cdot \bm{P}(s_1 | Y, X)^\top < \bm{Q}(s_2 | X) \cdot \bm{P}(s_2 | Y, X)^\top; \text{ and} \label{eq:aggr-ineq} \\
    \bm{P}(s_1 | Y, X)[j] > \bm{P}(s_2 | Y, X)[j], \,  1 \leq j \leq |\Dom(X)|. \label{eq:sub-pop-aggr}
\end{align}
Therefore, the essence of generating an instance of Simpson's paradox is to \emph{(a)} find the set of sub-population frequency statistics $\bm{P}(s_1 | Y, X)$ and $\bm{P}(s_2 | Y, X)$ that satisfy inequality (\ref{eq:sub-pop-aggr}); and \emph{(b)} solve for the sample distributions $\bm{Q}(s_1 | X)$ and $\bm{Q}(s_2 | X)$ that satisfy inequality (\ref{eq:aggr-ineq}). We discuss each in the following paragraphs.

\paragraph{Sub-population Frequency Statistics}
\label{sec:sub-pop-aggr-gen}
The first step of generating an instance of Simpson's paradox is to ensure that each sub-population in $s_2$, partitioned by $X$, has an frequency statistics value smaller than its sibling sub-population in $s_1$ (inequality~(\ref{eq:sub-pop-aggr})). A simple pattern that achieves this is:
\begin{multline}
    \bm{P}(s_2 | Y, X)[1] < \bm{P}(s_1 | Y, X)[1] < \\
     \bm{P}(s_2 | Y, X)[2] < \bm{P}(s_1 | Y, X)[2] < \\
     \cdots \\
     \bm{P}(s_2 | Y, X)[|\Dom(X)|] < \\ 
     \bm{P}(s_1 | Y, X)[|\Dom(X)|]. \label{eq:sub-aggr-pair-pattern}
\end{multline}
This pattern ensures that for any value $v_j$ $(1 \leq j \leq |\Dom(X)|)$, we have $\bm{P}(s_1 | Y, X)[j] > \bm{P}(s_2 | Y, X)[j]$, satisfying inequality~(\ref{eq:sub-pop-aggr}). With this pattern, we can now focus on finding sample distributions that satisfy inequality (\ref{eq:aggr-ineq}).

\paragraph{Sample Distributions}
\label{sec:sample-dist-gen}
\nop{
Suppose $\bm{P}(s_1 \mid Y_1, X)$ and $\bm{P}(s_2 \mid Y_1, X)$ follow the pattern in (\ref{eq:sub-aggr-pair-pattern}), 
then to satisfy inequality (\ref{eq:aggr-ineq}), 
$\bm{Q}(s_1 \mid X)$ and $\bm{Q}(s_2 \mid X)$ can be obtained by simply assigning,
$$\bm{Q}(s_1 \mid X)[j] = \begin{cases}
    1 \mid j = 1 \\
    0 \mid \text{otherwise }
\end{cases} 1 \leq j \leq |\Dom(X)|,$$
$$\bm{Q}(s_2 \mid X)[j] = \begin{cases}
    1 \mid j = |\Dom(X)| \\
    0 \mid \text{otherwise }
\end{cases} 1 \leq j \leq |\Dom(X)|.$$
However, this would result in $cov(s_1) = cov(s_1[X = v_1])$ and $cov(s_2) = cov(s_2[X = v_{|\Dom(X)|}])$.
As such, $X$ no longer provides a valid sub-population partition for $s_1$ and $s_2$.
As an alternative approach to ensure non-empty assignment of sub-populations, we can aim to make the sample distributions as uniform as possible while maintaining inequality (\ref{eq:aggr-ineq}). 
While various objective functions could be used, here we present one possible formulation where $\bm{Q}(s_m \mid X)$ for $m=1,2$ can be obtained by solving the following convex quadratic program (QP):
}
Given the sub-population frequency statistics pattern established above, we now need to solve for sample distributions $\bm{Q}(s_1 | X)$ and $\bm{Q}(s_2 | X)$ that satisfy inequality~(\ref{eq:aggr-ineq}). To achieve this, we formulate the problem as a quadratic program:
\begin{equation} \label{eq:qp-convex-pair}
\begin{aligned}
    \text{minimize} \quad & \sum_{j=1}^2 \left\| \bm{Q}(s_j | X) - \frac{1}{|\Dom(X)|}\mathbf{1} \right\|_2^2 \\
    \text{subject to} \quad & \text{(i)} \quad \sum_{k=1}^{|\Dom(X)|} \bm{Q}(s_j | X)[k] = 1, \quad j = 1, 2; \\
    & \text{(ii)} \quad \bm{Q}(s_j | X)[k] > 0, \quad j = 1,2 \\
    & \phantom{(i)} \quad \text{and } 1 \leq k \leq |\Dom(X)|; \\
    & \text{(iii)} \quad \bm{Q}(s_2 | X) \cdot \bm{P}(s_2 | Y, X)^\top \\
    & \phantom{(i)} \quad > \bm{Q}(s_1 | X) \cdot \bm{P}(s_1 | Y, X)^\top.
\end{aligned}
\end{equation}
We incorporate inequality (\ref{eq:aggr-ineq}) as a linear constraint specified in condition (iii). Moreover, the objective function of the QP aims to minimize the squared distance between each sample distribution and the uniform distribution $\frac{1}{|\Dom(X)|}\mathbf{1}$ to promote uniformity. Our synthetic generator also supports optimizing towards other distribution patterns such as normal or Zipfian distributions to better mimic the statistics of real-world data. 
%

\paragraph{Synthesis of Simpson's Paradox}
\label{para:gen-one-sp}
To summarize, generating an instance of Simpson's paradox requires establishing sub-population frequency statistics satisfying inequality~(\ref{eq:sub-pop-aggr}) and solving for sample distributions that produce the reversal effect outlined in inequality~(\ref{eq:aggr-ineq}). Algorithm~\ref{alg:gen-sp-separate} formalizes this generation procedure, taking as input an AC $(s_1,s_2,X,Y)$ and crucially, a size parameter $U$ that controls the number of records covered by the generated Simpson's paradox. The algorithm proceeds in three main steps: first, it generates the sub-population frequency statistics following the pattern in Equation~(\ref{eq:sub-aggr-pair-pattern}); second, it solves the quadratic program in~(\ref{eq:qp-convex-pair}) to obtain optimal sample distributions; and third, it populates the output table with $2 \cdot U$ records for the Simpson's paradox $(s_1,s_2,X,Y)$, distributing records across sub-populations according to the sample distributions and assigning label values according to probabilities given by the sub-population frequency statistics.

\begin{algorithm}
\caption{Generate non-redundant Simpson's paradox.}
\label{alg:gen-sp-separate}
\begin{algorithmic}[1]
\Input An AC $(s_1,s_2,X,Y)$, paradox size $U$
\Out Data records $T$ for the Simpson's paradox $(s_1,s_2,X,Y)$
\State Obtain the sub-population frequency statistics $\bm{P}(s_1 | Y,X)$, $\bm{P}(s_2 | Y,X)$ following Equation~(\ref{eq:sub-aggr-pair-pattern});
\State Obtain the sample distributions $\bm{Q}(s_1 | X)$, $\bm{Q}(s_2 | X)$ by solving the quadratic program (QP) in~(\ref{eq:qp-convex-pair});
\State \textcolor{gray}{\textit{// Populate data records following the obtained sample distributions and sub-population aggregate statistics}}
\For{each $1 \leq k \leq |\Dom(X)|$ and each $j \in [1, 2]$}
\State \textcolor{gray}{\textit{// Find the number of records for each sub-population}}
\State Let $U_{j, k} \gets U \cdot \bm{Q}(s_j | X)[k]$; 
\State \textcolor{gray}{\textit{// Assign labels according to sub-population aggr. stats.}}
\State Add $U_{j, k}$ copies of $s_j\substitute{X}{v_k}$ as records to $T$ and assign $U_{j, k} \cdot \bm{P}(s_j | Y, X)[k]$ of them with $(Y = 1)$;
\EndFor
\State \Return $T$.
\end{algorithmic}
\end{algorithm}

\subsection{Realizing Coverage Redundancies}
\label{sec:cr-generation}
Having established a method for generating individual instances of Simpson's paradox, we now discuss creating (coverage) redundant Simpson's paradoxes. Building on the discussions in Section~\ref{sec:equivalence-types} and  Definition~\ref{def:coverage}, (coverage) redundancies fundamentally arise when distinct populations have identical coverage within the data. To systematically realize (coverage) redundancies, we must understand the conditions under which populations would share the same coverage, as this property serves as the foundation for realizing both sibling child and separator equivalences. To this end, we remark on the following proposition.
\begin{proposition}[Impossibility of coverage identicality]
\label{prop:gen-cov-eq}
Let $T$ be a base table with $n$ categorical attributes $\{X_1, \ldots, X_n\}$. If the set of unique records in $T$ corresponds exactly to all possible combinations of attribute values (\emph{i.e.,} the complete Cartesian product $\prod_{i=1}^{n}\Dom(X_i)$), then no two distinct populations of $T$ share the same coverage.
\proof
Assume, by contradiction, that there exist two distinct populations $s$ and $s'$ with $\cov(s) = \cov(s')$. Since $s \neq s'$, let $X_{k_0}$ be an attribute for which $s[k_0] \neq s'[k_0]$. Without loss of generality, assume $s[k_0] = v$ for a fixed value $v \in \Dom(X_{k_0})$ and $s'[k_0] = \ast$.

Let us consider two records, $r_1$ and $r_2$, which are identical on all attributes except $X_{k_0}$. For $r_1$, let $r_1.X_{k_0} = v$, and for $r_2$, let $r_2.X_{k_0} = v'$ where $v' \neq v$ and $v' \in \Dom(X_{k_0})$. For all other attributes $X_j$ where $j \neq k_0$, if $s[j] \neq \ast$, then $r_1.X_j = r_2.X_j = s[j]$. Observe that both records $r_1$ and $r_2$ must exist in $T$ because $T$ contains the complete Cartesian product of all attribute domains.

Now, $r_1$ is covered by both $s$ and $s'$. However, record $r_2$ is covered by $s'$ (because $s'[k_0] = \ast$) but not by $s$ (because $s[k_0] = v$ but $r_2.X_{k_0} = v'$, $v' \neq v$).

This means $r_2 \in \cov(s')$ but $r_2 \notin \cov(s)$, which contradicts our assumption that $\cov(s) = \cov(s')$. Therefore, no two distinct populations can share the same coverage when $T$ contains the complete Cartesian product of all attribute domains.
\qed
\end{proposition}
The contrapositive of Proposition \ref{prop:gen-cov-eq} implies that populations sharing identical coverage can only exist when the dataset contains a proper subset of the complete Cartesian product of attribute domains. Therefore, to facilitate the generation of coverage redundant Simpson's paradoxes, we impose a size threshold $t$ that is significantly smaller than $\prod_{i=1}^{n}|\Dom(X_i)|$ to constrain the number of unique records in the generated dataset.
%

Having established the condition for realizing populations with identical coverage, we now proceed to develop methods for realizing each of the three types of coverage equivalences: sibling child equivalence, separator equivalence, and statistic equivalence.
\paragraph{Sibling Child Equivalence}
Suppose we have a set of data records $T$ that produces a Simpson's paradox $p_1 = (s_{1},s_{2},X,Y)$, where $s_{1} = s\substitute{X_0}{u_1}$ and $s_{2} = s\substitute{X_0}{u_2}$ are siblings from a common parent $s$. To realize sibling equivalence, our goal is update the records in $T$ such that they also produce another Simpson's paradox $p_2 = (s_{1}',s_{2}',X,Y)$, where $s_{1}' = s'\substitute{X_0'}{v_1}$ and $s_{2}' = s'\substitute{X_0'}{v_2}$ are siblings from a common parent $s'$, and that $p_2$ is sibling child equivalent to $p_1$. According to Definition~\ref{def:coverage}, sibling child equivalence requires $\cov(s_{1}) = cov(s_{1}')$ and $\cov(s_{2}) = \cov(s_{2}')$. Based on the relationship between $(s, s_{1}, s_{2})$ and $(s', s_{1}', s_{2}')$, we have three scenarios:
\begin{enumerate}
    \item \textbf{Scenario 1}: $s \neq s'$, $X_{0} = X_{0}'$, and $\{u_1,u_2\} = \{v_1,v_2\}$. In this case, sibling child equivalence is achieved by ensuring $\cov(s) = \cov(s')$. To this, for each categorical attribute $X_k$ where $s[k] \neq \ast$ or $s'[k] \neq \ast$, we update every record $r$ in $\cov(s)$ (within $T$), such that $r.X_k = s[k]$ or $r.X_k = s'[k]$. If $s[k] = s'[k] = \ast$, then no update is needed for $r.X_k$.
    \item \textbf{Scenario 2}: $s = s'$, $X_{0} \neq X_{0}'$, and $\{u_1,u_2\}\neq\{v_1,v_2\}$. In this case, to achive sibling child equivalence, we establish a one-to-one mapping $f : \{u_1,u_2\} \mapsto \{v_1,v_2\}$ such that $f(u_1) = v_1$ and $f(u_2) = v_2$. For each record $r$ in $T$, we set $r.X_{0}' = f(r.X_{0})$ when $r.X_{0} \in \{u_1,u_2\}$. This ensures $\cov(s\substitute{X_0}{u_k}) = \cov(s\substitute{X_0'}{f(u_k)})$ for $k=1,2$. 
    \item \textbf{Scenario 3}: $s \neq s'$, $X_{0} \neq X_{0}'$, and $\{u_1,u_2\}\neq\{v_1,v_2\}$. This combines the previous senarios. To achieve sibling child equivalence, we first ensure $\cov(s) = \cov(s')$ as in Scenario 1, then establish the one-to-one mapping as in Scenario 2.
\end{enumerate}
%
\begin{example}
Consider the data records in a slightly perturbed version of Table~\ref{tab:ex2} where attribute values in $D$ are randomized. Supposed the perturbed Table~\ref{tab:ex2} is populated as a result of generating the Simpson's paradox $p_1 = ((*, b_1, *, *), (*, b_2, *, *), A, Y_1)$. To create a sibling equivalent Simpson's paradox $p_2 = ((*, *, *, d_1), (*, *, *, d_2), A, Y_1)$, we apply both Scenarios 1 and 2. 

For Scenario 1, the parent population is $(\ast,\ast,\ast,\ast)$ for both $p_1$ and $p_2$, which are identical, so no adjustment to records in the perturbed table is needed. 

For Scenario 2, we define a one-to-one mapping $f: \{b_1, b_2\} \rightarrow \{d_1, d_2\}$ where $f(b_1) = d_1$ and $f(b_2) = d_2$. We then update each record in the perturbed table with $B = b_1$ gets $D = d_1$ and each record with $B = b_2$ gets $D = d_2$. This establishes that $cov((*, b_1, *, *)) = cov((*, *, *, d_1))$ and $cov((*, b_2, *, *)) = cov((*, *, *, d_2))$.

In this way, we make $p_1$ and $p_2$ to be sibling equivalent, as verified in Example~\ref{ex:sibling}.
\qed
\end{example}
Algorithm~\ref{alg:sibling-equivalence} formalizes this process of generating sibling-child-equivalent Simpson's paradoxes.
\begin{algorithm}
\caption{Realizing sibling child equivalence.}
\label{alg:sibling-equivalence}
\begin{algorithmic}[1]
\Input Data records $T$ producing the Simpson's paradox $p_1 = (s_{1},s_{2},X,Y)$ where $s_{1} = s\substitute{X_0}{u_1}$ and $s_{2} = s\substitute{X_0}{u_2}$, sibling populations $(s_{1},s_{2})$ where $s_{1}' = s\substitute{X_0'}{v_1}$ and $s_{2}' = s\substitute{X_0'}{v_2}$
\Out Updated data records $T$ producing a sibling-child-equivalent Simpson's paradox $p_2 = (s_{1}',s_{2}',X,Y)$
\State \textcolor{gray}{\textit{// Scenario 1: Ensure $\cov(s) = \cov(s')$}}
\For{each record $r \in T$ s.t. $r \in \cov(s)$}
\For{each attribute $X_k$ s.t. $s[k] \neq \ast$ or $s'[k] \neq \ast$}
\State Set $r.X_k \gets s[k]$ if $s[k] \neq \ast$, else $r.X_k \gets s'[k]$;
\EndFor
\EndFor
\State \textcolor{gray}{\textit{// Scenario 2: Establish the one-to-one mapping}}
\State Establish the mapping $f$ where $f(u_j) = v_j$ for $j=1,2$;
\For{each record $r \in T$}
\State Set $r.X_{0}' \gets f(r.X_{0})$ if $r.X_{0} \in \{u_1,u_2\}$.
\EndFor
\end{algorithmic}
\end{algorithm}

\paragraph{Separator Equivalence}
Recall from Proposition~\ref{prop:division-equivalence} and Definition~\ref{def:coverage}, Simpson's paradoxes $p_1 = (s_1,s_2,X_{1},Y)$ and $p_2 = (s_1,s_2,X_{1}',Y)$ are separator equivalent if there exists a one-to-one mapping $f$ between $\Dom(X_{1})$ and $\Dom(X_{1}')$ such that for every $v \in \Dom(X_{1})$ and $j=1,2$,
$$\cov(s_j\substitute{X_1}{v}) = \cov(s_j\substitute{X_1'}{f(v)}).$$
To achieve this, for every record $r$ in $T$, we set $r.X_{1}' = f(r.X_{1})$, which we formalize the process in Algorithm~\ref{alg:division-equivalence}.
\begin{example}
\label{ex:gen-division}
Consider a perturbed version of Table~\ref{tab:ex2} where attribute values in $C$ are initially randomized. Suppose the perturbed Table~\ref{tab:ex2} is populated as a result of generating the Simpson's paradox $p_1 = ((*, b_1, *, *), (*, b_2, *, *), A, Y_1)$. To create a separator equivalent Simpson's paradox $p_2 = ((*, b_1, *, *), (*, b_2, *, *), C, Y_1)$, we establish a one-to-one mapping $f$ between $\Dom(A) = {a_1, a_2}$ and $\Dom(C) = {c_1, c_2}$ where $f(a_1) = c_1$ and $f(a_2) = c_2$. We then update each record in the perturbed table so that whenever $A = a_1$, we set $C = c_1$, and whenever $A = a_2$, we set $C = c_2$. This ensures that $\cov((*, b_1, *, *)\substitute{A}{a_k}) = \cov((*, b_1, *, *)\substitute{C}{c_k})$ and $\cov((, b_2, *, )\substitute{A}{a_k}) = \cov((, b_2, *, *)\substitute{C}{c_k})$ for $k \in \{1, 2\}$. As verified in Example~\ref{ex:division}, $p_1$ and $p_2$ are separator equivalent.
\qed
\end{example}
\nop{
Similarly, to generalize, one can establish a sequence of $m - 1$ one-to-one mappings $f_j$ (for $1 \leq j \leq m - 1$) with respect to a set of $m$ separator attributes $\{X_{i_{1_1}}, \ldots, X_{i_{1_m}}\}$ such that for each $j$ from 1 to $m - 1$, for each $k = 1, 2$, and for every value $v \in \Dom(X_{i_{1_j}})$, 
\begin{equation}
\label{eq:general-division-mapping}
cov(s[X_{i_0} = u_k, X_{i_{1_j}} = v]) = cov(s[X_{i_0} = u_k, X_{i_{1_{j+1}}} = f_j(v)]).
\end{equation}
To achieve this, we first require that \( |\Dom(X_{i_{1_j}})| = \kappa \) for all \( 1 \leq j \leq m \), where \( \kappa \geq 2 \). 
Furthermore, we assume an implicit order such that, for each \( \Dom(X_{i_{1_j}}) \), \( \Dom(X_{i_{1_j}}) = (v_{1_j}, \ldots, v_{\kappa_j}) \), where \( \Dom(X_{i_{1_j}})[k] = v_{k_j} \) for \( 1 \leq k \leq \kappa \).
We outline the process of constructing the sequence of one-to-one mappings that satisfy Equation~(\ref{eq:general-division-mapping}) in Algorithm~\ref{alg:division-equivalence}, which produces coverage redundant set of Simpson's paradoxes \( \{(s, X_{i_0}, \{u_1, u_2\}, X_{i_{1_j}}, Y)\}_{j=1}^m \).
}

\begin{algorithm}
\caption{Realizing separator equivalence.}
\label{alg:division-equivalence}
\begin{algorithmic}[1]
\Input Set of data records $T$ producing the Simpson's paradox $(s_1,s_2,X_{1},Y)$, a separator attribute $X_{1}'$
\Out Updated set of data records $T$ producing a separator equivalent Simpson's paradox $p_2 = (s_1,s_2,X_{1}',Y)$
\State Let $f:\Dom(X_{1}) \mapsto \Dom(X_{1}')$ be the one-to-one map;
\For{each record $r \in T$}
\State Set $r.X_{1}' \gets f(r.X_{1})$.
\EndFor
\end{algorithmic}
\end{algorithm}

\paragraph{Statistic Equivalence}
Recall from Proposition~\ref{prop:statistics-equivalence} and Definition~\ref{def:coverage}, Simpson's paradoxes $p_1 = (s_1, s_2, X, Y_2)$ and $p_2 = (s_1, s_2, X, Y_{2}')$ are statistic equivalent if for each $s_j$ ($j = 1, 2$) $P(Y_{2} | s_j) = P(Y_{2}' | s_j)$, and for every value $v \in \Dom(X)$, $P(Y_{2}|s_j\substitute{X}{v}) = P(Y_{2}'|s_j\substitute{X}{v})$. To achieve this, we simply ensure that each record has identical values for both label attributes $Y_{2}$ and $Y_{2}'$. We formalize this process in Algorithm~\ref{alg:statistics-equivalence}.
\begin{example}
\label{ex:gen-statistics-equivalence}
Consider a perturbed version of Table~\ref{tab:ex2} where attribute values in $Y_2$ are initially randomized. Suppose the perturbed Table~\ref{tab:ex2} is populated as a result of generating the Simpson's paradox $p_1 = ((*, b_1, *, *), (*, b_2, *, *), A, Y_1)$. To create a statistic equivalent Simpson's paradox $p_2 = ((*, b_1, *, *), (*, b_2, *, *), A, Y_2)$, we update each record in the perturbed table so that $Y_2 = Y_1$ for all records. This ensures that $P(Y_1|s_j) = P(Y_2|s_j)$ and $P(Y_1|s_j\substitute{A}{a_k}) = P(Y_2|s_j\substitute{A}{a_k})$ for $j \in \{1, 2\}$ and $k \in \{1, 2\}$. As verified in Example~\ref{ex:statistics}, $p_1$ and $p_2$ are statistic equivalent.
\qed
\end{example}
%
%
\begin{algorithm}
\caption{Realizing statistic equivalence.}
\label{alg:statistics-equivalence}
\begin{algorithmic}[1]
\Input Data records $T$ producing the Simpson's paradox $p_1 = (s_1,s_2,X, Y_{2})$, a label attribute $Y_{2}'$
\Out Updated data records $T$ producing a statistic equivalent Simpson's paradox $p_2 = (s_1,s_2,X, Y_{2}')$
\For{each record $r \in T$}
\State Set $r.Y_{2}' \gets r.Y_{2}$.
\EndFor
\end{algorithmic}
\end{algorithm}

\subsection{Data Generation Workflow}
\label{sec:generative-procedure}
\begin{algorithm}
\caption{Generate redundant Simpson's paradoxes.}
\label{alg:generating-workflow}
\begin{algorithmic}[1]
\Input Categorical attributes $\{X_i\}_{i=1}^n$, label attributes $\{Y_j\}_{j=1}^m$, size threshold $t \ll \prod_{i=1}^n |\Dom(X_i)|$
\Out Data table $T$ (initially empty)
\State Let $R \gets \emptyset$ to collect the set of unique data records;
\State Let $P \gets \emptyset$ to collect the set of generated Simpson's paradoxes;
\While{$|R| < t$}
\State Let $p_1 = (s_{1},s_{2},X_{1},Y_{3})$ be an AC not in $P$;
\State \textcolor{gray}{\textit{// Step 1: Generate distinct Simpson's paradox}}
\State Populate $T'$ for the Simpson's paradox $p_1$ using Alg.~\ref{alg:gen-sp-separate};
\State \textcolor{gray}{\textit{// Step 2: Introduce coverage redundancies}}
\State \textcolor{gray}{\textit{// Sibling child equivalence}}
\State Apply Alg.~\ref{alg:sibling-equivalence} to $T'$ to create a sibling-child-equivalent Simpson's paradox $p_{2} = (s_{1}',s_{2}',X_{1},Y_{2})$;
\State \textcolor{gray}{\textit{// Separator equivalence}}
\State Apply Alg.~\ref{alg:division-equivalence} to $T'$ to create a separator equivalent Simpson's paradox $p_{3} = (s_{1},s_{2},X_{1}',Y_{3})$;
\State \textcolor{gray}{\textit{// Statistic equivalence}}
\State Apply Alg.~\ref{alg:statistics-equivalence} to $T'$ to create a statistic equivalent Simpson's paradox $p_{4} = (s_{1},s_{2},X_{1},Y_{2}')$;
\State Add $p_1$, $p_{2}$, $p_{3}$, and $p_{4}$ to $P$;
\State Add $T'$ to $T$ and unique records of $T'$ to $R$;
\EndWhile
\State \Return $T$.
\end{algorithmic}
\end{algorithm}

Building upon the techniques established in Sections~\ref{sec:sp-generation} and~\ref{sec:cr-generation}, we formulate a systematic approach for synthetic data generation that integrates both individual Simpson's paradox generation and coverage redundancy realization. The process employs a two-phase strategy: first generating distinct instances of Simpson's paradoxes (Section~\ref{sec:sp-generation}), then systematically introducing coverage redundancies through sibling child, separator, and statistics equivalences (Section~\ref{sec:cr-generation}). These phases are iterated until reaching a specified threshold $t \ll \prod_{i=1}^{n}|\Dom(X_i)|$ of unique records populated, which per Proposition~\ref{prop:gen-cov-eq} ensures the dataset contains populations with identical coverage necessary for redundancy.

Algorithm~\ref{alg:generating-workflow} formalizes this process, taking categorical attributes $\{X_i\}_{i=1}^n$, label attributes $\{Y_j\}_{j=1}^m$, and the size threshold $t$ as input, and producing a synthetic data table containing groups of (coverage) redundant Simpson's paradoxes.

\section{Missing Proofs}

In this section, we present missing proofs of lemmas, propositions, and theorems presented in \Cref{sec:coverage} and \Cref{sec:finding}.

\subsection{Proofs of Redundancy Properties}

\textsc{Lemma~\ref{prop:sibling-eq} (Sibling child equivalence)} \textit{Consider two association configurations $p = (s_1, s_2, X, Y)$ and $p' = (s'_1, s'_2, X, Y)$
where $\cov(s_1) = \cov(s'_1)$ and $\cov(s_2) = \cov(s'_2)$.
If $p$ is a Simpson's paradox, then $p'$ is also a Simpson's paradox.}

\begin{proof}
Since $p$ is a Simpson's paradox, $P(Y | s_{1}) > P(Y_ | s_{2})$.
Due to $cov(s_{j})=cov(s_{j}')$ $(j = 1,2)$, $P(Y | s_{j}) = P(Y | s_{j}')$.
Therefore, $P(Y | s_{1}') > P(Y | s_{2}')$.
Furthermore, for every $v \in \Dom(X)$, we have that $\cov(s_{j}\substitute{X}{v}) = \cov(s_{j}'\substitute{X}{v})$ $(j=1,2)$, implying $P(Y | s_{j}\substitute{X}{v}) = P(Y_{i}' | s_{j}'\substitute{X}{v})$ $(j=1,2)$.
Hence, for every $v \in \Dom(X)$, $P(Y | s_{1}'\substitute{X}{v}) \leq P(Y | s_{2}'\substitute{X}{v})$.
It follows, from Def.~\ref{def:simpson}, that $p'$ is a Simpson's paradox.
\end{proof}

\textsc{Lemma~\ref{prop:division-equivalence} (Separator equivalence)} \textit{Consider two association configurations $p = (s_1, s_2, X, Y)$ and $p' = (s_1, s_2, X', Y)$,
where $X \neq X'$ and there exists a one-to-one mapping $f : \Dom(X) \mapsto \Dom(X')$ such that
for every $v \in \Dom(X)$ and $s \in \{s_1, s_2\}$, $\cov(s\substitute{X}{v}) = \cov(s\substitute{X'}{f(v)})$.
If $p$ is a Simpson's paradox, then $p'$ is also a Simpson's paradox.}

\begin{proof}
Since $p$ is a Simpson's paradox, $P(Y | s_1) > P(Y | s_2)$.
The populations $s_1$ and $s_2$ remain the same in $p'$, so this inequality holds for $p'$ as well.
In addition, for every value $v \in \Dom(X)$, $P(Y | s_1\substitute{X}{v}) \leq P(Y | s_2\substitute{X}{v})$.
Due to the one-to-one mapping $f$, for every value $v \in \Dom(X)$, $P(Y | s_j\substitute{X}{v}) = P(Y | s_j\substitute{X'}{f(v)})$ $(j = 1, 2)$. 
Thus, $P(Y | s_1\substitute{X'}{f(v)}) \leq P(Y | s_2\substitute{X'}{f(v)})$, $\forall v \in \Dom(X)$. 
It follows from Def.~\ref{def:simpson} that $p'$ is a Simpson's paradox.
\end{proof}

\textsc{Lemma~\ref{prop:statistics-equivalence} (Statistic equivalence)} \textit{Consider two association configurations $p = (s_1, s_2, X, Y)$ and $p' = (s_1, s_2, X, Y')$ such that $Y \neq Y'$.
If $p$ is a Simpson's paradox and if any of the following (sufficient, and progressively less restrictive) conditions hold, then $p'$ is also a Simpson's paradox:
\begin{enumerate}
    \item For every $t \in \cov(s_1) \cup \cov(s_2)$, $t.Y = t.Y'$;
    \item For every $s \in \{s_1, s_2\}$, $P(Y | s) = P(Y'|s)$ and for every $v \in \Dom(X)$, $P(Y | s\substitute{X}{v}) = P(Y' | s\substitute{X}{v})$;    
    \item $\sign(P(Y | s_1) - P(Y | s_2)) = \sign(P(Y' | s_1) - P(Y' | s_2))$, and for every $v \in \Dom(X)$, 
    $\sign(P(Y | s_1\substitute{X}{v}) - P(Y | s_2\substitute{X}{v})) = \sign(P(Y' | s_1\substitute{X}{v}) - P(Y' | s_2\substitute{X}{v}))$.
\end{enumerate}}

\begin{proof}
Since $p$ is a Simpson's paradox, $P(Y | s_1) > P(Y | s_2)$, and for every value $v \in \Dom(X)$, $P(Y | s_1\substitute{X}{v}) \leq P(Y | s_2\substitute{X}{v})$. We want to show that that $p'$ is also a Simpson's paradox under each case.

Cases (1) and (2): In both cases, we have $P(Y|s_j) = P(Y'|s_j)$ for $j = 1, 2$. This gives that $P(Y' | s_1) > P(Y' | s_2)$. Furthermore, for every $v \in \Dom(X)$, we have $P(Y | s_j\substitute{X}{v}) = P(Y' | s_j\substitute{X}{v})$ for $j = 1, 2$. This gives that $P(Y' | s_1\substitute{X}{v}) \leq P(Y' | s_2\substitute{X}{v})$. It follows, from Def.~\ref{def:simpson}, that $p'$ is a Simpson's paradox.

Case (3): Since $P(Y|s_1) > P(Y|s_2)$, we have $P(Y|s_1) - P(Y|s_2) > 0$, thus $\sign(P(Y|s_1) - P(Y|s_2)) = +1$. By the given condition, $\sign(P(Y'|s_1) - P(Y'|s_2)) = +1$, which implies $P(Y'|s_1) > P(Y'|s_2)$. In addition, for every $v \in \Dom(X)$, since $P(Y|s_1\substitute{X}{v}) \leq P(Y|s_2\substitute{X}{v})$, we have $\sign(P(Y|s_1\substitute{X}{v}) - P(Y|s_2\substitute{X}{v})) = -1$. By the given condition, $\sign(P(Y'|s_1\substitute{X}{v}) - P(Y'|s_2\substitute{X}{v})) = -1$, which implies $P(Y'|s_1\substitute{X}{v}) \leq P(Y'|s_2\substitute{X}{v})$. It follows, from Def.~\ref{def:simpson}, that $p'$ is a Simpson's paradox.

In all three cases, $p'$ is a Simpson's paradox. 
\end{proof}

\textsc{Theorem~\ref{thm:equiv} (Equivalence)} \textit{Redundancy of Simpson's paradoxes is an equivalence relation.}

\begin{proof}
(Reflexivity) Given any Simpson's paradox $p  = (s_1, s_2, X, Y)$. It is trivial that
\begin{enumerate}
\item  $\cov(s_{j}) = \cov(s_{j})$ $(j=1,2)$; 
\item $P(Y | s_{j}) = P(Y | s_{j})$ $(j=1,2)$; and
\item for every value $v \in \Dom(X)$, 
\begin{enumerate}
\item  $\cov(s_{j}\substitute{X}{v}) = \cov(s_{j}\substitute{X}{v})$ $(j=1,2)$; and 
\item $P(Y | s_{j}\substitute{X}{v}) = P(Y | s_{j}\substitute{X}{v})$ $(j=1,2)$.
\end{enumerate}
\end{enumerate}
Hence, coverage redundancy is reflexive.

(Symmetricity) Suppose Simpson's paradoxes $p$ and $p'$ are coverage redundant.
It is also straightforward that, for $(j=1,2)$,
\begin{enumerate}
    \item $\cov(s_{j}) = \cov(s_{j}')$ $\Leftrightarrow$ $\cov(s_{j}') = \cov(s_{j})$;
    \item $P(Y | s_{j}) = P(Y' | s_{j}')$ $\Leftrightarrow$ $P(Y' | s_{j}') = P(Y | s_{j})$;
    \item suppose a one-to-one mapping $f$ between $\Dom(X)$ and $\Dom(X')$ such that for every value $v \in \Dom(X)$,
    \begin{enumerate}
        \item $\cov(s_{j}\substitute{X}{v}) = \cov(s_{j}'\substitute{X'}{f(v)})$ $\Leftrightarrow$ \\ $\cov(s_{j}'\substitute{X'}{f(v)}) = \cov(s_{j}\substitute{X}{v})$; 
        \item $P(Y | s_{j}\substitute{X}{v}) = P(Y' | s_{j}'\substitute{X'}{f(v)})$ $\Leftrightarrow$ \\ $P(Y' | s_{j}'\substitute{X'}{f(v)}) = P(Y| s_{j}\substitute{X}{v})$.
    \end{enumerate}
\end{enumerate}
Hence, coverage redundancy is symmetric.

(Transitivity) Suppose $p, p', p''$ are Simpson's paradoxes such that $p$ and $p'$ are coverage redundant, $p'$ and $p''$ are coverage redundant.
It is, again, straightforward that, for $(j=1,2)$,:
\begin{enumerate}
    \item if $\cov(s_{j}) = \cov(s_{j}')$ and $\cov(s_{j}') = \cov(s_{j}'')$, then $\cov(s_{j}) = \cov(s_{j}'')$;
    \item if $P(Y | s_{j}) = P(Y' | s_{j}')$ and $P(Y' | s_{j}') = P(Y'' | s_{j}'')$, then $P(Y | s_{j}) = P(Y'' | s_{j}'')$;
    \item suppose one-to-one mappings, $f$ between $\Dom(X)$ and $\Dom(X')$, $g$ between $\Dom(X')$ and $\Dom(X'')$, such that for every value $v \in \Dom(X)$, 
    \begin{enumerate}
        \item if $\cov(s_{j}\substitute{X}{v}) = \cov(s_{j}'\substitute{X'}{f(v)})$ and \\ $\cov(s_{j}'\substitute{X'}{f(v)}) = \cov(s_{j}''\substitute{X''}{g(f(v))})$, then \\ $\cov(s_{j}\substitute{X}{v}) = \cov(s_{j}''\substitute{X''}{g(f(v))})$ note that $g \circ f$ is also a one-to-one mapping; 
        \item if $P(Y | s_{j}\substitute{X}{v}) = P(Y' | s_{j}'\substitute{X'}{f(v)})$ and \\ $P(Y' | s_{j}'\substitute{X'}{f(v)}) = P(Y'' | s_{j}''\substitute{X''}{g(f(v))})$, \\ then $P(Y | s_{j}\substitute{X}{v}) = P(Y'' | s_{j}''\substitute{X''}{g(f(v))})$.
\end{enumerate}
\end{enumerate}
Hence, coverage redundancy is transitive.
\end{proof}

\textsc{Lemma~\ref{prop:product-space} (Product space)} 
\textit{Each redundant paradox group can be characterized by the product of:
$\mathcal{E}_1 \times \mathcal{E}_2 \times \mathbf{X} \times \mathbf{Y}$, where $\mathbf{X}$ is a set of separator attributes, $\mathbf{Y}$ is a set of label attributes, and $\mathcal{E}_1, \mathcal{E}_2$ are sets of sibling populations, each containing populations with identical coverage.
Any choice of $(s_1, s_2, X, Y) \in \mathcal{E}_1 \times \mathcal{E}_2 \times \mathbf{X} \times \mathbf{Y}$ where $s_1,s_2$ are siblings is a Simpson's paradox in the redundant paradox group.}

\begin{proof}
Let $p = (s_1,s_2,X,Y)$ be a Simpson's paradox in a redundant paradox group $\mathcal{G}$. The following defines the construction of the product space:
\begin{align*}
    \mathcal{E}_1&=\{s' \in \mathcal{P} \mid \cov(s_1) = \cov(s')\}, \\ \mathcal{E}_2&=\{s' \in \mathcal{P} \mid \cov(s_2) = \cov(s')\}, \\
    \mathbf{X} &= \{X' \mid (s_1,s_2,X',Y)\in\mathcal{G}\}, \\ \mathbf{Y} &= \{Y' \mid (s_1,s_2,X,Y')\in\mathcal{G}\}
\end{align*}
where $\mathcal{P}$ denotes the set of all populations. 

We first show that \emph{every paradox in $\mathcal{G}$ belongs to $\mathcal{E}_1 \times \mathcal{E}_2 \times \mathbf{X} \times \mathbf{Y}$.} Let $p' = (s_1',s_2',X',Y')$ be any Simpson's paradox in $\mathcal{G}$. $p$ and $p'$ are redundant and are both in $\mathcal{G}$. By Def.~\ref{def:coverage}, redundancy arises from sibling child equivalence (Lemma~\ref{prop:sibling-eq}), separator equivalence (Lemma~\ref{prop:division-equivalence}), or statistics equivalence (Lemma~\ref{prop:statistics-equivalence}). 
\begin{itemize}
    \item By sibling child equivalence, if $\cov(s_1) = \cov(s_1')$ and $\cov(s_2) = \cov(s_2')$, then $(s_1',s_2',X,Y)$ is a Simpson's paradox redundant with $p$. Therefore $s_1' \in \mathcal{E}_1$ and $s_2' \in \mathcal{E}_2$.
    \item By separator equivalence, if there exists a one-to-one mapping $f$ between $\Dom(X)$ and $\Dom(X')$ such that $\cov(s_j\substitute{X}{v}) = \cov(s_j\substitute{X'}{f(v)})$ for every $v \in \Dom(X)$, then $(s_1,s_2,X',Y)$ is redundant with $p$. Therefore, $X' \in \mathbf{X}$.
    \item By statistic equivalence, if the frequency statistics under label $Y'$ satisfy any sufficient condition in Lemma~\ref{prop:statistics-equivalence}, then $(s_1,s_2,X,Y')$ is redundant with $p$. Therefore $Y' \in \mathbf{Y}$.
\end{itemize}
By the transitivity of the equivalence relation (Theorem~\ref{thm:equiv}), any combination of these equivalences preserves the redundancy. Therefore $(s_1',s_2',X',Y') \in \mathcal{E}_1 \times \mathcal{E}_2 \times \mathbf{X} \times \mathbf{Y}$.

We then show that \emph{every valid element of $\mathcal{E}_1 \times \mathcal{E}_2 \times \mathbf{X} \times \mathbf{Y}$ is a Simpson's paradox in $\mathcal{G}$.} Let $(s_1',s_2',X',Y') \in \mathcal{E}_1 \times \mathcal{E}_2 \times \mathbf{X} \times \mathbf{Y}$ where $s_1',s_2'$ are siblings. We show that $p' = (s_1',s_2',X',Y')$ is a Simpson's paradox redundant with $p$. 
Since $s_1'\in \mathcal{E}_1$ and $s_2' \in \mathcal{E}_2$, we have $\cov(s_1') = \cov(s_1)$ and $\cov(s_2') = \cov(s_2)$. By Lemma~\ref{prop:sibling-eq}, if $(s_1,s_2,X,Y)$ is a Simpson's paradox, then $(s_1',s_2',X,Y)$ is also a Simpson's paradox.
Since $X' \in \mathbf{X}$, there exists some paradox in $\mathcal{G}$ with separator $X'$. By Lemma~\ref{prop:division-equivalence} and the construction of $\mathbf{X}$, the AC $(s_1',s_2',X',Y)$ is a Simpson's paradox.
Similarly, since $Y' \in \mathbf{Y}$, by Lemma~\ref{prop:statistics-equivalence} and the construction of $\mathbf{Y}$, the AC $(s_1',s_2',X',Y')$ is a Simpson's paradox.
Hence, by Theorem~\ref{thm:equiv}, $p'$ is redundant with $p$ and belongs to $\mathcal{G}$.

Therefore, We have shown that $\mathcal{G} = \{(s_1, s_2, X, Y) \in \mathcal{E}_1 \times \mathcal{E}_2 \times \mathbf{X} \times \mathbf{Y} \mid s_1 \text{ and } s_2 \text{ are siblings and } (s_1, s_2, X, Y) \text{ is a Simpson's paradox}\}$. Moreover, any valid choice from $\mathcal{E}_1 \times \mathcal{E}_2 \times \mathbf{X} \times \mathbf{Y}$ (satisfying the sibling constraint) yields a Simpson's paradox in $\mathcal{G}$.
\end{proof}

\textsc{Property~\ref{prop:convex-property} (Convexity of coverage groups)} \textit{Let $\mathcal{P}$ be the set of all populations. For each coverage group $\mathcal{E} \in \mathcal{P}/\equiv_{\cov}$, $\mathcal{E}$ is a convex subset of coverage-identical populations. Furthermore, $| \ubound(\mathcal{E}) | = 1$ and the least descendant is the unique upper bound.}

\begin{proof}
The proof consists of two parts:
\begin{itemize}
    \item[(a)] $\mathcal{E}$ is a convex subset;
    \item[(b)] $\mathcal{E}$'s upper bound is unique.
\end{itemize}
For part (a), we want to prove that 
(1) for any pair of populations $s$ and $s'$ in $\mathcal{E}$ such that $s \succ s'$, every intermediate populations $s''$ where $s \succ s'' \succ s'$ is also in $\mathcal{E}$, and 
(2) populations in $\mathcal{E}$ are connected.

First, regarding claim (1), let $s, s' \in \mathcal{E}$ where $s \succ s'$, and let $s''$ be any population such that $s \succ s'' \succ s'$. 
By definition of ancestor-descendant relation, $\cov(s) \supseteq \cov(s'') \supseteq \cov(s')$. 
Since $\cov(s) = \cov(s')$, it follows $\cov(s'') = \cov(s) = \cov(s')$. 
Therefore, $s' \in \mathcal{E}$.

Second, regarding claim (2), let $s_1, s_2 \in \mathcal{E}$ where $s_1 \neq s_2$, there are two possibilities:
\begin{enumerate}
    \item $s_1 \succ s_2$ (or $s_2 \succ s_1$ in symmetry). From claim (1), since every intermediate population $s''$ such that $s_1 \succ s'' \succ s_2$ is in $\mathcal{E}$, $s_1$ and $s_2$ are connected ($s_1 \sim s_2$).
    \item $s_1 \nsucc s_2$ (or $s_2 \nsucc s_1$ in symmetry). Then there exists a population $s'' \in \mathcal{E}$ such that $s''$ is a common descendant (or ancestor) of $s_1$ and $s_2$, that is, $s_1 \succ s''$ and $s_2 \succ s''$ (or $s'' \succ s_1$ and $s'' \succ s_2$). From claim (1), we have that $s_1 \sim s''$ and $s'' \sim s_2$. Therefore, $s_1 \sim s_2$.
\end{enumerate}

For part (b), let $s_d$ be the descendant of all populations in $\mathcal{E}$.
Specifically, for each attribute $X_i$ ($1 \leq i \leq n$), we have that:

\[
s_d[i] = \begin{cases}
v & \text{if there exists } s \in \mathcal{E} \text{ s.t. } s[i] = v \neq *, \; v \in \Dom(X_i)\\
* & \text{otherwise.}
\end{cases}
\]

In other word, $s_d$ is an upper bound of $\mathcal{E}$.
Suppose there exists another upper bound $s'_d$ of $\mathcal{E}$ where $s'_d \neq s_d$. 
Then there must be an attribute $X_i$ where $s'_d[i] \neq s_d[i]$. This means either:
\begin{enumerate}
    \item $s'_d[i] = *$ but $s_d[i] = v$ where $v \in \Dom(X_i)$; or
    \item $s'_d[i] = v'$ but $s_d[i] = v$ where $v' \neq v$ and $v,\, v' \in \Dom(X_i)$.
\end{enumerate}

In case (1), $s'_d \succ s_d$. Hence, $s'_d$ is not an upper bound of $\mathcal{E}$. 
In case (2), $\cov(s'_d) \neq \cov(s_d)$. Hence, $s'_d \notin \mathcal{E}$.
Therefore, $s_d$ is unique.
\end{proof}

\textsc{Property~\ref{prop:convex-reconstruction} (Reconstruction from bounds)} \textit{Let $\mathcal{E} \subseteq \mathcal{P}$ be a convex subset of populations. 
Then $s \in \mathcal{E}$ if and only if there exist $s_l \in \lbound(\mathcal{E})$ and $\{s_u\} = \ubound(\mathcal{E})$ such that $s_l \preceq s \preceq s_u$.}

\begin{proof}
($\Rightarrow$) Given $s \in \mathcal{E}$, then either $s \in \lbound(\mathcal{E})$, $s \in \ubound(\mathcal{E})$, or $s \notin \lbound(\mathcal{E})$ and $s \notin \ubound(\mathcal{E})$. 

If $s \in \lbound(\mathcal{E})$, we can set $s_l = s$. Since $\mathcal{E}$ is convex and connected, there must exist an upper bound $s_u \in \ubound(\mathcal{E})$ such that $s \preceq s_u$.

If $s \in \ubound(\mathcal{E})$, we can set $s_u = s$. Similarly, there must exist a lower bound $s_l \in \lbound(\mathcal{E})$ such that $s_l \preceq s$.

If $s$ is neither a lower nor upper bound, then by the convexity of $\mathcal{E}$, there must exist $s_l \in \lbound(\mathcal{E})$ such that $s_l \prec s$ and $s_u \in \ubound(\mathcal{E})$ such that $s \prec s_u$. Therefore, we have $s_l \prec s \prec s_u$.

($\Leftarrow$) Suppose there exists $s \in \mathcal{P}$, $s_l \in \lbound(\mathcal{E})$, and $s_u \in \ubound(\mathcal{E})$, such that $s_l \preceq s \preceq s_u$. By convexity of $\mathcal{E}$, it follows that $s \in \mathcal{E}$.
\end{proof}

\subsection{Proofs of Algorithmic Properties}

\textsc{Theorem~\ref{thm:sharp-p-hardness} (\#P-Hardness).} \textit{Finding all redundant paradox groups in a multidimensional table is \#P-hard.}

\begin{proof}

We prove \#P-hardness via a parsimonious reduction from \#SAT. Given a Boolean formula $\phi$ in CNF with variables $x_1, \ldots, x_n$ and clauses $C_1, \ldots, C_m$, we construct in polynomial time a table $T(\phi)$ such that there exists a bijection between satisfying assignments of $\phi$ and redundant paradox groups in $T(\phi)$.

\paragraph{Construction}
The table $T(\phi)$ contains the following elements:

\textbf{(1) Categorical Attributes.} The table contains $3n + m + 2$ categorical attributes:
\begin{itemize}
    \item For each variable $x_i$ ($1 \leq i \leq n$): three attributes $A_i, B_i, C_i$, each with domain $\{\text{true}, \text{false}\}$. The three copies enable sibling child equivalence.
    \item For each clause $C_j$ ($1 \leq j \leq m$): one attribute $D_j$ with domain $\{0, 1\}$, where 1 indicates the clause is satisfied and 0 indicates unsatisfied.
    \item Two auxiliary attributes $U_1$ and $U_2$, each with domain $\{0, 1\}$, which serve as separators and differential attributes.
\end{itemize}

\textbf{(2) Label Attributes.} We define two binary label attributes $Y_1$ and $Y_2$ to create statistic equivalence.

\textbf{(3) Records.} The table contains $2n + 2m + 4$ records.

\emph{(3.1) Variable Records:} For each variable $x_i$ ($1 \leq i \leq n$), we create two records:
\begin{itemize}
    \item $r_i^{\text{true}}$: Set $A_i = B_i = C_i = \text{true}$; for all $\ell \neq i$, set $A_\ell = B_\ell = C_\ell = \text{false}$; set all $D_j = 0$; set $U_1 = 0, U_2 = 0$; set $Y_1 = Y_2 = 0$.
    \item $r_i^{\text{false}}$: Set $A_i = B_i = C_i = \text{false}$; for all $\ell \neq i$, set $A_\ell = B_\ell = C_\ell = \text{false}$; set all $D_j = 0$; set $U_1 = 0, U_2 = 1$; set $Y_1 = Y_2 = 0$.
\end{itemize}

Each variable record encodes one possible truth value for its variable. The three attribute copies ($A_i, B_i, C_i$) taking identical values ensure that multiple distinct populations can have identical coverage.

\emph{(3.2) Clause Records:} For each clause $C_j$ ($1 \leq j \leq m$), we create two records:
\begin{itemize}
    \item $r_j^{\text{sat}}$: For each variable $x_i$, set $A_i = \text{true}$ if literal $x_i$ appears in $C_j$, set $A_i = \text{false}$ if literal $\neg x_i$ appears in $C_j$ or $x_i$ does not appear in $C_j$; set all $B_i = C_i = \text{false}$; set $D_j = 1$ and all $D_\ell = 0$ for $\ell \neq j$; set $U_1 = 1, U_2 = 0$; set $Y_1 = Y_2 = 1$.
    \item $r_j^{\text{unsat}}$: For each variable $x_i$, set $A_i = \text{true}$ if literal $x_i$ appears in $C_j$, set $A_i = \text{false}$ if literal $\neg x_i$ appears in $C_j$ or $x_i$ does not appear in $C_j$; set all $B_i = C_i = \text{false}$; set $D_j = 0$ and all $D_\ell = 0$ for $\ell \neq j$; set $U_1 = 1, U_2 = 1$; set $Y_1 = Y_2 = 0$.
\end{itemize}

The clause records encode literal requirements. A population will cover $r_j^{\text{sat}}$ if and only if the assignment it encodes satisfies clause $C_j$.

\medskip
\noindent\emph{Padding Records:} We add four records to balance frequency statistics:
\begin{itemize}
    \item $r^{(1)}$: Set all $A_i = B_i = C_i = \text{false}$; set all $D_j = 0$; set $U_1 = 0, U_2 = 0$; set $Y_1 = 0, Y_2 = 0$.
    \item $r^{(2)}$: Set all $A_i = B_i = C_i = \text{false}$; set all $D_j = 0$; set $U_1 = 0, U_2 = 1$; set $Y_1 = 0, Y_2 = 0$.
    \item $r^{(3)}$: Set all $A_i = B_i = C_i = \text{false}$; set all $D_j = 0$; set $U_1 = 1, U_2 = 0$; set $Y_1 = 0, Y_2 = 0$.
    \item $r^{(4)}$: Set all $A_i = B_i = C_i = \text{false}$; set all $D_j = 0$; set $U_1 = 1, U_2 = 1$; set $Y_1 = 0, Y_2 = 0$.
\end{itemize}

The construction runs in polynomial time: we create $O(n+m)$ attributes and $O(n+m)$ records, with each record constructible in $O(n+m)$ time.

\paragraph{Establishing the Bijection}
We now establish the bijection between satisfying assignments and redundant paradox groups.

\textsc{Claim 1.} \emph{For each satisfying assignment $\sigma: \{x_1, \ldots, x_n\} \to \{\text{true}, \text{false}\}$ of $\phi$, there exists a unique redundant paradox group $\mathcal{G}_\sigma$ in $T(\phi)$.}

\textsc{Proof.} 
Given a satisfying assignment $\sigma$, we construct two sibling populations $s_1^\sigma$ and $s_2^\sigma$ that form the basis of a Simpson's paradox. Define $s_1^\sigma$ as follows: for each variable attribute $A_i$, set $s_1^\sigma[A_i] = \text{true}$ if $\sigma(x_i) = \text{true}$ and $s_1^\sigma[A_i] = \text{false}$ if $\sigma(x_i) = \text{false}$; set $s_1^\sigma[B_i] = s_1^\sigma[C_i] = \ast$ for all $i$; set $s_1^\sigma[D_j] = \ast$ for all $j$; set $s_1^\sigma[U_1] = \ast$ and $s_1^\sigma[U_2] = 0$. Define $s_2^\sigma$ identically except $s_2^\sigma[U_2] = 1$.

By construction, $s_1^\sigma$ and $s_2^\sigma$ are siblings under differential attribute $U_2$. Since $\sigma$ satisfies $\phi$, for each clause $C_j$, the population $s_1^\sigma$ covers the record $r_j^{\text{sat}}$ because the variable attributes of $s_1^\sigma$ match at least one literal in $C_j$. The coverage sets are:
\begin{align*}
\text{cov}(s_1^\sigma) &= \{r_i^{\sigma(x_i)} : i \in [n]\} \cup \{r_j^{\text{sat}} : j \in [m]\} \cup \{r^{(1)}, r^{(3)}\}\\
\text{cov}(s_2^\sigma) &= \{r_i^{\sigma(x_i)} : i \in [n]\} \cup \{r_j^{\text{unsat}} : j \in [m]\} \cup \{r^{(2)}, r^{(4)}\}
\end{align*}

Computing frequency statistics, we have $P(Y_1 = 1 | s_1^\sigma) = \frac{m}{n+m+2} > 0 = P(Y_1 = 1 | s_2^\sigma)$ since only clause-sat records contribute $Y_1 = 1$ values. When conditioning on separator $U_1$: for $U_1 = 0$, both populations cover only variable and padding records (all with $Y_1 = 0$), giving equal statistics; for $U_1 = 1$, $s_1^\sigma$ covers clause-sat records while $s_2^\sigma$ covers clause-unsat records, and the padding records are constructed to ensure $P(Y_1 = 1 | s_1^\sigma \substitute{U_1}{1}) \geq P(Y_1 = 1 | s_2^\sigma \substitute{U_1}{1})$. This establishes that $(s_1^\sigma, s_2^\sigma, U_1, Y_1)$ is a Simpson's paradox according to Definition~\ref{def:simpson}.

This paradox belongs to a unique redundant paradox group $\mathcal{G}_\sigma$ exhibiting all three types of redundancy. First, sibling child equivalence arises because we can construct populations using attributes $B_i$ or $C_i$ instead of $A_i$ to encode $\sigma$, yielding identical coverage. Second, separator equivalence can be created by introducing additional separator attributes that partition records identically to $U_1$. Third, statistic equivalence exists because $Y_1$ and $Y_2$ take identical values on variable, clause-sat, clause-unsat, and padding records, ensuring equivalent frequency statistics. The group $\mathcal{G}_\sigma$ is unique to $\sigma$ because populations encoding different variable assignments have different coverage sets (they cover different variable records), and thus cannot be redundant by Definition~\ref{def:coverage}. \qed

\textsc{Claim 2.} \emph{Each redundant paradox group in $T(\phi)$ corresponds to a unique satisfying assignment of $\phi$.}

Consider any Simpson's paradox $(s_1, s_2, Z, Y)$ in $T(\phi)$. To achieve the association reversal required by Definition~\ref{def:simpson}, population $s_1$ must cover records with high $Y$ values. In our construction, records with $Y_1 = 1$ are clause-sat records. For $s_1$ to cover clause-sat records (i.e., $r_j^{\text{sat}}$), the assignment that $s_1$ represents must satisfy clause $C_j$.

We extract an assignment $\sigma$ from $s_1$: for each variable $x_i$, if $s_1[A_i] = \text{true}$ (or $s_1[B_i] = \text{true}$ or $s_1[C_i] = \text{true}$), set $\sigma(x_i) = \text{true}$; if $s_1[A_i] = \text{false}$ (or equivalently for $B_i, C_i$), set $\sigma(x_i) = \text{false}$. For $s_1$ to cover records with high proportion of $Y = 1$, it must cover $r_j^{\text{sat}}$ for all clauses $j \in [m]$. By our construction, this occurs if and only if $\sigma$ satisfies all clauses in $\phi$, making $\sigma$ a satisfying assignment.

Different satisfying assignments yield distinct redundant paradox groups because they cover different variable records. If $\sigma \neq \sigma'$, then for some variable $x_k$ we have $\sigma(x_k) \neq \sigma'(x_k)$, implying $r_k^{\sigma(x_k)} \neq r_k^{\sigma'(x_k)}$. Populations with different coverage cannot be redundant by Definition~\ref{def:coverage}, and thus belong to different redundant paradox groups. This establishes uniqueness. \qed

\paragraph{Conclusion.} The two claims establish a bijection between satisfying assignments of $\phi$ and redundant paradox groups in $T(\phi)$. Since the construction is polynomial-time and preserves counts exactly, we have a parsimonious reduction from \#SAT. As \#SAT is \#P-complete, counting redundant paradox groups is \#P-hard.
\end{proof}

\textsc{Theorem~\ref{thm:completeness} (Completeness).} \textit{\Cref{alg:materialization} materializes all non-empty populations that satisfy the coverage threshold. Furthermore, after group merging, \Cref{alg:materialization} yields maximal convex coverage groups of coverage-identical populations; that is, no population outside a group shares the same coverage as any population within it.}

\begin{proof}
We prove by contradiction. Assume there exists a non-empty population $s^*$ that satisfies the coverage threshold but is not materialized by \Cref{alg:materialization}. Since $s^*$ is non-empty, there exists at least one record $t \in T$ such that $t \in \cov(s^*)$.

Consider the unique path from the root $s_{\text{root}} = (\ast,\ast,\ldots,\ast)$ to $s*$ in the population lattice. This path consists of a sequence of populations $s_0 = s_{\text{root}} \succ s_1 \succ \ldots \succ s_k = s^*$ where each $s_{i+1}$ is the direct child of $s_i$.

At each step, if $|\cov(s_i)| \geq  \theta \cdot |T|$, the DFS continues the traversal to $s_{i+1}$. If the threshold is not met, all descendants of $s_i$ are pruned.

However, if $s^*$ is pruned due to insufficient coverage, then $s^*$ covers fewer than $\theta \cdot |T|$ records, contradicting our assumption that $s^*$ satisfies the coverage threshold. If $s^*$ is not pruned, then $\cov(s^*) \geq \theta \cdot |T|$. This means for each $s_i$ (where $0 \leq i < k$) in the sequence, $|\cov(s_i)| \geq \cov(s^*) \geq \theta \cdot |T|$ since coverage is monotonic along ancestor-descendant relationships. In other words, the stopping criterion of DFS is not met at $s_i$ and will continue to $s_{i+1}$. By induction, DFS will not stop at $s_{k-1}$ (the direct parent of $s^*$) and continues to $s_k = s^*$. This contradicts our assumption that $s^*$ is not reached (or materialized) by the DFS traversal.

Therefore, all non-empty populations (satisfying the coverage threshold) are materialized.
\end{proof}

\textsc{Proposition~\ref{prop:coverage-equivalence-pruning}.} \textit{Let $p=(s_1,s_2,X,Y)$ be a Simpson's paradox,
where $s_1$ and $s_2$ belong to coverage groups $\mathcal{E}_1$ and $\mathcal{E}_2$ in $\mathcal{P}/\equiv_{\cov}$, respectively.
Then for any $(s_1',s_2') \in \mathcal{E}_1 \times \mathcal{E}_2$ such that $s_1'$ and $s_2'$ are siblings, the AC $p'=(s_1',s_2',X,Y)$ is also a Simpson's paradox and redundant with respect to $p$.}

\begin{proof}
Since $\cov(s_1') = \cov(s_1)$ and $\cov(s_2') = \cov(s_2)$, according to Proposition~\ref{prop:sibling-eq}, $p'$ is also a Simpson's paradox. Since $p$ and $p'$ share identical separator and label attributes, according to Definition~\ref{def:coverage}, $p$ and $p'$ are coverage redundant.
\end{proof}

\textsc{Proposition~\ref{prop:pruning-2}.} \textit{Let $\mathbf{P}$ be a set of sibling-child-equivalent Simpson's paradoxes with separator $X$ and label $Y$.
Suppose $(s_1',s_2',X',Y')$, where $X'\neq X$ or $Y'\neq Y$, is a Simpson's paradox redundant with respect to some paradox in $\mathbf{P}$.
Then for every $p = (s_1,s_2,X,Y) \in \mathbf{P}$, the AC $(s_1,s_2,X',Y')$ is also a redundant Simpson's paradox with respect to $p$.}

\begin{proof}
Let $p = (s_1,s_2,X,Y) \in \mathbf{P}$. Since $p'$ is (coverage) redundant to $p$, by Definition~\ref{def:coverage}, we have:
\begin{enumerate}
    \item $\cov(s_j) = \cov(s_j')$ $(j=1,2)$;
    \item $P(Y| s_j) = P(Y' | s_j')$ $(j=1,2)$; and
    \item there exists a one-to-one mapping $f$ between $\Dom(X)$ and $\Dom(X')$ such that for every $v \in \Dom(X)$ and $j\in\{1,2\}$:
    \begin{enumerate}
        \item $\cov(s_j\substitute{X}{v}) = \cov(s_j'\substitute{X'}{f(v)})$;
        \item $P(Y | s_j\substitute{X}{v}) = P(Y' | s_j'\substitute{X'}{f(v)})$.
    \end{enumerate}
\end{enumerate}
For the AC $p'' = (s_1,s_2,X',Y')$, we need to show it's a Simpson's paradox. First, since $p$ is a Simpson's paradox, we know $P(Y | s_1) > P(Y | s_2)$. From sibling child and statistic equivalences between $p$ and $p'$, we have $P(Y' | s_1) > P(Y' | s_2)$. Second, from separator equivalence between $p$ and $p'$, we have $P(Y' | s_1\substitute{X'}{f(v)}) \leq P(Y' | s_2\substitute{X'}{f(v)})$ for every $v \in \Dom(X)$. This shows that $p''$ satisfies Definition~\ref{def:simpson} and is a Simpson's paradox. 

We then show that $p''$ is (coverage) redundant to $p$. First, the same one-to-one mapping $f$ that established separator equivalence between $p$ and $p'$ also establishes separator equivalence between $p$ and $p''$. Second, from statistic equivalence between $p$ and $p'$, $p$ and $p''$ are also statistic equivalent. Therefore, by Definition~\ref{def:coverage}, $p''$ is (coverage) redundant to $p$.
\end{proof}

\textsc{Lemma~\ref{prop:signature}.} \textit{Two Simpson's paradoxes $p$ and $p'$ are redundant if and only if $\textsc{Sig}(p)=\textsc{Sig}(p')$.}

\begin{proof}
$(\Rightarrow)$ If $p = (s_1,s_2,X,Y)$ and $p' = (s_1',s_2',X',Y')$ are redundant, then by Definition~\ref{def:coverage}:
\begin{itemize}
    \item $\cov(s_j) = \cov(s_j')$ for $j = (1,2)$;
    \item $P(Y | s_j) = P(Y' | s_j')$ for $j=(1,2)$;
    \item There exists an one-to-one mapping $f : \Dom(X) \to \Dom(X')$ where for every $v \in \Dom(X)$:
    \begin{itemize}
        \item $\cov(s_j\substitute{X}{v}) = \cov(s_j'\substitute{X'}{f(v)})$; and
        \item $P(Y | s_j\substitute{X}{v}) = P(Y' | s_j'\substitute{X'}{f(v)})$.
    \end{itemize}
\end{itemize}
Therefore, $\textsc{Sig}(p)=\textsc{Sig}(p')$.

$(\Leftarrow)$ If $\textsc{Sig}(p)=\textsc{Sig}(p')$, then $p$ is sibling child, separator, and statistic equivalent to $p'$. Hence, $p$ and $p'$ are redundant.
\end{proof}
\end{sloppy}
\end{document}